\documentclass[manuscript,nonacm]{acmart}

\setcopyright{none}
\copyrightyear{2021}
\acmYear{2021}
\acmDOI{}
\acmConference[]{}{}{}
\acmBooktitle{}
\acmPrice{}
\acmISBN{}

\usepackage{colortbl}
\usepackage{braket}
\usepackage{graphbox}
\usepackage[ruled]{algorithm2e}

\newcommand{\figlayout}{\raggedright\small\sffamily}
\newcommand{\figscale}{0.8}

\newcommand{\dd}{\delta}
\newcommand{\N}{\mathbb{N}}
\newcommand{\A}{\mathcal{A}}
\newcommand{\CC}{\mathcal{C}}
\newcommand{\CB}{\mathcal{CB}}

\newcommand{\M}{\mathcal{M}}
\newcommand{\TT}{\mathcal{T}}

\newcommand{\SX}{\mathcal{S}}
\newcommand{\LCL}{\mathsf{LCL}}
\newcommand{\LOCAL}{\mathsf{LOCAL}}
\newcommand{\CONGEST}{\mathsf{CONGEST}}

\DeclareMathOperator{\poly}{poly}

\newcommand{\LLL}{\mathsf{LLL}}

\newcommand{\PathPi}{\Pi^{\mathrm{path}}}
\newcommand{\PfPi}{\Pi_{\mathrm{pf}}}
\DeclareMathOperator{\certificateBuilder}{certificateBuilder}
\DeclareMathOperator{\findUnrestrictedCertificate}{findUnrestrictedCertificate}
\DeclareMathOperator{\removePathInflexibleConfigurations}{removePathInflexibleConfigurations}
\DeclareMathOperator{\findLogCertificate}{findLogCertificate}
\DeclareMathOperator{\constantCertificate}{constantCertificate}

\newcommand{\aab}{(a : b_1,\dots,a,\dots,b_\dd)}

\usepackage{thmtools}

\SetKwComment{comm}{$\rhd$ }{}
\SetCommentSty{textup}

\declaretheorem[style=plain,numberwithin=section]{theorem}
\declaretheorem[style=plain,sibling=theorem]{lemma}

\declaretheorem[style=plain,sibling=theorem]{corollary}

\declaretheorem[style=definition,sibling=theorem]{definition}

\DeclareMathOperator{\RCP}{RCP}
\DeclareMathOperator{\leaves}{leaves}
\DeclareMathOperator{\longpaths}{long-path-nodes}
\DeclareMathOperator{\flexibility}{flexibility}

\title{Locally Checkable Problems in Rooted Trees}

\author{Alkida Balliu}
\email{alkida.balliu@gssi.it}
\affiliation{%
\institution{Gran Sasso Science Institute}
\city{L’Aquila}
\country{Italy}
}

\author{Sebastian Brandt}
\email{brandt@cispa.de}
\affiliation{%
  \institution{CISPA Helmholtz Center for Information Security}
  \city{Saarbr\"ucken}
  \country{Germany}
}

\author{Yi-Jun Chang}
\email{cyijun@nus.edu.sg}
\affiliation{%
  \institution{National University of Singapore}
  \city{}
  \country{Singapore}
}

\author{Dennis Olivetti}
\email{dennis.olivetti@gssi.it}
\affiliation{%
\institution{Gran Sasso Science Institute}
\city{L’Aquila}
\country{Italy}
}

\author{Jan Studen\'y}
\email{jan.studeny@aalto.fi}
\affiliation{%
  \institution{Aalto University}
  \city{Espoo}
  \country{Finland}
}

\author{Jukka Suomela}
\email{jukka.suomela@aalto.fi}
\affiliation{%
  \institution{Aalto University}
  \city{Espoo}
  \country{Finland}
}

\author{Aleksandr Tereshchenko}
\email{aleksandr.tereshchenko@aalto.fi}
\affiliation{%
  \institution{Aalto University}
  \city{Espoo}
  \country{Finland}
}

\begin{document}

\begin{abstract}
Consider any locally checkable labeling problem $\Pi$ in \emph{rooted regular trees}: there is a finite set of labels $\Sigma$, and for each label $x \in \Sigma$ we specify what are permitted label combinations of the children for an internal node of label $x$ (the leaf nodes are unconstrained). This formalism is expressive enough to capture many classic problems studied in distributed computing, including vertex coloring, edge coloring, and maximal independent set.

We show that the distributed computational complexity of any such problem $\Pi$ falls in one of the following classes: it is $O(1)$, $\Theta(\log^* n)$, $\Theta(\log n)$, or $n^{\Theta(1)}$ rounds in trees with $n$ nodes (and all of these classes are nonempty). We show that the complexity of any given problem is the same in all four standard models of distributed graph algorithms: deterministic $\LOCAL$, randomized $\LOCAL$, deterministic $\CONGEST$, and randomized $\CONGEST$ model. In particular, we show that randomness does not help in this setting, and the complexity class $\Theta(\log \log n)$ does not exist (while it does exist in the broader setting of general trees).

We also show how to systematically determine the complexity class of any such problem~$\Pi$, i.e., whether $\Pi$ takes $O(1)$, $\Theta(\log^* n)$, $\Theta(\log n)$, or $n^{\Theta(1)}$ rounds. While the algorithm may take exponential time in the size of the description of $\Pi$, it is nevertheless practical: we provide a freely available implementation of the classifier algorithm, and it is fast enough to classify many problems of interest.
\end{abstract}

\maketitle

\section{Introduction}

We aim at \emph{systematizing} and \emph{automating} the study of computational complexity in the field of distributed graph algorithms. Many key problems of interest in the field are \emph{locally checkable}. While it is known that questions related to the distributed computational complexity of locally checkable problems are \emph{undecidable} in general graphs~\cite{NaorStockmeyer95, Brandt2017}, there is no known obstacle that would prevent one from completely automating the study of locally checkable problems in \emph{trees}. Achieving this is one of the major open problems in the field: currently only parts of the complexity landscape are known to be decidable~\cite{CP19timeHierarchy}, and the general decidability results are primarily of theoretical interest; \emph{practical} automatic techniques are only known for specific families of problems~\cite{Brandt2017,binary_lcls, lcls_on_paths_and_cycles}.

In this work we show that the study of locally checkable graph problems can be completely automated in \emph{regular rooted trees}. We not only give a full classification of the distributed complexity of any such problem (in all the usual models of distributed computing: deterministic and randomized $\LOCAL$ and $\CONGEST$), but we also present an algorithm that can automatically determine the complexity class of any given problem (with one caveat: our algorithm determines if the complexity is $n^{\Theta(1)}$, but not the precise exponent in this case). Even though the algorithm takes in the worst case exponential time in the size of the problem description, it is nevertheless practical: we have implemented it for the case of binary trees, and it is in practice very fast, classifying e.g.\ the sample problems that we present here in a matter of milliseconds~\cite{AnonymousRepo}.

\subsection{Setting}

In this work we study locally checkable problems defined in regular, unlabeled, not necessarily balanced, rooted trees of bounded degree. For our purposes, such a problem $\Pi$ is specified as a triple $(\delta, \Sigma, C)$, where $\delta \in \N$ is the number of children for the internal nodes, $\Sigma$ is a finite set of \emph{labels}, and $C$ is the set of permitted \emph{configurations}. Each configuration looks like $x : y_1 y_2 \dotsb y_\delta$, indicating that if the label of an internal node is $x$, then one of the possible labelings for its $\delta$ children is $y_1, y_2, \dotsc, y_\delta$, in some order (that is, the order of the children does not matter). The leaf nodes are unconstrained.

The reason why we choose this specific setting is the following. As soon as we consider \emph{inputs}, it is known that decidability questions become much harder \cite{balliu19lcl-decidability,Chang20}, and since even the case with no inputs is still not understood, we try to understand this setting first. Moreover, it is possible to use non-regular trees to encode trees with inputs, and for this reason we constrain only nodes with exactly $\delta$ children.

\subsection{Example: 3-coloring}\label{ssec:3col-example}

Consider the problem of $3$-coloring binary trees, i.e., trees in which internal nodes have $\delta = 2$ children. The possible labels of the nodes are $\Sigma = \{1,2,3\}$. The color of a node has to be different from the colors of any of its children; hence we can write down the set of configurations e.g.\ as follows:
\begin{equation}\label{eq:3col-example}
\begin{split}
    C = \bigl\{
    &1 : 22,\,
    1 : 23,\,
    1 : 33,\\[-1mm]
    &2 : 11,\,
    2 : 13,\,
    2 : 33,\\[-1mm]
    &3 : 11,\,
    3 : 12,\,
    3 : 22
    \bigr\}.
\end{split}
\end{equation}
We emphasize that the ordering of the children is irrelevant here; hence $1:23$ and $1:32$ are the same configuration. It is easy to verify that this is a straightforward correct encoding of the $3$-coloring problem in binary trees.

It is well-known that this problem can be solved in the $\LOCAL$ model of distributed computing in $O(\log^* n)$ rounds in rooted trees \cite[Section 3.4]{barenboim13distributed}, using the technique by \citet{ColeVishkin86}, and this is also known to be tight, both for deterministic and randomized algorithms \cite{Naor1991,Linial92}.

One can also in a similar way define the problem of $2$-coloring binary trees; it is easy to check that this is a global problem, with complexity $\Theta(n)$ rounds:
\begin{equation}\label{eq:2col-example}
    C = \bigl\{
    1 : 22,\,
    2 : 11
    \bigr\}.
\end{equation}

\subsection{Example: maximal independent set}\label{ssec:mis-example}

Let us now look at a bit more interesting problem: maximal independent sets (MIS). Let us again stick to binary trees, i.e., $\delta = 2$ children. The first natural idea for encoding MIS as a locally checkable problem would be to try to use only two labels, $0$ and $1$, with $1$ indicating that a node is in the independent set, but this is not sufficient to express both the notion of independence and the notion of maximality. However, three labels will be sufficient to correctly capture the problem. We set $\Sigma = \{1,a,b\}$, with $1$ indicating that a node is in the independent set, and choose the following configurations:
\begin{equation}\label{eq:mis-example}
    C = \bigl\{
    1 : aa,\,
    1 : ab,\,
    1 : bb,\,
    a : bb,\,
    b : b1,\,
    b : 11
    \bigr\}.
\end{equation}
Now it takes a bit more effort to convince oneself that this indeed correctly captures the idea of maximal independent sets. The key observations are these: a node with label $1$ cannot be adjacent to another node with label $1$, a node with label $a$ has to have $1$ above it, and a node with label $b$ has to have $1$ below it, so nodes with label $1$ clearly form a maximal independent set. Conversely, given any maximal independent set $X$ we can find a corresponding label assignment if we first assign labels $1$ to nodes in $X$, then assign labels $b$ to the parents of the nodes in $X$, and finally label the remaining nodes with label $a$. The only minor technicality is that this labeling corresponds to an MIS only for internal nodes of the tree, but as is often the case, once the internal parts are solved correctly, one can locally fix the labels near the root and the leaves.

Maximal independent set is a well-known symmetry-breaking problem, and e.g.\ in the case of a directed path ($\delta = 1$) it is known to be as hard as e.g.\ $3$-coloring. Hence one might expect that MIS on rooted regular binary trees also has got the complexity of $\Theta(\log^* n)$ rounds in the $\LOCAL$ model. \emph{This is not the case---maximal independent set in rooted binary trees can be solved in constant time!} Indeed, this is a good example of a non-trivial constant-time-solvable problem. It can be solved in exactly $4$ rounds, using the following idea (again, omitting some minor details related to what happens e.g.\ near the root).

\begin{figure}
\figlayout

(a) Label the nodes with $4$-bit strings based on how to reach them:
\vspace{-4mm}
\begin{center}
\includegraphics[page=14,scale=\figscale]{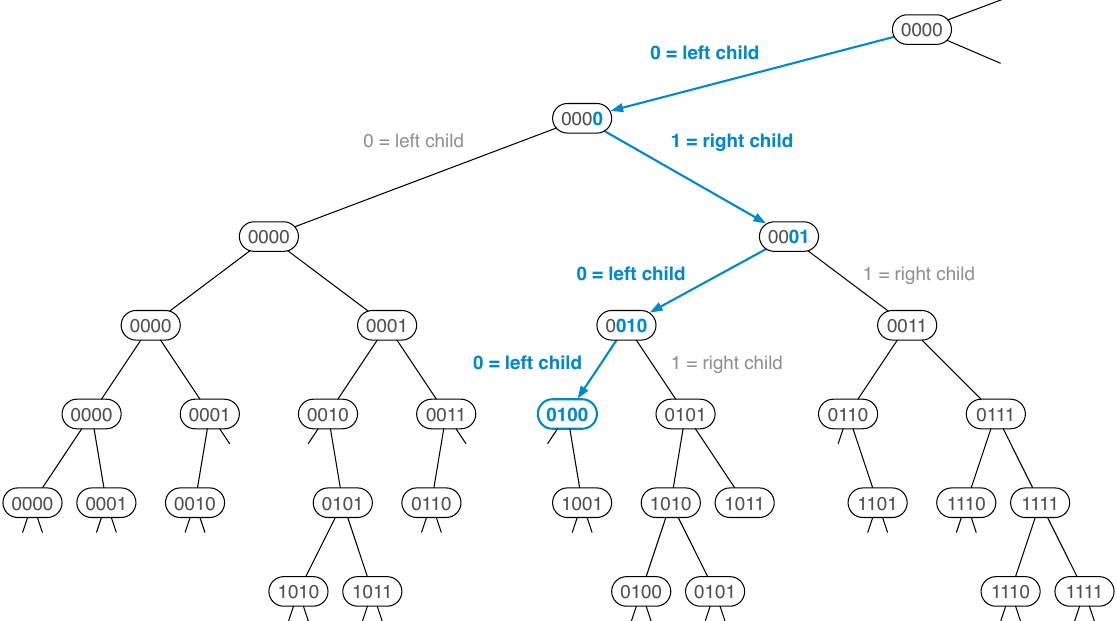}
\end{center}
\bigskip
\bigskip
\bigskip

(b) Choose the corresponding element from the 16-element string \eqref{eq:mis-solution}:
\vspace{-4mm}
\begin{center}
\includegraphics[page=15,scale=\figscale]{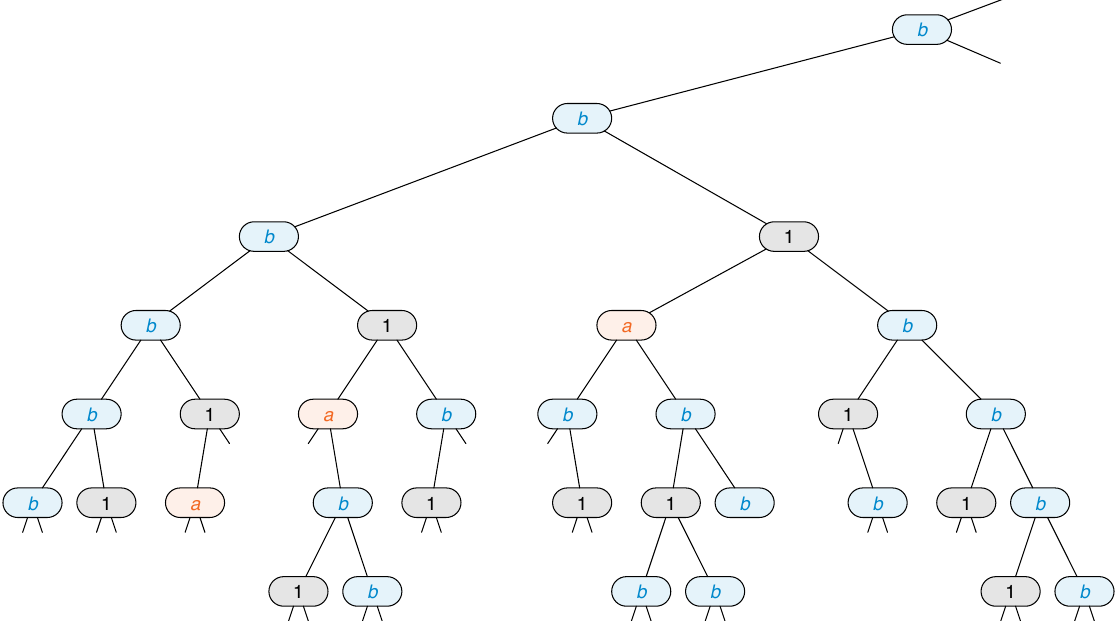}
\end{center}

\caption{Finding a maximal independent set in $O(1)$ rounds (Section~\ref{ssec:mis-example}).}\label{fig:mis-alg}
\Description{}
\end{figure}

First, we need to pick some consistent way of referring to your ``left'' child and the ``right'' child (for this reason, we can assume that a port numbering is available, that is, a node can send a message to a specific child, by indexing it with a number from $1$ to $\delta$, or we can assume that nodes have unique identifiers, and then we can order the children by their unique identifiers).
Label all nodes first with an empty string. Then we repeat the following step for $4$ times: add $0$ to your string and send it to your left child, and add $1$ to your string and send it to your right child. Your new label is the label that you received from your parent. This way all nodes get labeled with a $4$-bit string (see \autoref{fig:mis-alg}a). A key property is this: if my string is $xyzw$, the string of my parent is $0xyz$ or $1xyz$. Finally, interpret the binary string as a number between $0$ and $15$, and output the corresponding element of the following string (using $0$-based indexing; see \autoref{fig:mis-alg}b):
\begin{equation}\label{eq:mis-solution}
b\,1\,a\,b\,b\,b\,1\,b\,b\,1\,1\,b\,b\,b\,1\,b.
\end{equation}
One can verify the correctness of the algorithm by checking all $2^3$ possible cases: for example, if a node is labeled with $x010$, it will output either symbol $2$ of \eqref{eq:mis-solution}, which is $a$, or symbol $10$, which is $1$. Its two children will have labels $0100$ and $0101$, so they will output symbols $4$ and $5$ of \eqref{eq:mis-solution}, which are $b$ and $b$. This results in a configuration $a:bb$ or $1:bb$, both of which are valid in \eqref{eq:mis-example}.

The key point of the example is this: even though the algorithm is somewhat involved, we can use the computer program accompanying in this work to \emph{automatically} discover this algorithm and to determine that this problem is indeed constant-time solvable! Also, this problem demonstrates that there are $O(1)$-round-solvable locally checkable problems in rooted regular trees that require strictly more than zero rounds, while e.g.\ in the previously-studied family of binary labeling problems \cite{binary_lcls} all $O(1)$-round-solvable problems are known to be zero-round solvable.

\subsection{Example: branch 2-coloring}\label{ssec:branch-example}

As the final example, let us consider the following problem, with $\delta = 2$ and $\Sigma = \{1,2\}$:
\begin{equation}\label{eq:branch-example}
    C = \bigl\{
    1 : 12,\,
    2 : 11
    \bigr\}.
\end{equation}
This problem is, in essence, $2$-coloring with a choice: starting with a node of label $1$ and going downwards, there is always a monochromatic path labeled with $1,1,1,1,\dotsc$, and a properly colored path labeled with $1,2,1,2,\dotsc$. It turns out that the choice makes enough of a difference: the complexity of this problem is $\Theta(\log n)$ rounds. We encourage the reader to come up with an algorithm and a matching lower bound---with our techniques we get a tight result immediately.


\begin{table}
\caption{An overview of the landscape and decidability of the round complexity of $\LCL$ problems in the $\LOCAL$ model. The case studied in the present work (unlabeled, rooted, regular trees) is highlighted with shading, and the darker shade indicates the key new results. The decidability is given assuming P $\ne$ PSPACE $\ne$ EXPTIME. We have listed a few key references for each column, focusing on decidability aspects; the overall picture of the complexity landscape is the result of a long sequence of papers, including \cite{ColeVishkin86, Linial92, Naor1991, CKP19exponential, BFHKLRSU16, BBOS18almostGlobal, BHKLOS18lclComplexity, BBOS20paddedLCL}.}\label{tab:overview}
\centering
\newcommand{\mysp}{4.0mm}
\newcommand{\myss}{1.0mm}
\newcommand{\mysl}{2.0mm}
\newcommand{\hsp}{\hspace{\mysp}}
\newcommand{\hs}{\hspace{\myss}\hspace{\myss}}
\newcommand{\hsl}{\hspace{\mysl}}
\newcommand{\plog}{\log^\alpha}
\newcommand{\yy}{$\checkmark$}
\newcommand{\kludge}{\\[-0.02mm]}
\newcolumntype{t}{@{}>{\columncolor{ACMLightBlue!20!white}[\myss][\myss]}c@{}}
\newcolumntype{T}{@{\hspace{-\myss}}>{\columncolor{ACMLightBlue!50!white}[0pt][0pt]}c@{\hspace{\myss}}}
\newcommand{\hl}[1]{\multicolumn{1}{@{}T@{}}{#1}}
\begin{tabular}{@{}ll@{\hsp}c@{\hs}c@{\hs}c@{\hs}c@{\hsp}c@{\hs}c@{\hs}t@{\hs}c@{\hs}c@{\hsp}c@{}}
\toprule
\textbf{Setting}
&
& \multicolumn{4}{@{}l@{}}{\emph{Paths and cycles}}
& \multicolumn{5}{@{}l@{}}{\emph{Trees}}
& \emph{General} \\
\cmidrule(r{\mysp}){3-6}\cmidrule(r{\mysp}){7-11}\cmidrule{12-12}
& no input           & \yy & \yy & \yy &     & \yy & \yy &    \yy & \yy &     &     \kludge
& regular                   & \yy & \yy &     &     & \yy & \yy &    \yy & \yy &     &     \kludge
& directed or rooted        & \yy & \yy &     &     &     &     &    \yy &     &     &     \kludge
& binary output             & \yy &     &     &     &     & \yy &        &     &     &     \kludge
& homogeneous               &     &     &     &     & \yy &     &        &     &     &     \\
\midrule
\textbf{Complexity}
& $O(1)$                    & D+R & D+R & D+R & D+R & D+R & D+R &    D+R & D+R & D+R & D+R \kludge
\textbf{classes}
& $\cdots$                  & --- & --- & --- & --- & --- & --- &    --- & --- & --- & --- \kludge
& $\Theta(\log \log^* n)$   & --- & --- & --- & --- & --- & --- &\hl{---}& ?   & ?   & D+R \kludge
& $\cdots$                  & --- & --- & --- & --- & --- & --- &\hl{---}& ?   & ?   & D+R \kludge
& $\Theta(\log^* n)$        & --- & D+R & D+R & D+R & D+R & --- &    D+R & D+R & D+R & D+R \kludge
& $\cdots$                  & --- & --- & --- & --- & --- & --- &    --- & --- & --- & --- \kludge
& $\Theta(\log \log n)$     & --- & --- & --- & --- & R   & R   &\hl{---}& R   & R   & R   \kludge
& $\Theta(\plog \log n)$    & --- & --- & --- & --- & --- & --- &    --- & --- & --- & ?   \kludge
& $\cdots$                  & --- & --- & --- & --- & --- & --- &    --- & --- & --- & --- \kludge
& $\Theta(\log n)$          & --- & --- & --- & --- & D+R & D+R &    D+R & D+R & D+R & D+R \kludge
& $\cdots$                  & --- & --- & --- & --- & --- & --- &    --- & --- & --- & D+R \kludge
& $\Theta(n^{1/k})$         & --- & --- & --- & --- & --- & --- &\hl{($n$)}& ?   &($n$)& D+R \kludge
& $\cdots$                  & --- & --- & --- & --- & --- & --- &    --- & --- & --- & D+R \kludge
& $\Theta(n)$               & D+R & D+R & D+R & D+R & --- & D+R &    D+R & D+R & D+R & D+R \\
\midrule
\textbf{Decidability}
& P                         & \yy & \yy & \yy & --- & ?   & (D) &    ?   & ?   & --- & --- \kludge
& PSPACE                    & \yy & \yy & \yy & ?   & ?   & (D) &    ?   & ?   & --- & --- \kludge
& EXPTIME                   & \yy & \yy & \yy & ?   & ?   & (D) &\hl{($k$)}& ?   & ?   & --- \kludge
& decidable                 & \yy & \yy & \yy & \yy & ?   & (D) &\hl{\yy}& (H) & (H) & --- \\
\midrule
\textbf{References} &
& \cite{NaorStockmeyer95,Brandt2017,binary_lcls}
& \cite{NaorStockmeyer95,Brandt2017}
& \cite{lcls_on_paths_and_cycles}
& \cite{balliu19lcl-decidability}
& \cite{BHOS19HomogeneousLCL}
& \cite{binary_lcls}
& this
& \cite{Chang20,CP19timeHierarchy}
& \cite{Chang20,CP19timeHierarchy}
& \cite{NaorStockmeyer95,Brandt2017}
\kludge
& & & & & & & & work & & &
\\
\midrule
\textbf{Legend}
& \multicolumn{11}{l@{}}{\yy = yes} \\
& \multicolumn{11}{l@{}}{? = unknown} \\
& \multicolumn{11}{l@{}}{--- = not possible} \\
& \multicolumn{11}{l@{}}{$\alpha > 1$} \\
& \multicolumn{11}{l@{}}{$k = 2, 3, \dotsc$} \\
& \multicolumn{11}{l@{}}{D = class exists for deterministic algorithms} \\
& \multicolumn{11}{l@{}}{R = class exists for randomized algorithms} \\
& \multicolumn{11}{l@{}}{($n$) = the current construction assumes the knowledge of $n$; unknown without this information} \\
& \multicolumn{11}{l@{}}{($k$) = does not determine the value of $k$ for the class $\Theta(n^{1/k})$} \\
& \multicolumn{11}{l@{}}{(D) = known only for deterministic complexities, unknown for randomized} \\
& \multicolumn{11}{l@{}}{(H) = known only for classes between $\Omega(\log n)$ and $O(n)$} \\
\bottomrule
\end{tabular}
\end{table}

\subsection{Contributions}

As we have seen, the family of locally checkable problems in regular rooted trees is rich and expressive. Using auxiliary labels similar to what we saw in the MIS example in Section~\ref{ssec:mis-example}, we can encode, in essence, any locally checkable labeling problem ($\LCL$ problem) \cite{NaorStockmeyer95} in the classic sense, as long as the problem is such that the interesting part is related to what happens in the internal parts of regular trees. We have already seen that there are problems with at least four distinct complexity classes: $O(1)$, $\Theta(\log^* n)$, $\Theta(\log n)$, and $\Theta(n)$. In \autoref{sec:poly-region} we also show how to generate problems of complexity $\Theta(n^{1/k})$ for any $k = 1, 2, 3, \dotsc$

We prove in this work that \emph{this list is exhaustive}: any problem that can be represented in our formalism has complexity $O(1)$, $\Theta(\log^* n)$, $\Theta(\log n)$, or $\Theta(n^{1/k})$ in rooted regular trees with $n$ nodes. This is a \emph{robust} result that does not depend on the specific choice of the model of computing: the complexity of a given problem is the same, regardless of whether we are looking at the $\LOCAL$ model or the $\CONGEST$ model, and regardless of whether we are using deterministic or randomized algorithms.

One of the surprising consequences is that \emph{randomness does not help} in rooted regular trees. In unrooted regular trees there are problems (the canonical example being the \emph{sinkless orientation} problem) that can be solved with the help of randomness in $\Theta(\log \log n)$ rounds, while the deterministic complexity is $\Theta(\log n)$~\cite{Brandt2017}. This class of problems disappears in rooted trees.

Our main contribution is that the \emph{complexity of any given problem in this formalism is decidable}: there is an algorithm that, given the description of a problem $\Pi$ as a list of permitted configurations, outputs the computational complexity of problem $\Pi$, putting it in one of the four possible classes, i.e., determines whether the complexity is $O(1)$, $\Theta(\log^* n)$, $\Theta(\log n)$, or $\Theta(n^{1/k})$ for some $k$; in the fourth case our algorithm does not determine the exponent $k$, but then one could (at least in principle) use the more general decision procedure by \citet{Chang20} to determine the value of $k$.

While our algorithm takes in the worst case exponential time in the size of the description of $\Pi$, the approach is nevertheless \emph{practical}. We have \emph{implemented} the algorithm for the case of $\delta = 2$, and made it freely available online~\cite{AnonymousRepo}. Even though it is not at all optimized for performance, it classifies for example all of our sample problems above in a matter of \emph{milliseconds}.

We summarize our key results and compare them with prior work in Table~\ref{tab:overview}.

\section{Related work}

\subsection{Landscape of LCL problems in the LOCAL model}

\subsubsection{Paths and cycles}

We know that, on graph families such as paths and cycles, there are $\LCL${}s with complexities (both deterministic and randomized) of $O(1)$ (e.g., trivial problems), $\Theta(\log^* n)$~\cite{ColeVishkin86, Linial92, Naor1991} (e.g., $3$-coloring), and $\Theta(n)$ (e.g., global problems such as properly orienting a path/cycle).  Moreover, there are complexity \emph{gaps}, that is, in these families of graphs, there are no $\LCL${}s with round complexity between $\omega(1)$ and $o(\log^* n)$~\cite{NaorStockmeyer95}, and between $\omega(\log^* n)$ and $o(n)$~\cite{CKP19exponential}. These works show that the only possible complexities for $\LCL$ problems on paths and cycles are $O(1)$, $\Theta(\log^* n)$, and $\Theta(n)$, and randomness does not help in solving problems faster.

\subsubsection{Trees}

For the case of the graph family of trees almost everything is understood nowadays. As in the case of paths and cycles, we have $\LCL${}s with time complexities (both deterministic and randomized) $O(1)$, $\Theta(\log^* n)$, and $\Theta(n)$. On trees, we know that there is more: there are $\LCL$ problems with both deterministic and randomized complexity of $\Theta(\log n)$ (e.g., problems of the form ``copy the input of the nearest leaf''), and $\Theta(n^{1/k})$ for any $k \ge 2$~\cite{CP19timeHierarchy}. Moreover, there are cases where randomness helps, in fact there are problems that have $\Theta(\log n)$ deterministic and $\Theta(\log\log n)$ randomized complexity~\cite{BFHKLRSU16}. As far as gaps are concerned, let us first consider the spectrum of time complexities of $\omega(\log^* n)$, and then the one of $o(\log^* n)$. \citet{CKP19exponential} showed that the deterministic complexity of any $\LCL$ problem on bounded-degree trees is either $O(\log^* n)$ or $\Omega(\log n)$, while its randomized complexity is either $O(\log^* n)$ or $\Omega(\log \log n)$. Moreover, \citet{CP19timeHierarchy} showed that any algorithm that takes $n^{o(1)}$ rounds can be sped up to run in $O(\log n)$ rounds. \citet{BBOS18almostGlobal} showed that there is a gap between $\omega(\sqrt{n})$ and $o(n)$ for deterministic algorithms, and \citet{Chang20} extended these results and showed that there is a gap between $\omega(n^{1/k})$ and $o(n^{1/(k-1)})$, for any $k \ge 2$, for both deterministic and randomized algorithms.
The spectrum of time complexities of $o(\log^* n)$ is still not entirely understood. \citet{CP19timeHierarchy} showed that ideas similar to \citet{NaorStockmeyer95} can be used to prove that there are no $\LCL${}s on bounded-degree trees with time complexity between $\omega(1)$ and $o(\log \log^* n)$. Also, in the same paper, the authors conjectured that it should be possible to extend this gap up to $o(\log^* n)$. While this still remains an open question, \citet{BHOS19HomogeneousLCL} showed that such a gap exists for a special subclass of $\LCL${}s, called \emph{homogeneous} $\LCL${}s. Observe that all the mentioned results hold for the setting of \emph{unrooted} trees, and in this setting there are still many open questions related to \emph{decidability}.
In this work, we prove decidability results for a restriction of this setting, that is, for \emph{rooted} trees.

\subsubsection{General graphs}
In general bounded-degree graphs there are $\LCL${}s with the same time complexity as in trees, so the question is if there are also the same gaps, or if in the case of general graphs we have a denser spectrum of complexities. First of all, the gaps of the lower spectrum on trees hold also on general graphs: we still have the $\omega(1)$ \--- $o(\log\log^* n)$ gap for both deterministic and randomized algorithms, the $\omega(\log^* n)$ \--- $o(\log n)$ for deterministic algorithms, and the $\omega(\log^* n)$ \--- $o(\log\log n)$ gap for randomized algorithms. Also, \citet{CP19timeHierarchy} showed that any $o(\log n)$-round randomized algorithm can be sped up to run in $O(T_{\LLL})$ rounds, where $T_{\LLL}$ is the time required for solving with randomized algorithms the distributed constructive Lov\'asz Local Lemma problem ($\LLL$) \cite{CPS17DistrLLL} under a polynomial criterion. By combining this result with the results on the complexity of $\LLL$ by \citet{FischerGhaffari17LLL} and the network decomposition one by \citet{RG20NetDecomposition}, we get a gap for randomized algorithms between $\omega(\poly(\log\log n))$ and $o(\log n)$. \citet{BHKLOS18lclComplexity} showed that, differently from the case of trees, the regions between $\omega(\log\log^* n)$ and $o(\log^* n)$  and between $\omega(\log n)$ and $o(n)$ are dense. In fact, for any complexity $T$ in these regions, it is possible to define an $\LCL${} with a time complexity that is arbitrary close to $T$. Also, in the case of trees, randomness either helps exponentially or not at all, while in the case of general graphs this is not the case anymore. In fact, \citet{BBOS20paddedLCL} showed that there are $\LCL$ problems on general graphs where randomness helps only polynomially by defining $\LCL${}s with deterministic complexity $\Theta(\log^k n)$ and randomized complexity $\Theta(\log^{k-1} n\log\log n)$, for any integer $k\ge 1$.

\subsubsection{Special settings}

Over the years, researchers have investigated the complexity of interesting subclasses of $\LCL${}s. We already mentioned homogeneous $\LCL${}s on trees~\cite{BHOS19HomogeneousLCL}, that, on a high level, are $\LCL${}s for which the hard instances are $\Delta$-regular trees. For this subclass of $\LCL$ problems, the spectrum of deterministic complexities consists of $O(1)$, $\Theta(\log^* n)$, and $\Theta(\log n)$. Also, as in the case of trees, there are cases where randomness helps: there are homogeneous $\LCL${}s with $\Theta(\log n)$ deterministic and $\Theta(\log\log n)$ randomized complexity. These are the only possible complexities for homogeneous $\LCL${}s. \citet{Brandt2017} studied $\LCL${}s on $d$-dimensional torus grids, and showed that there are $\LCL${}s with complexity (both deterministic and randomized) $O(1)$, $\Theta(\log^* n)$, and $\Theta(n^{1/d})$. The authors showed that these are the only possible complexities, implying that randomness does not help. \citet{binary_lcls} studied \emph{binary labeling problems}, that are $\LCL${}s that can be expressed with no more than two labels in the \emph{edge labeling} formalism \cite{BBHORS19MMlowerBound,Olivetti2019REtor} (such $\LCL${}s include, for example, sinkless orientation). The authors showed that, in trees, there are no such $\LCL${}s with deterministic round complexity between $\omega(1)$ and $o(\log n)$, and between $\omega(\log n)$ and $o(n)$, proving that the spectrum of deterministic complexities of binary labeling problems in bounded-degree trees consists of $O(1)$, $\Theta(\log n)$ and $\Theta(n)$. The authors also studied the randomized complexity of binary labeling problems that have deterministic complexity $\Theta(\log n)$, showing that for some of them randomness does not help, while for some others it does help (note that from previous work we know that, in this case, randomness either helps exponentially or not at all). Determining the tight randomized complexity of all binary labeling problems is still an open question.

\subsection{Decidability of LCL problems}
As we have seen, there are often gaps in the spectrum of distributed complexities of $\LCL${}s. Hence, a natural question that arises is the following: given a specific $\LCL$, can we decide on which side of the gap it falls? In other words, are these classifications of $\LCL$ problems decidable? We can push this question further and ask whether it is possible to automate the design of distributed algorithms for optimally solving $\LCL{}$s. There is a long line of research that has investigated these kind of questions.

For graph families that consist of \emph{unlabeled} paths and cycles (that is, nodes do not have any label in input), the complexity of a given $\LCL$ is decidable~\cite{NaorStockmeyer95, Brandt2017, lcls_on_paths_and_cycles}. The next natural question is whether we have decidability in the case of trees (rooted or not). Because the structure of a tree can be used to encode input labels, researchers had to first understand the role of input labels in decidability. For this purpose, \citet{balliu19lcl-decidability} studied the decidability of \emph{labeled} paths and cycles, showing that the complexity of $\LCL${}s in this setting is decidable, but it is PSPACE-hard to decide it, and this PSPACE-hardness result extends also for the case of bounded-degree unlabeled trees (since the structure of the tree may encode input labels). The authors also show how to automate the design of asymptotically optimal distributed algorithms for solving $\LCL${}s in this context. Later, \citet{Chang20} improved these results showing that, in this setting, it is EXPTIME-hard to decide the complexity of $\LCL${}s. While the decidability on bounded degree trees is still an open question, there are some positive partial results in this direction. In fact, \citet{CP19timeHierarchy} along with the $\omega(\log n)$ \--- $n^{o(1)}$ gap, showed also that we can decide on which side of the gap the complexity of an $\LCL$ lies. Moreover, \citet{binary_lcls} showed that, the deterministic complexity of binary labeling problems on trees is decidable and we can automatically find optimal algorithms that solve such $\LCL${}s. The works of \citet{Brandt19RE} and \citet{Olivetti2019REtor} played a fundamental role in further understanding to which extent we can automate the design of algorithms that optimally solve $\LCL{}$s.

Unfortunately, in general, the complexity of an $\LCL$ is not decidable. In fact, Naor and Stockmeyer showed that, even on unlabeled non-toroidal grid graphs, it is undecidable whether the complexity of a given $\LCL$ is $O(1)$~\cite{NaorStockmeyer95}. For unlabeled toroidal grids, \citet{Brandt2017} showed that, given an $\LCL$, it is decidable whether its complexity is $O(1)$, but it is undecidable whether its complexity is $\Theta(\log^* n)$ or $\Theta(n)$. On the positive side, the authors showed that, given an $\LCL$ with round complexity $O(\log^* n)$, one can automatically find an $O(\log^* n)$ rounds algorithm that solves it.

\section{Road map}
We will start by providing some useful definitions in Section~\ref{sec:definitions}. Then, in Section~\ref{sec:superlog} we will consider the spectrum of complexities in the $\Omega(\log n)$ region (that is, $\Theta(\log n)$ and above). We will define an object called \emph{certificate for $O(\log n)$ solvability}, for which we will prove, in \autoref{thm:superlogdecidability}, that we can decide the existence in polynomial time. We will prove in \autoref{path-flexible-form-impl-log(n)-ub} that, if such a certificate for a problem exists, then the problem can be solved in $O(\log n)$ time with a deterministic algorithm, even in the $\CONGEST$ model, while if such a certificate does not exist then we will prove in \autoref{no-path-flexible-impl-n-lb} that the problem requires $n^{\Omega(1)}$ rounds, even in the $\LOCAL$ model and even for randomized algorithms. By combining these results, we will essentially obtain a \emph{decidable gap} between $\omega(\log n)$ and $n^{o(1)}$ that is robust on the choice of the model.

We will then consider, in Section~\ref{sec:sublog}, the spectrum of complexities in the $O(\log n)$ region. We will define the notion of \emph{certificate for $O(\log^* n)$ solvability}, and we will prove, in \autoref{certificate_search_runtime}, that we can decide in exponential time if such a certificate exists. We will also prove, in \autoref{log*_algorithm}, that the existence of such a certificate implies a deterministic $O(\log^* n)$ algorithm for the $\CONGEST$ model, while we will prove in \autoref{nologstarcert} that the non-existence of such a certificate implies an $\Omega(\log n)$ randomized lower bound for the $\LOCAL$ model. Hence, also in this case we obtain a decidable gap that is robust on the choice of the model.

In Section~\ref{sec:sublogstar}, we consider the spectrum of complexities in the $O(\log^* n)$ region. We will define the notion of \emph{certificate for $O(1)$ solvability}, that will be nothing else than a certificate for $O(\log^* n)$ solvability that has some special property. We will show, in \autoref{constant_certificate_search_runtime}, that also in this case, we can decide its existence in exponential time, and we will show in \autoref{certificate_and_aab_implies_O(1)} that its existence implies a constant-time deterministic algorithm for the $\CONGEST$ model, while we will show in \autoref{no1cert} that the non-existence implies an $\Omega(\log^* n)$ lower bound for the $\LOCAL$ model. Hence, we will obtain that there are only four possible complexities, $O(1)$, $\Theta(\log^* n)$, $\Theta(\log n)$, and $n^{\Omega(1)}$, and that for all problems we can decide which of these four complexities is the right one.

For the fine-grained structure inside the $n^{\Omega(1)}$ class we refer to the prior work \cite{Chang20,CP19timeHierarchy,smallmessages}; while these papers study the case of \emph{unrooted} trees, we note that the orientation can be encoded as a locally checkable input, and the results are also applicable here. It follows that there are only classes $O(1)$, $\Theta(\log^* n)$, $\Theta(\log n)$, and $\Theta(n^{1/k})$ for $k = 1, 2, \dotsc$, and the exact class (including the value of $k$) is decidable. Although  the existence of the gap $\omega(n^{1/(k+1)})$ -- $o(n^{1/k})$~\cite{Chang20} applies to regular rooted trees, the problems with complexity $\Theta(n^{1/k})$ that have been shown to exist in \cite{CP19timeHierarchy} are not defined on regular rooted trees (e.g., nodes of different degrees may have different constraints).  In \autoref{sec:poly-region}, we define problems with complexity  $\Theta(n^{1/k})$ in regular rooted trees, showing that the complexity class $\Theta(n^{1/k})$ is non-empty for regular rooted trees.

\section{Definitions}\label{sec:definitions}

In this section we define some notions that will be used in the following sections.

\subsection{Input graphs}

We assume that all input graphs will be \emph{unlabeled rooted trees} where each node has either exactly $\delta$ or $0$ children for some positive integer $\delta$. That is, input graphs are \emph{full $\delta$-ary trees}. For convenience, when not specified otherwise, a tree $T$ is assumed to be a full $\delta$-ary tree.

\subsection{Models of computing}

The models that we consider in this work are the classical $\LOCAL$ and $\CONGEST$ model of distributed computing. Let $G$ be any graph with $n$ nodes and maximum degree $\Delta$. In the $\LOCAL$ model, each node of $G$ is equipped with an identifier in $\{1,2,\ldots,\poly(n)\}$, and the initial knowledge of a node consists of its own identifier, its degree (i.e., the number of incident edges), the total number $n$ of nodes, and $\Delta$ (in the case of rooted trees, each node knows also which of its incident edges connects it to its parent). Nodes try to learn more about the input instance by communicating with the neighbors. The computation proceeds in synchronous rounds, and at each round nodes send messages to neighbors, receive messages from them, and perform local computation. Messages can be arbitrarily large and the local computational can be of arbitrary complexity. Each node must terminate its computation at some point and decide on its local output. The running time of a distributed algorithm running at each node in the $\LOCAL$ model is determined by the number of rounds needed such that all nodes have produced their local output. In the randomized version of the $\LOCAL$ model, each node has access to its own stream of random bits. The randomized algorithms considered in this paper are Monte Carlo ones, that is, a randomized algorithm of complexity $T$ that solves a problem $P$ must terminate at all nodes upon $T$ rounds and this should result in a global solution for $P$ that is correct with probability at least $1 - 1/n$.

There is only one difference between the $\CONGEST$ and the $\LOCAL$ model, and it lies in the size of the messages. While in the $\LOCAL$ model messages can be arbitrarily large, in the $\CONGEST$ model the size of the messages is bounded by $O(\log n)$ bits.

\subsection{LCL problems}

We define $\LCL$ problems as follows.
\begin{definition}[$\LCL$ problem]
An $\LCL$ problem is a triple $\Pi = (\delta,\Sigma,\CC)$ where:
\begin{itemize}
	\item $\delta$ is the number of allowed children;
	\item $\Sigma$ is a finite set of (output) labels;
	\item $\CC$ is a set of tuples of size $\delta+1$ from $\Sigma^{\delta+1}$ called \emph{allowed configurations}.
\end{itemize}
\end{definition}
A configuration $(a,b_1,\ldots,b_\delta)$ will also be written as $(a : b_1,\ldots,b_\delta)$, in order to highlight that the label $a$ is for the parent and $b_1,\ldots,b_\delta$ are labels of the leaves. Sometimes we will omit the commas, and just write $(a : b_1\ldots b_\delta)$. Sometimes even the parenthesis will be omitted, obtaining $a : b_1\ldots b_\delta$, that is the notation used in e.g.\ Section~\ref{ssec:mis-example}.
As a shorthand notation, for an $\LCL$ problem $\Pi$, we will also denote the labels and configurations of $\Pi$ by $\Sigma(\Pi)$ and $\CC(\Pi)$.

\begin{definition}[solution]
A solution to an $\LCL$ problem $\Pi$ for a tree $T$ is a labeling function $\lambda$ for which:
\begin{itemize}
	\item every node $v \in T$ is labeled with a label $\lambda(v)$ from $\Sigma(\Pi)$;
	\item every node $v \in T$ with $\delta$ children $v_1,\dots,v_\delta$ satisfies that there exists a permutation $\rho \colon \{1, \dots, \delta\} \to \{1, \dots, \delta\}$ such that  $(\lambda(v) : \lambda(v_{\rho(1)}),\dots,\lambda(v_{\rho(\delta)}))$ is in $\CC(\Pi)$.
\end{itemize}
\end{definition}
In other words, a solution is a labeling for the nodes that must satisfy some local constraints. Note that only nodes with $\delta$ children are constrained, but that such $\LCL$ problems could be well-defined even on non-full $\delta$-ary trees (nodes with a number of children different from $\delta$ are unconstrained). Full $\delta$-ary trees are the hardest instances for the problems as every node is constrained.

\begin{definition}[restriction]\label{def:restriction}
	Given an $\LCL$ problem $\Pi = (\delta,\Sigma,\CC)$, a restriction of $\Pi$ to labels $\Sigma' \subseteq \Sigma$ is a new $\LCL$ problem $\Pi' = (\delta,\Sigma',\CC')$, where $\CC'$ consists of all configurations in $\CC$ that only use labels in $\Sigma'$.
\end{definition}

\begin{definition}[continuation below]
Let $\Pi$ be an $\LCL$ problem. Label $\sigma \in \Sigma(\Pi)$ has a \emph{continuation below} if there exists a configuration $(\sigma : \sigma_1,\dots,\sigma_\dd) \in \CC(\Pi)$.
\end{definition}

\begin{definition}[continuation below with specific labels]
	Let $\Pi = (\delta,\Sigma,\CC)$ be an $\LCL$ problem. Label $\sigma \in \Sigma(\Pi)$ has a \emph{continuation below with labels in $\Sigma' \subseteq \Sigma$} if there exists a configuration $(\sigma : \sigma_1,\dots,\sigma_\dd) \in \CC(\Pi)$ such that $\{\sigma,\sigma_1,\dots,\sigma_\dd\} \subseteq\nobreak \Sigma'$.
\end{definition}

\begin{definition}[path-form of an $\LCL$ problem]
\label{path_form}
Let $\Pi = (\delta,\Sigma,C)$ be an $\LCL$ problem. The \emph{path-form} of $\Pi$ is the $\LCL$ problem $\PathPi = (1,\Sigma,C')$, where $(a : b) \in C'$ if and only if there exists a configuration $(a : b_1, b_2,\dots,b_\dd) \in C$ with $b = b_i$ for some $i$.
\end{definition}

See Figure~\ref{fig:path-flexible-form}b for an illustration of the path-form.

\subsection{Automata and flexibility}

\begin{definition}[automaton associated with path-form of an $\LCL$ problem; \cite{lcls_on_paths_and_cycles}]
\label{def:automaton}
Let $\Pi$ be an $\LCL$ problem. The automaton $\M(\Pi)$ associated with the path-form of $\Pi$ is a nondeterministic unary semiautomaton defined as follows:
\begin{itemize}
	\item The set of states is $\Sigma(\Pi)$.
	\item There is a transition from state $a$ to state $b$ if there is a configuration $(a : b)$ in the path-form $\PathPi$ of $\Pi$.
\end{itemize}
\end{definition}

See Figure~\ref{fig:path-flexible-form}c for an illustration of the automaton.

\begin{definition}[flexible state of an automaton; \cite{lcls_on_paths_and_cycles}]\label{def:flexible-state}
\label{flexible_state}
A state $a$ from $\M(\Pi)$ is \emph{flexible} if there is a natural number $K$ such that for all $k \ge K$ there is a walk $a \leadsto a$ of length exactly $k$ in~$\M(\Pi)$. The smallest number $K$ that satisfies this property is the \emph{flexibility} of state $a$, in notation $\flexibility(a)$.
\end{definition}

As the set of states of the automaton is the set of labels, we can expand the notion of flexibility of a state to the notion of flexibility of a label.

\begin{definition}[path-flexibility]\label{def:path-flexibility}
Let $\Pi$ be an $\LCL$ problem and $\PathPi$ its path-form. A label $\sigma \in \Sigma(\Pi)$ is \emph{path-flexible} if $\sigma$ is a flexible state in automaton $\M(\Pi)$, and \emph{path-inflexible} otherwise. 

Moreover, an $\LCL$ problem $\Pi$ is \emph{path-flexible} if all labels are path-flexible labels and its automaton $\M(\Pi)$ has one strongly connected component.
\end{definition}

\subsection{Graph-theoretic definitions}

\begin{definition}[root-to-leaf path]
A \emph{root-to-leaf} path in a tree is a path that starts at the root and ends at one of its leaves.
\end{definition}

\begin{definition}[hairy path]
A full $\delta$-ary tree $T$ is called a \emph{hairy path} if it can be obtained by attaching leaves to a directed path such that all nodes of the path have exactly $\delta$ children.
\end{definition}

\begin{definition}[minimal absorbing subgraph]\label{def:minimal-absorbing-subgraph}
Let $G$ be a directed graph. A subgraph $G' \subseteq G$, is called a \emph{minimal absorbing subgraph} if $G'$ is a strongly connected component of $G$ and $G'$ does not have any outgoing edges.
\end{definition}

We note that a minimal absorbing subgraph exists for any directed graph.

\begin{definition}[ruling set]
Let $G$ be a graph. A \emph{$(k,l)$-ruling set} is a subset $S$ of nodes of $G$ such that the distance between any two nodes in $S$ is at least $k$, and the distance between any node in $G$ and the closest node in $S$ is at most~$l$.
\end{definition}

\section{Super-logarithmic region}\label{sec:superlog}
In this section we prove that there is no $\LCL$ problem $\Pi$ with distributed time complexity between $\omega(\log n)$ and $n^{o(1)}$. Also, we prove that, given a problem $\Pi$, we can \emph{decide} if its complexity is $O(\log n)$ or $n^{\Omega(1)}$.
In view of~\cite{Chang20,CP19timeHierarchy}, randomness does not help for $\LCL$ problems with round complexity $\Omega(\log n)$, so we focus on the deterministic setting.

\subsection{High-level idea}

The key idea is that we \emph{iteratively prune} the description of problem $\Pi$ by removing subsets of labels that we call \emph{path-inflexible}---these are sets of labels that require long-distance coordination (cf.\ $2$-coloring). After each such step, we may arrive at a subproblem that contains a new path-inflexible set, but eventually the pruning process will terminate, as there is only a finite number of labels.

Assume the pruning process terminates after $k$ steps. Let $X_1,\allowbreak X_2, \dotsc, X_k$ be the sets of labels we removed during the process, and let $X'$ be the set of labels that is left after no path-inflexible labels remain. We have two cases:
\begin{enumerate}
  \item Set $X'$ is empty. In this case we can show that the round complexity of the problem $\Pi$ is at least $\Omega(n^{1/k})$. To prove this, we make use of a $k$-level construction that generalizes the one used for $2\frac12$-coloring in \cite{CP19timeHierarchy}. We argue that, roughly speaking, for $o(n^{1/k})$-round algorithms, no label from set $X_i$ can be used for labeling the level-$j$ nodes, for each $j \geq i$, as this requires coordination over distance $\Theta(n^{1/k})$.

  \item Set $X'$ is non-empty. In this case, after removing the sets $X_i$, we are left with a non-empty \emph{path-flexible} subproblem $\Pi' \subseteq \Pi$, and we can make use of the flexibility to solve $\Pi'$ in $O(\log n)$ rounds. Hence the original problem $\Pi$ is also solvable in $O(\log n)$ rounds.
\end{enumerate}
We say that problem $\Pi$ has a \emph{certificate for $O(\log n)$ solvability} if and only if the set $X'$ is non-empty.

Formally, we prove the following theorems.

\begin{theorem}
\label{path-flexible-form-impl-log(n)-ub}
Let\/ $\Pi$ be a problem having a certificate for $O(\log n)$ solvability. Then $\Pi$ is solvable in $O(\log n)$ rounds in the $\CONGEST$ model.
\end{theorem}
\begin{theorem}
\label{no-path-flexible-impl-n-lb}
Let\/ $\Pi$ be an $\LCL$ problem having no certificate for $O(\log n)$ solvability.
  Then both the randomized and the deterministic complexity of\/ $\Pi$ in the $\LOCAL$ model are $\Omega(n^{1/k})$ for some $k \ge 1$.
\end{theorem}
\begin{theorem}\label{thm:superlogdecidability}
Whether an $\LCL$ problem $\Pi$ has round complexity $O(\log n)$ or $n^{\Omega(1)}$ can be decided in polynomial time.
\end{theorem}

\subsection{Certificate}

\begin{figure}
\figlayout

(a) Problem $\Pi_0$:
\begin{center}
\includegraphics[page=7,scale=\figscale]{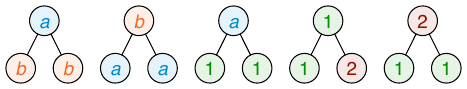}
\end{center}
\medskip

(b) Path-form $\PathPi_0$ of $\Pi_0$:
\begin{center}
\includegraphics[page=8,scale=\figscale]{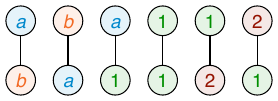}
\end{center}
\medskip

(c) Automaton $\M(\Pi_0)$:
\begin{center}
\includegraphics[page=9,scale=\figscale]{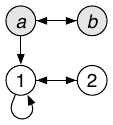}
\end{center}

(d) Problem $\Pi_1$:
\begin{center}
\includegraphics[page=10,scale=\figscale]{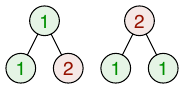}
\end{center}

(e) Path-form $\PathPi_1$ of $\Pi_1$:
\begin{center}
\includegraphics[page=11,scale=\figscale]{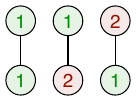}
\end{center}

(f) Automaton $\M(\Pi_1)$:
\begin{center}
\includegraphics[page=12,scale=\figscale]{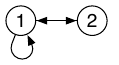}
\end{center}

(g) Problem $\Pi_\mathrm{pf}$:
\begin{center}
\includegraphics[page=13,scale=\figscale]{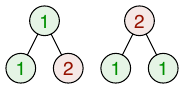}
\end{center}

\caption{Certifying that problem $\Pi_0$ is solvable in $O(\log n)$ rounds (see~\autoref{alg:path_flexible_form_finder}). The sample problem $\Pi_0$ is a combination of the branch $2$-coloring problem \eqref{eq:branch-example}, using labels $1$ and $2$, and the normal $2$-coloring problem \eqref{eq:2col-example}, using labels $a$ and $b$. The grayed out states in the automaton denote inflexible states. Problem $\Pi_\mathrm{pf}$ is the \emph{path-flexible} form of problem $\Pi_0$.}\label{fig:path-flexible-form}
\Description{}
\end{figure}

We present an algorithm that decides whether the complexity of a given problem $\Pi$ is $O(\log n)$ or $n^{\Omega(1)}$ rounds.

We start by defining a procedure that, given a problem $\Pi$, creates its restriction $\Pi'$ to path-flexible states of $\Pi$; see \autoref{alg:inflexible_removal}. However, note that states that were path-flexible in $\Pi$ may become path-inflexible in $\Pi'$; hence problem $\Pi'$ may still contain path-inflexible states.

\begin{algorithm}[tbp]
\caption{\label{alg:inflexible_removal}$\removePathInflexibleConfigurations(\Pi)$}
\DontPrintSemicolon
\KwIn{$\LCL$ problem $\Pi$}
\KwOut{$\LCL$ problem $\Pi'$, a restriction of $\Pi$ to its path-flexible states}
\BlankLine
Construct $\PathPi$, the path-form of $\Pi$. \comm*{See \autoref{path_form}.}
Construct the automaton $\M(\Pi)$. \comm*{See \autoref{def:automaton}.}
$\Sigma' \gets$ the set of path-flexible states of $\Pi$. \comm*{See \autoref{def:path-flexibility}.}
$\Pi' \gets$ the restriction of $\Pi$ to labels $\Sigma'$. \comm*{See \autoref{def:restriction}.}
\Return $\Pi'$
\end{algorithm}

Next we describe a new procedure $\findLogCertificate$ that uses \autoref{alg:inflexible_removal} to analyze the complexity of a given problem. This procedure either returns $\epsilon$ to indicate that the problem requires $n^{\Omega(1)}$ rounds, or it returns a new problem $\PfPi$ that is a restriction of $\Pi$, but that will be nevertheless solvable in $O(\log n)$ rounds (and therefore $\Pi$ is also solvable in this time).

Informally, procedure $\findLogCertificate$ applies iteratively \autoref{alg:inflexible_removal} until one of the following happens:
\begin{itemize}
	\item We obtain an empty problem. In this case we return $\epsilon$. We will show that this can only happen if $\Pi$ requires $n^{\Omega(1)}$ rounds.
	\item We reach a non-empty fixed point $\Pi_i$. In this case we further restrict $\Pi_i$ to the labels that induce a minimal absorbing subgraph in the automaton associated with its path-form. Let $\PfPi$ be the problem constructed this way. We return $\PfPi$, and we say that $\PfPi$ is the \emph{certificate for\/ $O(\log n)$ solvability}. We will show that $\PfPi$ and hence also the original problem $\Pi$ can be solved in $O(\log n)$ rounds.
Note that a minimal absorbing subgraph has the property that any labeling of the two endpoints of a sufficiently long path with labels from the subgraph admits an extension of the solution to the entire path with labels from the subgraph.
This provides the intuition why reducing the labels to those of a minimal absorbing subgraph allows for an $O(\log n)$-round algorithm using the rake-and-compress approach explained in Section~\ref{sec:upbo}.
\end{itemize}
The procedure is described more formally in \autoref{alg:path_flexible_form_finder}, and an example of execution for a concrete problem can be seen in \autoref{fig:path-flexible-form}.

\begin{algorithm}[tbp]
\DontPrintSemicolon
\caption{\label{alg:path_flexible_form_finder}$\findLogCertificate(\Pi)$}
\KwIn{$\LCL$ problem $\Pi$}
\KwOut{Certificate for $O(\log n)$ solvability if it exists, or $\epsilon$ otherwise}
\BlankLine
$\Pi_0 \gets \Pi$\;
$i \gets 0$\;
\Repeat{$\Pi_{i} = \Pi_{i-1}$}{
	$\Pi_{i+1} \gets \removePathInflexibleConfigurations(\Pi_i)$ \comm*{See \autoref{alg:inflexible_removal}.}
	$i \gets i+1$
}
\uIf{$\Pi_{i}$ is empty}{
	\Return $\epsilon$ \comm*{$\Pi$ cannot be solved in $O(\log n)$ rounds.}
	}
\Else{
  $\Sigma' \gets$ labels that induce a minimal absorbing subgraph of automaton $\M(\Pi_i)$ \comm*{See \autoref{def:minimal-absorbing-subgraph}.}
  $\PfPi \gets$ the restriction of $\Pi_i$ to $\Sigma'$\;
	\Return $\PfPi$ \comm*{Certificate for $O(\log n)$ solvability.}
}
\end{algorithm}

Now, let us prove some of the properties of \autoref{alg:path_flexible_form_finder}. First, we observe that this is indeed a polynomial-time algorithm.

\begin{lemma}\label{lem:superlogdecidability}
\autoref{alg:path_flexible_form_finder} runs in polynomial time in the size of the description of\/ $\Pi$.
\end{lemma}
\begin{proof}
When creating a successive restrictions of $\Pi$ in \autoref{alg:path_flexible_form_finder}, we always remove at least one label. Hence we invoke \autoref{alg:inflexible_removal} at most $|\Sigma(\Pi)|$, and \autoref{alg:inflexible_removal} runs in polynomial time \cite{lcls_on_paths_and_cycles}.
\end{proof}

Then we prove that the step where we restrict to a minimal absorbing subgraph behaves well; in particular, it will preserve flexibility.
\begin{lemma}
\label{properties_of_pfpi}
Let\/ $\Pi$ be a non-empty\/ $\LCL$ problem, such that all of its states are path-flexible. Let\/ $\Sigma'$ be a set of labels that induces a minimal absorbing subgraph of automaton $\M(\Pi)$, and let\/ $\PfPi$ be the restriction of\/ $\Pi$ to labels\/ $\Sigma'$.
Then all states of\/ $\PfPi$ are flexible, there is a walk between any two states of $\M(\PfPi)$, and $\M(\PfPi)$ has at least one edge.
\end{lemma}
\begin{proof}
First, let us prove that the restriction will preserve the flexibility of the states that remain. Since for every state in a minimal absorbing subgraph all outgoing edges are connected to states in the same minimal absorbing subgraph by definition, then no configuration for these states will be removed, and all returning walks for a state will stay.
Second, a walk between any two states of $\M(\PfPi)$ is implied by the fact that $\M(\PfPi)$ is strongly connected.
Lastly, $\M(\PfPi)$ has at least one edge, as \emph{every node} has returning walks, hence incoming and outgoing edges, and the minimal absorbing subgraph is non-empty.
\end{proof}

In the rest of the section, we will prove that our certificate for $O(\log n)$ solvability indeed characterizes $O(\log n)$ solvability in the following sense: if $\Pi$ has a certificate for $O(\log n)$ solvability, then $\Pi$ can be solved in $O(\log n)$ rounds, otherwise there is an $n^{\Omega(1)}$ lower bound for $\Pi$. Hence
\autoref{lem:superlogdecidability} implies \autoref{thm:superlogdecidability}.

\subsection{Upper bound}\label{sec:upbo}

We prove that, if \autoref{alg:path_flexible_form_finder} does not return $\epsilon$, then the original problem $\Pi$ can be solved in $O(\log n)$ rounds. Note that $\PfPi$, the result of \autoref{alg:path_flexible_form_finder}, is obtained by considering a subset of labels of $\Pi$ and all constraints that use only those labels, hence a solution for $\PfPi$ is also a valid solution for $\Pi$. Hence, we prove our claim by providing an algorithm solving $\PfPi$ in $O(\log n)$ rounds. For this purpose, we start by providing a procedure that is a modified version of the rake and compress procedures of \citet{Miller1985}, where, informally, we remove degree-$2$ nodes only if they are contained in long enough paths. We start with some definitions.
Note that in a rooted tree we assume that each edge $\{u,v\}$ is oriented from $u$ to $v$ if $v$ is the parent of $u$.

\begin{definition}[leaves]
Let $G=(V,E)$ be a graph. We define that $\leaves(G) \subseteq V$ is the set of all nodes with indegree~$0$.
\end{definition}

\begin{definition}[long-paths]
Let $G=(V,E)$ be a graph and $p$ be a constant. Let $X \subseteq V$ consist of all nodes of indegree~$1$. We define that $\longpaths(G,p) \subseteq X$ consists of the set of all nodes that belong to a connected component of size at least $p$ in the subgraph of $G$ induced by $X$.
\end{definition}

We now define our variant of the rake-and-compress procedure. Note that a similar variant, for unrooted trees, appeared in \cite{binary_lcls}.

\begin{definition}[$\RCP$]
Let $p \in \set{1,2,\dotsc}$. Procedure $\RCP(p)$ iteratively partitions the set of nodes $V$ into non-empty sets $V_1,V_2,\dotsc,V_L$ for some $L$ as follows:
\[
G_0 = G, \quad
V_{i+1} = \leaves(G_i) \cup \longpaths(G_i,p) \quad
G_{i+1} = G_i \setminus V_{i+1}.\]
\end{definition}

Note that the graphs $G_i$ can be disconnected.

We now prove an upper bound on the highest possible layer obtained by the procedure. In particular, we prove that there is some layer $L = O(\log n)$ such that $G_{L+1}$ is empty. For this purpose, we now prove that the number of nodes of $G_{i+1}$ is at least a $(1-\frac{1}{6p})$ factor smaller than the number of nodes of $G_i$, implying that after $O(\log n)$ steps we obtain an empty graph.

\begin{lemma}
	Let $p$ be a constant and let $G = (V,E)$ be a tree with $n$ nodes. At least one of the following holds:
	\[
	\bigl|\mspace{1mu}\leaves(G)\mspace{1mu}\bigr| \geq \frac{n}{6p}
	\quad\text{or}\quad
	\bigl|\mspace{1mu}\longpaths(G,p)\mspace{1mu}\bigr| \geq \frac{n}{3} .
	\]
\end{lemma}
\begin{proof}
Let $n_0,n_1,n_{2+}$ be the number of nodes of indegree 0, 1, and 2 or more, respectively. We have $n = n_0 + n_1 + n_{2+}$. The number of edges in a tree is $m = n - 1$ or by counting using indegrees we get $m \geq 0 n_0 + 1 n_1 + 2 n_{2+}$, from which we obtain $ n_0 + n_1 + n_{2+} - 1 \geq 0 n_0 + 1 n_1 + 2 n_{2+}$. Hence $n_{2+} < n_0$.
We have $|\mspace{1mu}{\leaves(G)}\mspace{1mu}| = n_0$ so if $n_0 \geq n/(6p)$, the claim holds. In what follows, we assume that $n_0 < n/(6p)$ which together with $n_{2+} < n_0$ implies that $n_0 + n_{2+} < n/(3p)$. This implies that the total number of nodes of indegree $1$ nodes is $n_1 = n - (n_0 + n_{2+}) > n - n/(3p) \geq 2n/3$.
Consider the subgraph $G_1$ induced by indegree-1 nodes of $G$. If we contract each connected component of $G_1$ into an edge, we obtain a tree $G'$ in which we have $n' = n_{2+} + n_0$ nodes and $m' = n_{2+} + n_0 - 1$ edges. As each edge represents at most one connected component of $G_1$, there are fewer than $n/(3p)$ components in $G_1$. Hence we have $n_1 > 2n/3$ indegree-1  nodes that are contained in less than  $n/(3p)$ connected components. Since components of size less than $p$ can contain at most $p \cdot n/(3p) = n/3$ nodes in total, then there have to be at least $2n/3 - n/3 = n/3$ nodes in the components of size at least $p$, hence $|\mspace{1mu}{\longpaths(G,p)}\mspace{1mu}| \geq n/3$.
\end{proof}

We now prove an upper bound on the time required for all nodes $v$ to know the layer $i$ in which they belong, that is, the layer $i$ satisfying $v \in V_i$.
\begin{lemma}
\label{rcp_is_O(log(n))}
$\RCP(p)$ can be computed in $O(\log n)$ rounds in the $\LOCAL$ and $\CONGEST$ models.
\end{lemma}
\begin{proof}
	We build a virtual graph where we iteratively remove nodes for $O(\log n)$ rounds. At each step, nodes can check in $1$ round which neighbors have already been removed, and hence compute their indegree in the virtual graph. Nodes mark themselves as removed if their indegree is $0$, or if their indegree is $1$ and they are in paths of length at least $p$. The result of each node is the step in which they have been marked as removed. Notice that each step requires $O(p)$ rounds, even in $\CONGEST$, and since $p$ is a constant, this procedure requires $O(\log n)$ rounds in total.
\end{proof}

We are now ready to prove that, if \autoref{alg:path_flexible_form_finder} returns some problem  $\PfPi$, then $\PfPi$ (and hence $\Pi$) can be solved in $O(\log n)$ rounds, proving \autoref{path-flexible-form-impl-log(n)-ub}.

\begin{proof}[Proof of \autoref{path-flexible-form-impl-log(n)-ub}]
Let $\Pi$ be a problem having a certificate for $O(\log n)$ solvability. Then we will show that $\Pi$ is $O(\log n)$ solvable in the $\CONGEST$ model.

Let $\PfPi$ be the path-flexible form from \autoref{alg:path_flexible_form_finder}.
Let \[k = \max_{\sigma \in \Sigma(\PfPi)} (\flexibility(\sigma)) + \bigl|\Sigma(\PfPi)\bigr|,\] where $\flexibility(\sigma)$ is the flexibility of a state $\sigma$ in $\M(\PfPi)$ as defined in \autoref{flexible_state}.

Given a tree $T$, we solve $\PfPi$ as follows. We start by running the procedure $\RCP(k)$ on $T$. After this process, each node $v$ knows the set $V_i$, $1 \le i\le L$, which it belongs to. Then, we compute a distance-$k$ coloring by using a palette of $O(1)$ colors, which can be done in $O(\log^* n)$ rounds, even in $\CONGEST$, since $\dd$ is constant (using, e.g., Linial's algorithm~\cite{Linial92} on power graphs).

We then process the layers one by one, from layer $L$ to $1$. For each layer $i$, we label all (unlabeled) nodes in $V_i$ and all of their children. We need to deal with two cases, either we are labeling a long path or we are labeling a leaf node (both in $G_{i-1}$).

If a node $v_j \in V_i$ is a leaf node, then by definition its children were not processed yet so they are not labeled. Node $v_j$ could be labeled with its parent, or it is unlabeled. But in both cases, we can complete this labeling (by labeling the descendants of $v_j$ and possibly $v_j$ itself) as every label has a continuation below.

If we need to label a long path $P$, then by construction all inner nodes have no fixed labels so far. The topmost node can be labeled (as we may have already processed its parent) and the bottom-most node has indegree one, and thus it is connected to exactly one node from an upper layer, and hence it will have exactly one child already labeled. To label all nodes of $P$, we proceed in several steps. First, we exploit the precomputed distance-$k$ coloring to compute a $(k,k)$-ruling set on each path in parallel in constant time, by iterating through the constantly many color classes and adding to the ruling set all nodes of the processed color for which no node in distance at most $k-1$ is already contained in the ruling set.
The ruling set nodes split the path into constant-length chunks.
Next, for each endpoint of the path, we remove the closest ruling set node from the ruling set.
This ensures that all chunks are of length at least $k$.
Then, we label all nodes that still remain in the ruling set with an arbitrarily chosen label from $\Sigma(\PfPi)$.
Finally, we label the nodes in the constant-length chunks (and their children) in a consistent manner.
This is possible since each label used for the ruling set nodes is flexible and has a walk to any other label in $\M(\PfPi)$ (as proved by \autoref{properties_of_pfpi}) and the ruling set nodes are far enough apart (more than the flexibility of any label in $\PfPi$).

As all of these steps can be performed in constant time (provided the precomputed distance-$k$ coloring), we can label the whole tree in $L \cdot O(1) + O(\log^* n) = O(\log n)$ rounds.
\end{proof}

\subsection{Lower bound}\label{sec:poly-lb}

We prove that if \autoref{alg:path_flexible_form_finder} returns $\epsilon$, then the original problem $\Pi=(\delta, \Sigma, \CC)$ requires $n^{\Omega(1)}$ rounds to solve.

\paragraph{A sequence of labels.} If
\autoref{alg:path_flexible_form_finder} returns $\epsilon$ after $k$ iterations, then there is a sequence $\Sigma_1, \Sigma_2 \ldots, \Sigma_k$ of sets of labels meeting the following conditions and leading to a sequence $\Pi_0, \Pi_1, \Pi_2 \ldots, \Pi_k$ of $\LCL$ problems:
\begin{itemize}
    \item $\Pi_0 = \Pi$.
    \item For $1 \leq i \leq k$, $\Pi_i$ is the $\LCL$ problem that is the restriction of $\Pi_{i-1}$ to the label set $\Sigma(\Pi_{i-1})  \setminus \Sigma_i$, or equivalently, $\Pi_i$ is the restriction of the original $\LCL$ problem $\Pi$ to the label set $\Sigma \setminus (\Sigma_{1} \cup \Sigma_{2} \cup \cdots \cup \Sigma_i)$.
    \item For $1 \leq i \leq k$, $\Sigma_i$ is the set of path-inflexible labels in $\Pi_{i-1}$.
    \item $\Sigma_{\Pi_k} = \emptyset$, so $\Sigma = \Sigma_1 \cup \Sigma_2 \cup \cdots \cup \Sigma_k$ is a partition of $\Sigma$.
\end{itemize}

The set of labels $\Sigma_i \subseteq \Sigma$ consists of the labels removed during the $i$th iteration of \autoref{alg:path_flexible_form_finder}, as they are path-inflexible in  $\Pi_{i-1}$. As \autoref{alg:path_flexible_form_finder}  returns $\epsilon$ after $k$ iterations,  $\Sigma = \Sigma_1 \cup \Sigma_2 \cup \cdots \cup \Sigma_k$ is a partition of $\Sigma$.
The goal of this section is to show that solving $\Pi$ requires $\Omega(n^{1/k})$ rounds.

\paragraph{Centered graphs.}
A \emph{radius-$t$ centered graph} is a pair $(G,v)$ where $v \in V$ is a node in $G=(V,E)$ so that all $u\in V$ are within distance $t$ to $v$, and each $u \in V$ whose distance to $v$ is exactly $t$ is permitted to have incident edges of the form $e = \{u, ?\}$, indicating that $e$ is an external edge that connects $u$ to some unknown node outside of $G$. As we only consider rooted trees, we assume that all edges are oriented towards the root, so that each node has outdegree at most $1$.

Observe that the view of a node $v$ after $t$ rounds of communication in $\LOCAL$ can be described by a radius-$t$ centered graph.
Therefore, a $t$-round $\LOCAL$ algorithm on $n$-node graphs is simply an assignment of a label $\sigma \in \Sigma$ to each radius-$t$ centered graph $(G,v)$ where each node in $G$ has a distinct $O(\log n)$-bit identifier.

\paragraph{Terminology.}
In this section, we use the term radius-$t$ \emph{view} of $v$ to denote the corresponding radius-$t$ centered graph, and the term radius-$t$ \emph{neighborhood} of $v$ to denote the set of nodes that are within distance $t$ to $v$. Note that the radius-$t$ {view} of $v$ contains more information than the subgraph induced by the  radius-$t$ neighborhood of $v$, as the  radius-$t$ {view} of $v$ includes information about the external edges.

\paragraph{Permissible labels.}
From now on, we fix $A$ to be any $\LOCAL$ algorithm that solves $\Pi=(\delta, \Sigma, \CC)$ in $t$ rounds on $n$-node graphs.
Given such an algorithm $A$, we say that a label $\sigma \in \Sigma$ is \emph{permissible} for $(G,v)$ if there exists some assignment of distinct $O(\log n)$-bit identifiers to the nodes in $G$ such that the output of $v$ is $\sigma$ when we run $A$ on $G$.

Using the notion of permissible labels, to show that $\Pi=(\delta, \Sigma, \CC)$ cannot be solved in $t$ rounds on $n$-node graphs, it suffices to find a graph $G=(V,E)$ that has at most $n$ nodes such that there exists a node $v \in V$ such that no label $\sigma \in \Sigma$ is permissible for the radius-$t$ centered graph $(H,v)$ corresponding to the radius-$t$ view of $v$ in $G$.
The following lemma is useful for showing that some label $\sigma \in \Sigma$ is not permissible for some  radius-$t$ centered graph $(G,v)$.

\begin{lemma}\label{lem:lb-aux}
Let\/ $A$ be a\/ $t$-round $\LOCAL$ algorithm that solves\/ $\Pi=(\delta, \Sigma, \CC)$ for\/ $n$-node rooted trees.
Let\/ $(G,v)$ be a fixed radius-$t$ centered graph.
Let\/ $\Sigma^\diamond \subseteq \Sigma$ be a subset of labels, and let\/ $\Pi'$ be the restriction of\/ $\Pi$ to\/ $\Sigma \setminus \Sigma^\diamond$.
Suppose there exists a number\/ $K$ such that for any\/ $d \geq K$,
we can construct a rooted tree\/ $T$ containing a directed path\/ $P= v_1 \leftarrow v_2 \leftarrow \cdots \leftarrow v_{d+1}$ meeting the following conditions.
\begin{enumerate}
        \item The radius-$t$ views of\/ $v_1$ and\/ $v_{d+1}$ are isomorphic to\/ $(G,v)$.
        \item Let\/ $S$ denote the set of nodes\/ $v_1, v_2, \ldots, v_{d+1}$ and their children. Then each\/ $\sigma \in \Sigma^\diamond$ is not permissible for the radius-$t$ views of each\/ $u \in S$.
        \item The radius-$(t+2)$ neighborhood of\/ $v_i$ contains at most\/ $n$ nodes, for each\/ $1 \leq i \leq d+1$.
\end{enumerate}
Then, for each\/ $\sigma \in \Sigma \setminus \Sigma^\diamond$, the following holds:
if\/ $\sigma$ is path-inflexible in\/ $\Pi'$, then\/ $\sigma$ is not permissible for\/ $(G,v)$.
\end{lemma}

\begin{figure}
	\centering
  \includegraphics[page=23,scale=\figscale]{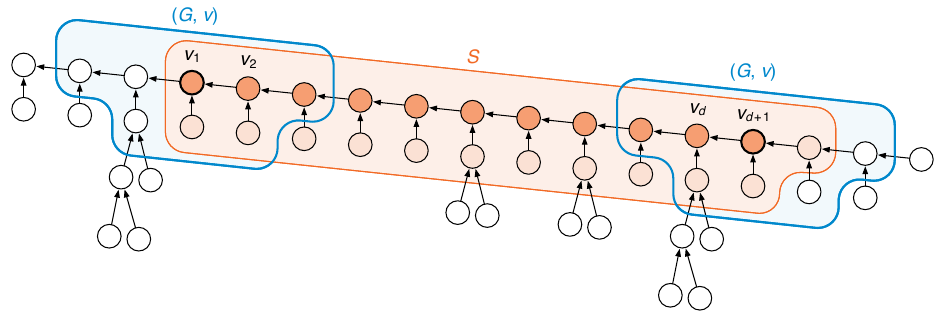}
	\caption{An illustration of \autoref{lem:lb-aux} with $t = 2$.}
	\label{fig:lb-aux}
  \Description{}
\end{figure}

See \autoref{fig:lb-aux} for an illustration of \autoref{lem:lb-aux}.
Before proving \autoref{lem:lb-aux}, let us give a brief, informal example of how we might apply it. We assume we have already established that algorithm $A$ cannot output labels from $\Sigma^\diamond$ in certain ``tricky'' radius-$t$ views. Tree $T$ is then constructed so that nodes of $S$ have tricky views, so algorithm $A$ is forced to solve the restriction $\Pi'$ of $\Pi$ around path $P$. Now if $\Pi'$ contains some path-inflexible labels, we can apply \autoref{lem:lb-aux} to rule out the possibility of using path-inflexible labels along path $P$, so we learn that the view $(G,v)$ is super-tricky, as it rules out not only $\Sigma^\diamond$, but also all path-inflexible labels of $\Pi'$. We can repeat this argument to discover many super-tricky views, by constructing different trees $T$.

This way we can start with a problem $\Pi_0$, and rule out the use of path-inflexible labels of $\Pi_0$ at least in some family of tricky views. Hence in those views we are, in essence, solving problem $\Pi_1$, which is the restriction of $\Pi_0$ to path-flexible labels. We repeat the argument, and rule out the use of path-inflexible labels of $\Pi_1$ in at least some family of super-tricky views, etc.

If we eventually arrive at an empty problem, we have reached a contradiction: algorithm $A$ cannot solve the original problem in some family of particularly tricky views. However, plenty of care will be needed to keep track of the specific family of views, as well as to make sure that we can still construct a suitable tree $T$ using only such views. We will get back to these soon, but this informal introduction will hopefully help to see why we first seek to prove this somewhat technical statement.

\begin{proof}[Proof of \autoref{lem:lb-aux}]
Fix a label $\sigma \in \Sigma \setminus \Sigma^\diamond$ such that $\sigma$ is path-inflexible in $\Pi'$. By the definition of path-flexibility, for any $K$, there exists an integer $d \geq K$ such that the following statement holds:
\begin{itemize}
    \item For any length-$d$
directed path
$P= v_1 \leftarrow v_2 \leftarrow \cdots \leftarrow v_{d+1}$ such that each node $v_i$ in $P$ is assigned a label $\lambda(v_i) \in \Sigma \setminus \Sigma^\diamond$, if the two end points $v_1$ and $v_{d+1}$ are labeled with $\sigma$,
then the labeling of $P$, interpreted as 1-ary tree, is not a valid solution for the path-form of $\Pi'$. More precisely, there must exist $1 \leq i \leq d+1$ such that $(\lambda(v_i) : \lambda(v_{i+1}))$ is not an allowed configuration in the path-form of $\Pi'$.
\end{itemize}

For the rest of the proof, we pick $d$ be a sufficiently large number such that the above statement holds. The precise choice of $d$ is to be determined. We consider a rooted tree $T$ that satisfies the lemma statement for this parameter $d$.
We assume that $\sigma$ is permissible for $(G,v)$, and then we will derive a contradiction.

Now consider the path $P$ in the lemma statement. Here all nodes in $P$ have exactly $\delta$ children, and all nodes of $P$ and their children can only be assigned labels from $\Sigma \setminus \Sigma^\diamond$ by $A$. Again, if $v_1$ and $v_{d+1}$ are labeled with $\sigma$, then there must exist $1 \leq i \leq d+1$ such that the node configuration of $v_i$ and its $\delta$ children is not an allowed configuration of $\Pi'$. Furthermore, it cannot be an allowed configuration for $\Pi$, either, as $\Pi'$ contains all configurations that consist of labels from $\Sigma \setminus \Sigma^\diamond$.

Consider any assignment of $O(\log n)$-bit identifiers to the nodes in $(G,v)$ that makes $A$ output $\sigma$, and apply this assignment to the radius-$t$ neighborhoods of $v_1$ and $v_{d+1}$ in $T$. Extend this identifier assignment to cover all nodes
that are within distance $t+2$ to some $v_i$ such that the radius-$(t+2)$ neighborhood of any $v_i$ does not contain repeated identifiers. This is possible because we assume that $T$ satisfies the property that the radius-$(t+2)$ neighborhood of $v_i$ contains at most $n$ nodes, for each $1 \leq i \leq d+1$, and because that we may choose $d \gg t$ to be sufficiently large so that the radius-$(t+2)$ neighborhood of each $u \in S$ cannot simultaneously intersect the radius-$(t+2)$ neighborhood of both $v_1$ and $v_{d+1}$.
Although some identifiers may appear several times
in $T$ and the total number of nodes in $T$ may exceed $n$. As we will later see, they are not problematic.

Consider the output labels of $v_i$ and its children, for $1 \leq i \leq d+1$, resulting from simulating $A$ on $T$. Note that $A$ by definition cannot output labels that are not permissible, and hence all of these nodes receive labels from $\Sigma \setminus \Sigma^\diamond$. Our choice of $d$ implies that there exists $1 \leq i \leq d+1$ such that the node configuration corresponding to $v_i$ and its children is not in~$\CC$. We take the subtree $T'$ induced by the union of the radius-$(t+1)$ neighborhood of  $v_i$ and its children. Since the radius-$(t+2)$ neighborhood of $v_i$ contains at most $n$ nodes, the rooted tree  $T'$ also contains at most $n$ nodes.
The output labelings of $v_i$ and its children due to simulating $A$ are the same in both $T$ and $T'$, as their radius-$t$ views are invariant of the underlying network being $T$ or $T'$. This violates the correctness of $A$, as  $T'$ contains at most $n$ nodes. Thus, $\sigma$ cannot be permissible for $(G,v)$.
\end{proof}

\paragraph{A hierarchical construction of rooted trees.}
 We first consider the following natural recursive construction of rooted trees. A \emph{bipolar tree} is a tree with two distinguished nodes $s$ and $t$, and it is also viewed as a rooted tree by setting $s$ as the root. The unique path connecting $s$ and $t$ is called the \emph{core path} of the bipolar tree. We consider the following operation $\bigoplus^x$.
\begin{itemize}
    \item Given a rooted tree $T$, define $\bigoplus^x T$ as the result of the following construction. Start with an $x$-node path $(v_1, v_2, \ldots, v_x)$.  Consider $x (\delta-1)$ copies  of $T$, indexed by two numbers $i$ and $j$: \[\bigl\{T_{i, j}'\bigr\}_{1 \leq i \leq x,\ 1 \leq j \leq \delta-1.}\]
    For $1 \leq i \leq x$ and $1 \leq j \leq \delta-1$, make the root of $T_{i, j}'$ a child of $v_i$ by adding an edge connecting them. Finally, set the two distinguished nodes of the resulting tree by $s = v_1$ and $t = v_x$.
\end{itemize}

\begin{figure}
\begin{center}
\includegraphics[page=16,scale=\figscale]{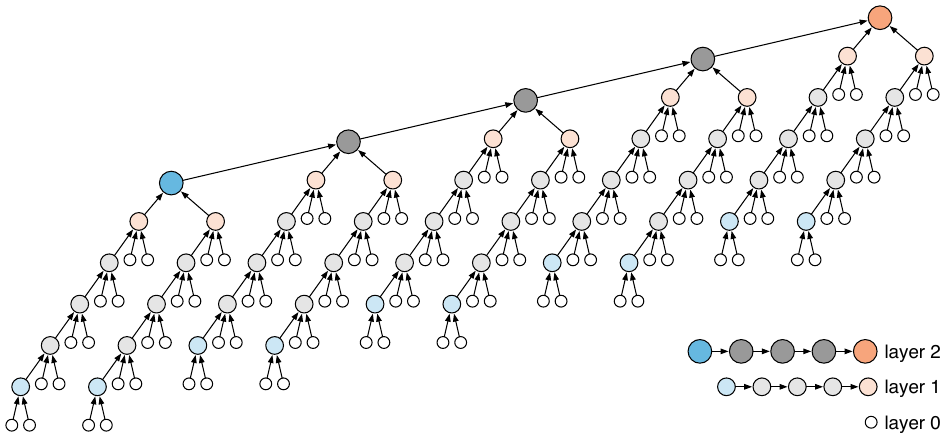}
\end{center}
\caption{Construction of the bipolar tree $T_k^x$, illustrated here for $\delta = 3$, $x = 5$, and $k = 2$.}\label{fig:bipolar-trees}
\Description{}
\end{figure}

Based on this operation, we construct a sequence of bipolar trees $T_0^x, T_1^x, \ldots, T_k^x$, where the nodes in $T_i^x$ are  partitioned into layers $0, 1, \ldots, i$; see \autoref{fig:bipolar-trees}:
\begin{itemize}
    \item For $i = 0$, define $T_0^x$ as the trivial bipolar tree consisting of only one isolated node $v$ with $s = v$ and $t = v$. We say that $v$ is in layer $0$.
    \item For $1 \leq i \leq k$, define  $T_i^x = \bigoplus^x T_{i-1}^x$. We say that all nodes in the core path of  $T_i^x$ are in layer $i$.
\end{itemize}

For any constant $\delta$, it is straightforward to see that the number of nodes in $T_k^x$ is $n = O(x^k)$, so $x = \Omega(n^{1/k})$. For each $1 \leq i \leq k$, the layer-$i$ nodes form paths consisting of exactly $x$ nodes. We call such an $x$-node path a layer-$i$ path.

The tree $T_k^x$ is analogous to the lower bound graph used in~\cite{CP19timeHierarchy} for establishing the tight $\Omega(n^{1/k})$ lower bound for some artificial $\LCL$ problem considered in~\cite{CP19timeHierarchy}.
The $\Omega(n^{1/k})$ lower bound proof in~\cite{CP19timeHierarchy} involves an argument showing that to solve the given $\LCL$ problem it is necessary that the two endpoints of a layer-$i$ path communicate with each other, and this costs at least $x = \Omega(n^{1/k})$ rounds.

When $x \geq 2$ and $1 \leq j \leq k$, there are three possible degrees in $T_j^x$: $1$, $\Delta-1 = \delta$, and $\Delta = \delta+1$. A node has degree $1$ if and only if it is in layer zero. A node has degree $\delta$ if and only if it is the root or it is the last node $v_x$ in some layer-$i$ path $v_1 \leftarrow v_2 \leftarrow \cdots \leftarrow v_x$. All the remaining nodes have degree exactly $\Delta = \delta +1$. Intuitively, the nodes with degree $\delta$ are in the boundary of the graph.

\paragraph{High-level ideas.}
To prove the $\Omega(n^{1/k})$ lower bound,  we will construct a sequence $S_1 \supseteq S_2 \supseteq \cdots \supseteq S_k$ of non-empty sets of radius-$t$ centered graphs, where $t = \Omega(n^{1/k})$. Each radius-$t$ centered graph used in our construction is isomorphic to the radius-$t$ view of some node $v$ in some graph of at most $n$ nodes.
By applying \autoref{lem:lb-aux} inductively with $\Sigma^\diamond = \Sigma_1 \cup \Sigma_2 \cup \cdots \cup \Sigma_{i-1}$ and $\sigma \in \Sigma_i$, we will show that for any given $t$-round algorithm $A$, the labels in $\Sigma_i$ are not permissible for the radius-$t$ centered graphs in $S_i$, for each $1 \leq i \leq k$.  Therefore, the given $\LCL$ problem $\Pi=(\delta, \Sigma, \CC)$ cannot be solved in $t$ rounds.

The construction of $S_1 \supseteq S_2 \supseteq \cdots \supseteq S_k$ requires somewhat complicated definitions. To motivate these forthcoming definitions, we begin with describing a natural attempt to prove the  $\Omega(n^{1/k})$ lower bound directly using the trees $T_k^x$ and see why it does not work.

Suppose that the given $\LCL$ problem $\Pi=(\delta, \Sigma, \CC)$ can be solved in $t=o(n^{1/k})$ rounds on $n$-node graphs by an algorithm $A$. We pick $x$ to be a sufficiently large number such that  $x = \Theta(t)$  and the number of nodes in $T_k^x$ is at most $O(x^k) < n$.
Recall that $\Sigma_1$ is the set of path-inflexible labels for the original $\LCL$ problem $\Pi=(\delta, \Sigma, \CC)$. Let $P = v_1 \leftarrow v_2 \leftarrow \cdots \leftarrow v_{x}$ be a layer-$1$ path in $T_k^x$. It is straightforward to see that for each $1+t < j < x-t$, the radius-$t$ view of each $v_j$ is identical. Let $(G,v)$ denote the corresponding radius-$t$ centered subgraph. By the path inflexibility of the labels in $\Sigma_1$, all labels $\sigma \in \Sigma_1$ are not permissible for $(G,v)$. Intuitively, this is because that we can find in $T_k^x$ a path connecting two views isomorphic to $(G,v)$ with a flexible path length. Similarly, we can apply the same argument for layer-$\ell$ paths for each $1 \leq \ell \leq k$, so we infer  that the labels in $\Sigma_1$ cannot be used to label nodes in layer $1$ or above.

For the inductive step, suppose that we already know that $\Sigma_1, \Sigma_2, \ldots \Sigma_{i-1}$ cannot be used to label nodes in layer $i-1$ or above. We consider the $\LCL$ problem that is the restriction of $\Pi$ to the set of labels $\Sigma_i \cup \Sigma_{i+1} \cup \cdots \cup \Sigma_k$. The above argument still works if we replace $\Sigma_1$ by $\Sigma_i$ and only consider the layers $i \leq \ell \leq k$, as we recall that $\Sigma_i$ is the set of path-inflexible labels when we restrict to the set  of labels $\Sigma_i \cup \Sigma_{i+1} \cup \cdots \cup \Sigma_k$.

It appears that this approach allows us to show that for each $1 \leq \ell \leq k$, the layer-$\ell$ nodes cannot be labeled using $\Sigma_1, \Sigma_2, \ldots, \Sigma_{\ell}$, so the given algorithm $A$ cannot produce any output label for the layer-$k$ nodes, contradicting the correctness of $A$. This approach, however, has one issue. Consider again the layer-$\ell$ path $P = v_1 \leftarrow v_2 \leftarrow \cdots \leftarrow v_{x}$ in $T_k^x$ in the above discussion. We are only able to show that the labels in $\Sigma_1$ cannot be used to label the nodes $v_j$ for each $1+t < j < x-t$, as the radius-$t$ view of the remaining nodes in $P$ are different. This is problematic because in the next level of induction, when we try to show that the labels in $\Sigma_2$ cannot be used to label some node $v$ that is in layer $2$ or above, the proof relies on the condition that the labels in $\Sigma_1$ cannot be used to label $v$ \emph{and its children}; see \autoref{lem:lb-aux} and its proof. In particular,  showing that $\Sigma_1$ cannot be used to label the middle nodes in $P$ whose radius-$t$ views are identical is not enough.

To resolve this issue, we need to consider essentially all possible radius-$t$ centered graph $(G,v)$ corresponding to a radius-$t$ view of a layer-$i$ node, and we have to make sure that for any sufficiently large number $d$, we can find a rooted tree $T$ that contains a length-$d$ directed path $P= v_1 \leftarrow v_2 \leftarrow \cdots \leftarrow v_{d+1}$ such that the radius-$t$ views of the two endpoints $v_1$ and $v_{d+1}$ are isomorphic to  $(G,v)$ and  all the intermediate nodes $v_2, v_3, \ldots, v_{d}$ are in layer $i$ or above, so that \autoref{lem:lb-aux} is applicable.

\begin{figure}
\begin{center}
\includegraphics[page=17,scale=\figscale]{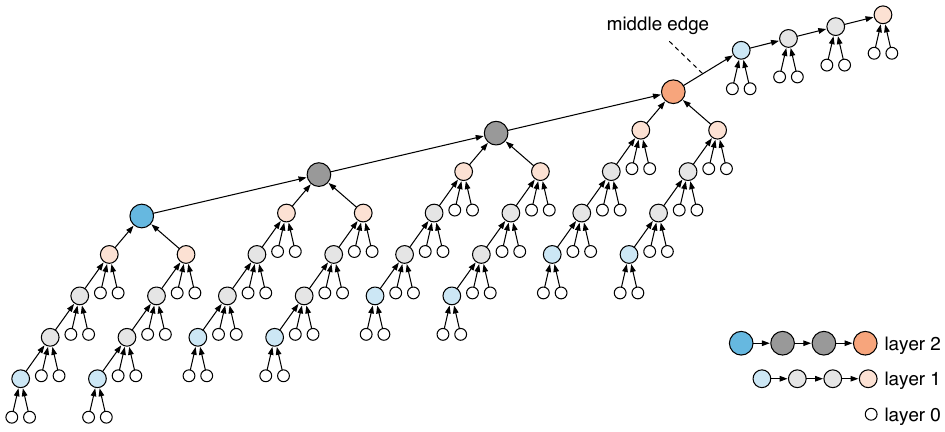}
\end{center}
\caption{Tree $T_{i \leftarrow j}^x$, for $\delta = 3$, $x = 4$, $j = 2$, and $i = 1$.}\label{fig:tree-concatenate}
\Description{}
\end{figure}

To deal with the views $(G,v)$ that do not belong to the central part of the long paths, we will need to concatenate two trees $T_i^x$ and $T_j^x$ for some $1 \leq i \leq k$ and $1 \leq j \leq k$ to obtain directed paths starting and ending with the same view $(G,v)$, so that we can apply \autoref{lem:lb-aux}. Such a concatenation will create new views that did not exist before in $T_k^x$. In order to capture all such views, we will consider the following definition $T_{i \leftarrow j}^x$ and build the argument around it; see \autoref{fig:tree-concatenate}.

\begin{itemize}
    \item For $1 \leq i \leq k$ and $1 \leq j \leq k$, define $T_{i \leftarrow j}^x$ as the result of the following construction. Let $T_1' = T_i^x$ (distinguished nodes are $s_1$ and $t_1$) and $T_2' = T_j^x$ (distinguished nodes are $s_2$ and $t_2$). Concatenate these two bipolar trees into a new bipolar tree by adding an edge $\{t_1, s_2\}$ and setting $s = s_1$ and $t = t_2$. We call $e = \{t_1, s_2\}$  the \emph{middle edge}.  The layer numbers of the nodes are kept when $T_i^x$ and $T_j^x$ are linked together into $T_{i \leftarrow j}^x$.
\end{itemize}

We make the following two  observations. For the special case of $i=j$, $T_{i \leftarrow i}^x$ is simply  $\bigoplus^{2x} T_{i-1}^x$.
For any number $t$, the number of nodes in the radius-$t$ neighborhood of any node in $T_{i \leftarrow j}^x$ is $O(t^{\max\{i,j\}}) = O(t^k)$, regardless of $x$.

\paragraph{A sequence of sets of radius-$t$ centered graphs.}
Now, we are ready to define the set of radius-$t$ centered graphs $S_i$, for each $1 \leq i \leq k$. In the definition of  $S_i$, we let $x$ be any integer such that $x \geq 2t+2$. It will be clear from the construction of $S_i$ that the definition of  $S_i$  is invariant of the choice of $x$, as long as $x$ is sufficiently large comparing with $t$.

The set $S_i$ consists of all radius-$t$ centered graphs $(G,v)$ such that there exists a node $u$ in the rooted tree $T_{a \leftarrow b}^x$ for some $a$ and $b$ meeting the following conditions.
\begin{itemize}
\item $i \leq a \leq k$ and $i \leq b \leq k$.
    \item The radius-$t$ view of $u$ is isomorphic to $(G,v)$.
    \item The radius-$t$ view of $u$ contains at least one node in the middle edge $e$ of $T_{a \leftarrow b}^x$.
    \item The layer number of $u$ is at least $i$.
\end{itemize}

Note that the threshold $x \geq 2t+2$ is chosen to make sure that for each node $u'$ in the radius-$t$ neighborhood of $u$, if $u'$ is not in layer zero, then its degree is exactly  $\Delta = \delta+1$, that is, $u'$ has one parent and $\delta$ children.

It is clear from the above definition of $S_i$ that we have $S_1 \supseteq S_2 \supseteq \cdots \supseteq S_k \neq \emptyset$.
Before we proceed, we prove a result showing that $S_i$ includes essentially all radius-$t$ view for layer-$i$ nodes in $T_i^x$. Formally, for each $1 \leq i \leq k$, we define $S_i^\ast$ as the set of all radius-$t$ centered graphs $(G,v)$ meeting the following conditions.
\begin{itemize}
    \item There exist $x \geq 1$, $i \leq j \leq k$, and a layer-$i$ node $u$ in $T_j^x$ such that the radius-$t$ view of $u$ in  $T_j^x$ is isomorphic to $(G,v)$. Furthermore, for each node $u'$ in the radius-$t$ neighborhood of $u$, if $u'$ is not in layer zero, then its degree is exactly  $\Delta = \delta+1$.
\end{itemize}

Intuitively, $S_i^\ast$ is the set of all possible radius-$t$ views for layer-$i$ nodes, excluding those near the boundary.
We exclude the views involving boundary nodes because we want to focus on the interior part of the graph where all nodes have the same degree $\Delta = \delta +1$, except the layer-0 nodes whose degree is always one.

\begin{lemma}\label{lem:lb-aux-6}
For each $1 \leq i \leq k$, we have $S_i^\ast \subseteq S_i$.
\end{lemma}
\begin{proof}
Consider the node $u$ in the graph $T_j^x$ in the definition of $S_i^\ast$.
Since  the radius-$t$ neighborhood of $u$ does not include any non-leaf node whose degree smaller than $\Delta = \delta+1$, we may assume that $x$ is an arbitrarily large number.
Let $u^\ast$ be any node in the radius-$t$ neighborhood of $u$ that has the highest layer number. Let $i^\ast$ be the layer number of  $u^\ast$. We have $i \leq i^\ast \leq j \leq k$. The radius-$t$ neighborhood of $u$ is confined to some subgraph $T_{i^\ast}^{x}$ of $T_{j}^{x}$ where $u^\ast$ lies on the core path of $T_{i^\ast}^{x}$. The graph $T_{i^\ast}^{x}$ can be viewed as a subgraph of  $T_{i^\ast \leftarrow i^\ast}^{x}$  such that $u^\ast$ is a node in the middle edge of  $T_{i^\ast \leftarrow i^\ast}^{x}$. As $u$ is within distance $t$ to $u^\ast$ and the radius-$t$ view of $u$ in $T_{i^\ast \leftarrow i^\ast}^{x}$, $T_{i^\ast}^{x}$, and the original graph $T_{j}^{x}$ are identical, we conclude that the radius-$t$ view of $u$ is isomorphic to some member in $S_{i}$ by considering the graph $T_{i^\ast \leftarrow i^\ast}^{x}$.
\end{proof}

\paragraph{The lower bound proof.}
For any given integer $t$, we pick $n$ to be the maximum number of nodes in the radius-$(t+2)$ neighborhood of any node in $T_{i \leftarrow j}^x$, over all choices of $i$, $j$, and $x$ such that $1 \leq i \leq k$, $1 \leq j \leq k$, and $x \geq 1$. It is clear that $n = O(t^k)$, or equivalently $t = \Omega(n^{1/k})$.
Therefore, to prove an $\Omega(n^{1/k})$ lower bound for the given problem $\Pi$, it suffices to show the non-existence of a $t$-round algorithm ${A}$ that solves $\Pi$ on $n$-node graphs.

Suppose such an algorithm $A$ exists. In \autoref{lem:lb-aux-2}, whose proof is deferred, we will prove by induction that all labels in $\Sigma_i$ are not permissible for all centered graphs in $S_i$, for each $1 \leq i \leq k$. In particular, this means that all labels in $\Sigma$ are not permissible for all centered graphs in $S_k$, as $\Sigma = \Sigma_1 \cup \Sigma_2 \cup \cdots \cup \Sigma_k$ and $S_1 \supseteq S_2 \supseteq \cdots \supseteq S_k$.

\begin{lemma}\label{lem:lb-aux-2}
For each $1 \leq j \leq k$, all labels in $\Sigma_j$ are not permissible for all centered graphs in $S_j$.
\end{lemma}

We now prove the main result of this section assuming \autoref{lem:lb-aux-2}.

\begin{lemma}\label{lem-no-path-flexible-lb}
If
\autoref{alg:path_flexible_form_finder}  returns $\epsilon$ after $k$ iterations, then $\Pi$ requires   $\Omega(n^{1/k})$  rounds to solve.
\end{lemma}
\begin{proof}
In view of the above discussion, it suffices to show that the algorithm $A$ considered above does not exist.
Recall that $S_k \neq \emptyset$, and each $(G,v) \in S_k$ is isomorphic to the radius-$t$ view of some node $u$ in  $T_{a \leftarrow b}^x$ with  $1 \leq i \leq k$, $1 \leq j \leq k$, and $x \geq 2t+2$. Furthermore, the radius-$(t+2)$ neighborhood of $u$ contains at most $n$ nodes. If we run $A$ on the subgraph  induced by the radius-$(t+2)$ neighborhood of $u$ in  $T_{a \leftarrow b}^x$, then according to \autoref{lem:lb-aux-2} the algorithm  $A$ does not output any label for $u$, violating the correctness of $A$, so such an algorithm  $A$ does not exist.
\end{proof}

It is clear that \autoref{lem-no-path-flexible-lb} implies \autoref{no-path-flexible-impl-n-lb}.

\paragraph{Constructing a rooted tree $T$ for applying \autoref{lem:lb-aux}.}
For the rest of the section, we prove \autoref{lem:lb-aux-2}. We begin with describing the construction of the rooted tree $T$ needed for applying \autoref{lem:lb-aux} in the proof of  \autoref{lem:lb-aux-2}; see \autoref{fig:tree-T} for an illustration.

\begin{figure}
\begin{center}
\includegraphics[page=18,scale=\figscale]{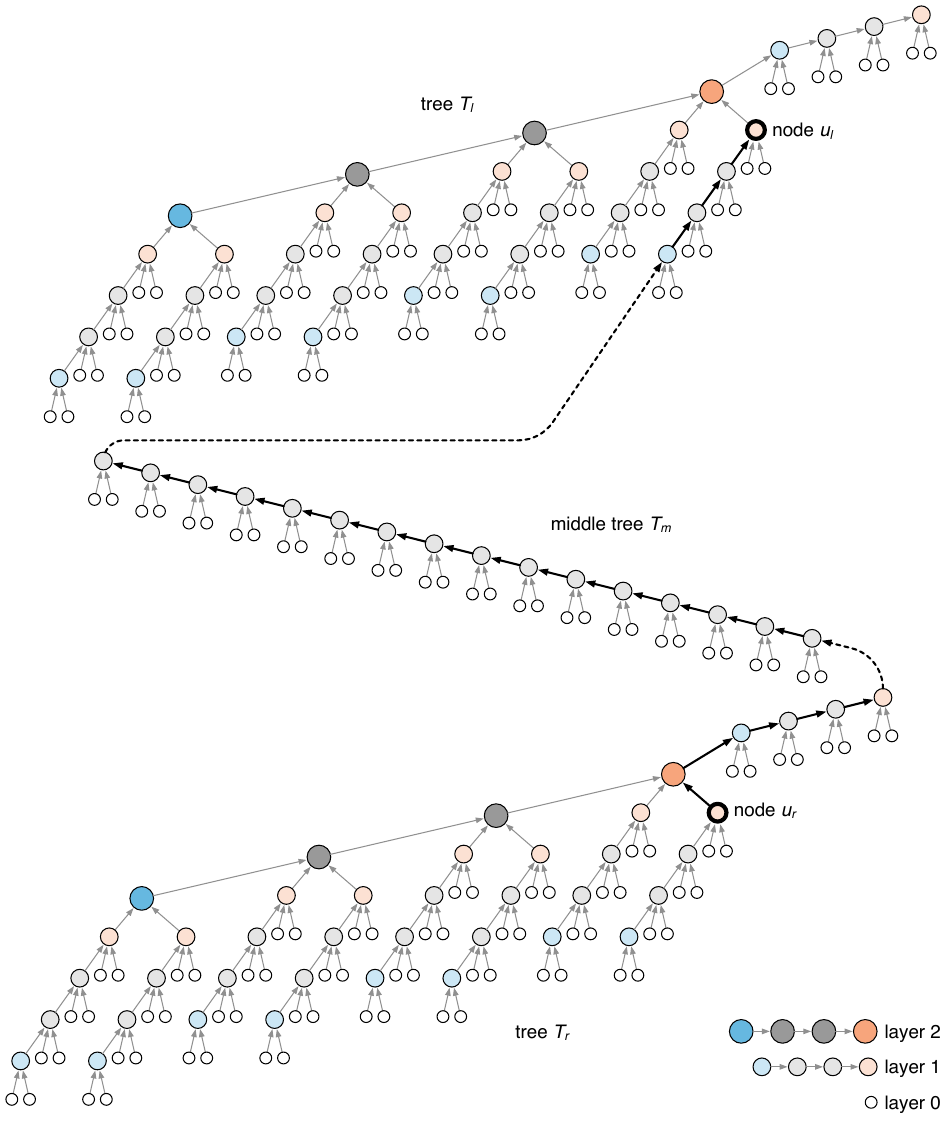}
\end{center}
\caption{Constructing a rooted tree $T$ for applying \autoref{lem:lb-aux}. Here $(G,v)$ was isomorphic to the view of some node $u$ in layer $1$ of $T_{2 \leftarrow 1}^x$. Therefore we construct two copies of $T_{2 \leftarrow 1}^x$, one of them is called $T_l$ and the other one is $T_r$, and we identify the nodes $u_l$ and $u_r$ that have views isomorphic to $(G,v)$. We identify the unique path $P_r$ from $u_r$ to the root of $T_r$ and the unique layer-$1$ path $P_l$ that takes us to $u_l$ (dark arrows). Finally, we connect $P_r$ through the middle tree to $P_l$. Note that the resulting path $P = P_l \leftarrow P_m \leftarrow P_r$ does not use layer-$0$ nodes.}\label{fig:tree-T}
\Description{}
\end{figure}

The construction of $T$ is parameterized by any $(G,v) \in S_j$, for any $1 \leq j \leq k$. Recall from the definition of $S_j$ that $(G,v)$ is isomorphic to the radius-$t$ view of some layer-$i$ node $u$ in $T_{a \leftarrow b}^x$ such that $1 \leq a \leq k$, $1 \leq b \leq k$, and $j \leq i \leq \min\{a,b\}$, and this radius-$t$ neighborhood contains at least one node in the middle edge $e$ of $T_{a \leftarrow b}^x$.

From now on we fix $x = 2t+4$.
Then $D = (k+1)x-1$ is an upper bound on the length of any root-to-leaf path in $T_{a \leftarrow b}^x$, for any $1 \leq a \leq k$ and $1 \leq b \leq k$. We define $K = 2D+x+1$.

The construction of $T$ is also parameterized by a distance parameter $d$ such that $d \geq K$. In the rooted tree $T$ that we construct, there will be a length-$d$ directed path $P = v_1 \leftarrow v_2 \leftarrow \cdots \leftarrow v_{d+1}$ satisfying some good properties to make \autoref{lem:lb-aux}  applicable.

Intuitively, $T$ will be the result of concatenating two copies $T_l$ and $T_r$ of $T_{a \leftarrow b}^x$ via a middle tree $T_m = \bigoplus^{y} T_{i-1}^{x}$, and then $P$ will be the unique directed path in $T$ connecting the two copies of $u$ in $T_l$ and $T_r$. Note that $T_m = \bigoplus^{y} T_{i-1}^{x}$ is simply a variant of $T_i^x =  \bigoplus^{x} T_{i-1}^{x}$ such that the length of the core path is $y$ instead of $x$. We select $y$ to make the length of $P$ equals $d$. The points of concatenation will be selected to ensure that all nodes in $P$ are in layer $i$ or above.
Formally, the construction of the rooted tree $T$ and its length-$d$ path $P$ is as follows.

\begin{description}
    \item[The two trees $T_l$ and $T_r$.] Recall that $u$ is a layer-$i$ node in $T_{a \leftarrow b}^x$ whose radius-$t$ neighborhood  contains at least one node in the middle edge. We consider two copies of $T_{a \leftarrow b}^x$, called $T_l$ and $T_r$. The two copies of $u$ in $T_l$ and $T_r$ are called $u_l$ and $u_r$. Similarly, we write $(s_l, t_l)$ and $(s_r, t_r)$ to denote the two distinguished nodes of $T_l$ and $T_r$.
    \item[The path $P_l$.] If $u_l$ is on the core path of $T_l$, then we define $P_l$ to be the unique directed path $u_l \leftarrow \cdots \leftarrow t_l$. Otherwise, then consider the unique layer-$i$ path $w_1 \leftarrow \cdots \leftarrow w_x$ that contains  $u_l$, and then we define $P_l$ to be the unique directed path $u_l \leftarrow \cdots w_x$. Observe that all nodes in $P_l$ are in layer $i$ or above.
    \item[The path $P_r$.] We define $P_r$ to be the unique directed path $s_r \leftarrow \cdots \leftarrow u_r$ in $T_r$. Observe that all nodes in $P_r$ are in layer $i$ or above.
    \item[The middle tree $T_m$ and its path $P_m$.] Let $D_l$ and $D_r$ denote the lengths of $P_l$ and $P_r$. Note that $D_l \leq D$ and $D_r \leq D$. We define $T_m = \bigoplus^{y} T_{i-1}^{x}$, where $y = d - D_l - D_r - 1$, and we define $P_m$ as the core path of $T_m$. Note that we must have $y \geq x$, due to the assumption $d \geq K$ and our choice of $K=2D+x+1$. Clearly, all nodes in $P_m$ are in layer $i$.
    \item[Concatenation.] Now, we are ready to define the rooted tree $T$ and its associated length-$d$ directed path $P$. We construct the directed path $P$ by adding two edges to concatenate the three paths $P_l$, $P_m$, and $P_r$ together: $P = P_l \leftarrow P_m \leftarrow P_r$. The length of $P$ is exactly $d$ due to our choice of $y = d - D_l - D_r - 1$. The rooted tree $T$ is the result of this concatenation of $T_l$, $T_m$, and $T_r$.
\end{description}

The radius-$t$ views of $u_l$ is isomorphic to $(G,v)$, regardless of the underlying graph being $T$ or $T_l$.
Similarly, the radius-$t$ views of $u_r$ is isomorphic to $(G,v)$, regardless of the underlying graph being $T$ or $T_r$.
Hence the radius-$t$ neighborhoods of the two endpoints of $P$ in $T$ are isomorphic to the given radius-$t$ centered graph $(G,v)$.
We also note that all nodes in $P$ are in layer $i$ or above, so all children of nodes in $P$ are in layer $i-1$ or above.

Next, we will prove some additional properties of $T$ and $P$. We begin with  \autoref{lem:lb-aux-3} and \autoref{lem:lb-aux-4}. Informally, in these lemmas we show that the local view seen from an edge connecting $T_l$, $T_m$, and $T_r$ is isomorphic to the local view seen from the middle edge $e$ of $T_{a' \leftarrow b'}^x$, for some choices of $i \leq a' \leq k$ and $i \leq b' \leq k$.

\begin{lemma}\label{lem:lb-aux-3}
Let $e_l= u' \leftarrow v'$ be the edge connecting $T_l$ and $T_m$. Let $U_l$ be the union of the radius-$(x-1)$ neighborhood of $u'$ and $v'$ in $T$. There is a subgraph $T_l'$ of  \ $T$ isomorphic to $T_{a' \leftarrow b'}^x$ for some $i \leq a' \leq k$ and $b' = i$ such that $T_l'$ contains all  nodes in $U_l$. In the isomorphism, the  edge $e_l$ is mapped to the middle edge $e$ of $T_{a' \leftarrow b'}^x$.
\end{lemma}
\begin{proof}
Let $a'$ be the layer number of $u'$. Note that we have either $a' = i$ or $a' = b$. In any case,  $i \leq a' \leq k$.
Consider the $2x$-node path $u_1' \leftarrow \cdots \leftarrow u_x' = u' \leftarrow v' = v_1' \leftarrow \cdots \leftarrow v_x'$  in $T$ defined as follows.
\begin{itemize}
    \item $e_l= u' \leftarrow v'$ is the edge connecting $T_l$ and $T_m$.
    \item $u_1' \leftarrow \cdots \leftarrow u_x'$ is the unique layer-$a'$ path in $T_l$ containing $u'$.
    \item $v_1' \leftarrow \cdots \leftarrow v_x'$ is the path formed by the first $x$ nodes in $P_m$.
\end{itemize}
 We consider the subgraph $T_{a'}^x$ induced by the nodes  $u_1' \leftarrow \cdots \leftarrow u_x'$  and their descendants in $T_l$. As $v_1' \leftarrow \cdots \leftarrow v_x'$ are the first $x$ nodes in the $y$-node core path of $T_m = \bigoplus^{y} T_{i-1}^{x}$,  the nodes $v_1' \leftarrow \cdots \leftarrow v_x'$  and their descendants induce a subgraph $T_{b'}^x$ with $b' = i$.
We choose $T_l'$ to be the union of these two subgraphs $T_{a'}^x$ and  $T_{b'}^x$, together with the edge $e_l= u' \leftarrow v'$. It is clear that $T_l'$ is isomorphic to $T_{a' \leftarrow b'}^x$ and contains all nodes in $U_l$.
\end{proof}

\begin{lemma}\label{lem:lb-aux-4}
Let $e_r= u' \leftarrow v'$ be the edge connecting $T_m$ and $T_r$. Let $U_r$ be the union of the radius-$(x-1)$ neighborhood of $u'$ and $v'$ in $T$. There is a subgraph $T_r'$ of \ $T$ isomorphic to $T_{a' \leftarrow b'}^x$ for  $a' = i$ and $b' = a$ such that $T_r'$ contains all  nodes in $U_r$. In the isomorphism, the  edge $e_r$ is mapped to the middle edge $e$ of $T_{a' \leftarrow b'}^x$.
\end{lemma}
\begin{proof}
Recall that $T_m = \bigoplus^{y} T_{i-1}^{x}$ with $y \geq x$ and $T_r = T_{a \leftarrow b}^x$ is formed by connecting $T_a^x$ and $T_b^x$. We write $v_1' \leftarrow \cdots \leftarrow v_y'$ to denote the core path of $T_m$, and we let $T' = \bigoplus^{x} T_{i-1}^{x} = T_i^x$ be a subtree of $T_m$ induced by the $x$-node subpath $v_{y-x+1}' \leftarrow \cdots \leftarrow v_y'$ and the descendants of the nodes in this subpath.

The edge  $e_r= u' \leftarrow v'$  connects the two trees $T' = T_i^x$ and $T_a^x$, as $u'$ is the distinguished node $t$ of $T' = T_i^x$ and $v'$ is the distinguished node $s$ of $T_a^x$. Therefore, we may take $T_r'$ to be the union of $T' = T_i^x$ and $T_a^x$, together with the edge  $e_r= u' \leftarrow v'$. The tree $T_r'$ is isomorphic to $T_{i \leftarrow a}^x$ and contains all nodes in $U_r$.
\end{proof}

Combining \autoref{lem:lb-aux-3} and \autoref{lem:lb-aux-4}, in \autoref{lem:lb-aux-5} we show that the local neighborhood of any node in $T$ is isomorphic to the local neighborhood of some node in $T_{a' \leftarrow b'}^{x'}$, for some choices of $1 \leq a' \leq k$, $1 \leq b' \leq k$ and $x' \geq 1$. In the proof of \autoref{lem:lb-aux-5} we utilizes the fact that $x = 2t+4$.

\begin{lemma}\label{lem:lb-aux-5}
For each node $w$ in $T$, the subgraph induced by its radius-$(t+2)$ neighborhood is isomorphic to the subgraph induced by the radius-$(t+2)$ neighborhood of some node $w'$ in $T_{a' \leftarrow b'}^{x'}$ for some $1 \leq a' \leq k$, $1 \leq b' \leq k$ and $x' \geq 1$.
\end{lemma}
\begin{proof}
The proof is done by a case analysis.
We write $U$ to denote the set of nodes within the radius-$(t+2)$ neighborhood of $w$ in $T$. If $U$ is completely confined in one of $T_l$ or $T_r$, then the lemma  holds with $T_{a' \leftarrow b'}^{x'} = T_{a \leftarrow b}^{x}$, as both $T_l$ and $T_r$ are isomorphic to $T_{a \leftarrow b}^{x}$. If $U$ is completely confined in $T_m$, then the lemma  holds with $T_{a' \leftarrow b'}^{x'}  = T_{i-1 \leftarrow i-1}^{y}$, as $T_m = \bigoplus^{y} T_{i-1}^{x}$ is a subgraph of $T_{i-1 \leftarrow i-1}^{y}$, as we recall that $y \geq x$.

If $U$ contains the edge $e_l$ connecting $T_l$ and $T_m$, then $U \subseteq U_l$, where $U_l$ is defined in \autoref{lem:lb-aux-3}. Note that the fact that $x=2t+4$ is used to show that $U \subseteq U_l$.
Therefore, the lemma  holds with the tree $T_{a' \leftarrow b'}^x$ considered in \autoref{lem:lb-aux-3}.

Finally, the remaining case is that $U$ contains the edge $e_r$ connecting $T_m$ and $T_r$. Similar to the previous case, using \autoref{lem:lb-aux-4} we obtain that  $U \subseteq U_r$, so  the lemma  holds with the tree $T_{a' \leftarrow b'}^x$ considered in \autoref{lem:lb-aux-4}.
\end{proof}

Same as the notation used in \autoref{lem:lb-aux}, for the rest of the section, we write $S$ to denote the set of the nodes in $P$ and their children.
Using \autoref{lem:lb-aux-6}, \autoref{lem:lb-aux-3}, and \autoref{lem:lb-aux-4}, we prove \autoref{lem:lb-aux-7}, which shows that the radius-$t$ view of each node in $S$ belongs to $S_{i-1}$.

\begin{lemma}\label{lem:lb-aux-7}
If $i > 1$, then the radius-$t$ view of each node in $S$ belongs to $S_{i-1}$.
\end{lemma}
\begin{proof}
Let $w \in S$.
Let $i'$ be the layer number of $w$.
From the construction of $P$ we already know that all nodes on the path $P$ has layer number at least $i$, so their children have layer number at least $i-1$, and so $i' \geq i-1$.

We first consider the case where the radius-$t$ neighborhood of $w$ contains a node in the edge $e_l$ connecting $T_l$ and $T_m$. Then $w$ has the same radius-$t$ view in both $T$ and the graph $T_{a' \leftarrow b'}^x$ considered in \autoref{lem:lb-aux-3}.
Since $i' \geq i-1$, $a' \geq i > i-1$, $b' \geq i > i-1$, and $w$ is within distance $t$ to a node in the middle edge of $T_{a' \leftarrow b'}^x$, this radius-$t$ view belongs to $S_{i-1}$ by its definition.
The case of where the radius-$t$ neighborhood of $w$ contains a node in the edge $e_r$ connecting $T_m$ and $T_r$ can be handled using \autoref{lem:lb-aux-4} similarly.

From now on, we assume that  the radius-$t$ neighborhood of $w$ does not contain any node in $e_l$  and  $e_r$. There are three cases depending on whether the radius-$t$ neighborhood of $w$ is confined to $T_l$, $T_m$, or $T_r$.

Consider the case where  the  radius-$t$ neighborhood of $w$ is confined to $T_m$. Since $x$ is sufficiently large, the radius-$t$ neighborhood of $w$ does not contain any non-leaf node whose degree is not $\Delta = \delta+1$. Observe that $T_m = \bigoplus^{y} T_{i-1}^{x}$ is a subgraph of $T_i^y$ as $y \geq x$, so
the radius-$t$ view of $w$ is the same in $T$, $T_m$, and $T_{i}^y$, and so this radius-$t$ view  belongs to $S_{i'}^\ast$. By \autoref{lem:lb-aux-6}, we have $S_{i'}^\ast \subseteq S_{i'}$. We also have $S_{i'} \subseteq S_{i-1}$ because $i' \geq i-1$. Hence we conclude that this radius-$t$ view  belongs to $S_{i-1}$, as desired.

For the rest of the proof, we consider the case  where  the  radius-$t$ neighborhood of $w$ is confined to $T_l$, as the case of $T_r$ is similar. Recall that $T_l = T_{a \leftarrow b}^{x}$ is constructed by concatenating $T_a^x$ and $T_b^x$ by a middle edge $e$. If the radius-$t$ neighborhood of $w$ contains a node of $e$, then we know that this radius-$t$ view belongs to $S_{i-1}$, as we have $a \geq i > i-1$, $b \geq i > i-1$, and $i' \geq i-1$.
Otherwise, the radius-$t$ neighborhood of $w$ is confined to either $T_a^x$ or $T_b^x$. Similarly, we may use \autoref{lem:lb-aux-6} to show that the radius-$t$ view of $w$ is in $S_{i-1}$.
\end{proof}

Using \autoref{lem:lb-aux-5} and \autoref{lem:lb-aux-7}, we are now ready to prove  \autoref{lem:lb-aux-2}.

\begin{proof}[Proof of \autoref{lem:lb-aux-2}]
By induction hypothesis, suppose that the lemma statement holds for smaller $j$-values. Fix any $(G,v) \in S_j$. Then $(G,v)$ is isomorphic to the radius-$t$ neighborhood of a   layer-$i$ node $u$ in $T_{a \leftarrow b}^x$ such that $j \leq i \leq \min\{a,b\}$ and this radius-$t$ neighborhood  contains at least one node in the middle edge $e$ of $T_{a \leftarrow b}^x$.

Given $T_{a \leftarrow b}^x$ and $u$, construct the rooted tree $T$ and its directed path $P$ as we discuss above. Remember in our construction there is a number $K$ such that for each $d \geq K$, we are able to make $d$ the length of $P$.

Consider $\Sigma^\diamond = \Sigma_1 \cup \Sigma_2 \cup \cdots \cup \Sigma_{j-1}$. Fix any $\sigma \in \Sigma_j$. Recall that $\Sigma_j$ is the set of path-inflexible labels for the restriction of $\Pi$ to $\Sigma \setminus \Sigma^\diamond$.  To prove the lemma, it suffices to show that $\sigma$ is not permissible for $(G,v)$.

We apply \autoref{lem:lb-aux} with the rooted tree $T$ and its directed path $P$ with $\Sigma^\diamond$. We will see that the properties of $T$ and $P$  that we discuss above imply that the three conditions of  \autoref{lem:lb-aux} are met.
Condition~(1) follows immediately from the construction of $T$.
For Condition~(2), if $j = 1$, then $\Sigma^\diamond = \emptyset$, so Condition~(2) trivially holds; if $j > 1$, then $i \geq j > 1$, so we may apply \autoref{lem:lb-aux-7} to obtain that for each node $w \in S$, its radius-$t$ neighborhood in $T$ is in $S_{i-1} \subseteq S_{j-1}$. Therefore, by induction hypothesis, we know that each $\sigma' \in \Sigma^\diamond$ is not permissible for the radius-$t$ view of each $w \in S$, so Condition~(2) holds.
Condition~(3) follows from \autoref{lem:lb-aux-5} that the radius-$(t+2)$ neighborhood of each node in $P$ is isomorphic to the radius-$(t+2)$ neighborhood of some node in $T_{a' \leftarrow b'}^{x'}$ for some $1 \leq a' \leq k$, $1 \leq b' \leq k$ and $x' \geq 1$, and our choice of $n$ guarantees that the  radius-$(t+2)$ neighborhood of any node in $T_{a' \leftarrow b'}^{x'}$ cannot contain more than $n$ nodes.
Hence \autoref{lem:lb-aux} is applicable, so $\sigma$ is not permissible for $(G,v)$.
\end{proof}

\section{Sublogarithmic region}\label{sec:sublog}

In this section we prove that there is no $\LCL$ problem $\Pi$ with distributed time complexity between $\omega(\log^* n)$ and $o(\log n)$. Also, we prove that, given a problem $\Pi$, we can \emph{decide} if its complexity is $O(\log^* n)$ or $\Omega(\log n)$. Moreover, we prove that randomness cannot help: if a problem has \emph{randomized} complexity $O(\log^* n)$, then it has the same \emph{deterministic} complexity.

\subsection{High-level idea}

Informally, we prove that all problems that are $O(\log^* n)$ solvable can be solved in a normalized way, that is the following:
\begin{itemize}
	\item Split the rooted tree in constant size rooted subtrees, where each root has some minimum distance from the leaves. Note that each leaf is the root of another subtree.
	\item In each subtree, assign labels to the leaves, such that for any assignment to the root, the subtree can be completed with a valid labeling.
	\item Complete the labeling in each subtree independently.
\end{itemize}
Note that the only part requiring  $\Theta(\log^* n)$ is the first one, while the rest requires constant time. We then also prove that we can \emph{decide} if there is a subset of labels, and an assignment for the leaves of the subtrees, that satisfies the second point.

\begin{figure}
\figlayout

(a) Problem: 3-coloring in binary trees:
\begin{center}
\includegraphics[page=1,scale=\figscale]{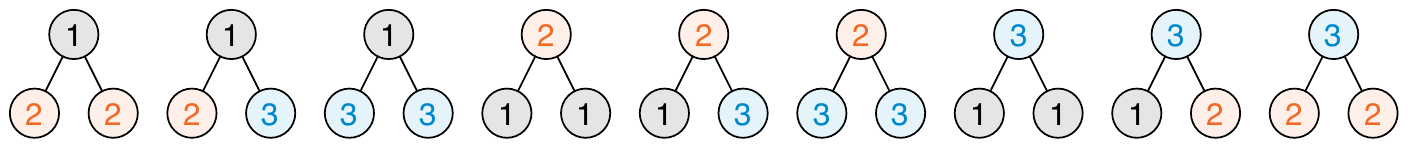}
\end{center}
\medskip

(b) Finding a certificate:
\begin{center}
\includegraphics[page=2,scale=\figscale]{figs.pdf}
\end{center}
\medskip

(c) Certificate for $O(\log^* n)$-round solvability:
\begin{center}
\includegraphics[page=3,scale=\figscale]{figs.pdf}
\end{center}

\caption{Finding a uniform certificate for $O(\log^* n)$ solvability (Definition~\ref{def:uniform-logstar-cert}) for the $3$-coloring problem (Section~\ref{ssec:3col-example}).}\label{fig:logstar-cert-3col}
\Description{}
\end{figure}

\subsection{Certificate}

We start by defining what is a uniform certificate of $O(\log^*n)$ solvability. Informally, it is a sequence of labeled trees having the same depth and the leaves labeled in the same way, such that for each label used in the trees there is a tree with the root labeled with that label. An example of such a certificate for the $3$-coloring problem is depicted in \autoref{fig:logstar-cert-3col}.

\begin{definition}[uniform certificate for $O(\log^*n)$ solvability]\label{def:uniform-logstar-cert}
Let $\Pi$ be an $\LCL$ problem. A uniform certificate of $O(\log^*n)$ solvability for $\Pi$ with labels  $\Sigma_\TT = \{\sigma_0,\ldots,\sigma_t\} \subseteq \Sigma(\Pi)$ and depth $d$ is a sequence $\TT$ of $t$ labeled trees (denoted by $\TT_i$) such that:
\begin{enumerate}
    \item Each tree is a complete $\delta$-ary tree of depth $d$ ($d$ has to be at least one).
    \item Each tree $\TT_i$ is labeled with labels from $\Sigma_\TT$ and correct w.r.t. configurations $\CC(\Pi)$.
    \item Let $\overline{\TT}_i$ be the tree obtained by starting from $\TT_i$ and removing the labels of all non-leaf nodes. It must hold that all trees $\overline{\TT}_i$ are isomorphic, preserving the labeling.
    \item Root of tree $\TT_i$ is labeled with label $\sigma_i$.
\end{enumerate}
\end{definition}
We will see that a problem $\Pi$ can be solved in $O(\log^* n)$ rounds if and only if a certificate of $O(\log^*n)$ solvability for $\Pi$ exists. We will later show that we can decide if such a certificate exists. We will now give an alternative definition of certificate, that we will later prove to be equivalent.
\begin{definition}[coprime certificate for $O(\log^*n)$ solvability]
	Let $\Pi$ be an $\LCL$ problem. A coprime certificate of $O(\log^*n)$ solvability for $\Pi$ with labels $\Sigma_{\TT} = \{\sigma_0,\ldots,\sigma_{t}\} \subseteq \Sigma(\Pi)$ and depth pair $(d_1,d_2)$ is a pair of sequences $\TT^1$ and $\TT^2$ of $t$ labeled trees (denoted by $\TT^1_i$ and $\TT^2_i$) such that:
	\begin{enumerate}
		\item The depths $d_1$ and $d_2$ are coprime.
		\item Each tree of $\TT^1$ (resp. $\TT^2$) is a complete $\delta$-ary tree of depth $d_1 \ge 1$ (resp. $d_2 \ge 1$).
		\item Each tree is labeled with labels from $\Sigma(\Pi)$ and correct w.r.t. configurations $\CC(\Pi)$.
		\item Let $\overline{\TT}^1_i$ (resp. $\overline{\TT}^2_i$) be the tree obtained by starting from $\TT^1_i$ (resp. $\TT^2_i$) and removing the labels of all non-leaf nodes. It must hold that all trees $\overline{\TT}^1_i$ (resp. $\overline{\TT}^2_i$) are isomorphic, preserving the labeling.
		\item The root of the tree $\TT_i^1$ (resp. $\TT^2_i$) is labeled with label $\sigma_i$.
	\end{enumerate}
\end{definition}
Note that the difference between a uniform certificate and a coprime certificate is that a coprime certificate requires two uniform certificates of coprime depth, but it allows internal nodes of the trees to be labeled from labels of $\Sigma(\Pi)$ that are not in $\Sigma_\TT$.
In the following, we will sometimes omit the type of the certificate, and we will just talk about \emph{certificate for  $O(\log^*n)$ solvability}. In this case, we will refer to a \emph{uniform} certificate.

\subsection{Upper bound}

We now present an $O(\log^*n)$-round algorithm that is able to solve $\Pi$ if there exists a certificate for $O(\log^*n)$ solvability for $\Pi$.

\begin{theorem}
\label{log*_algorithm}
Assume that a uniform or coprime certificate for\/ $O(\log^* n)$ solvability for\/ $\Pi$ exists. Then $\Pi$ can be solved in $O(\log^* n)$ rounds in the $\CONGEST$ model.
\end{theorem}
\begin{proof}
	We will prove our claim by describing an algorithm $A$. The algorithm will consist of two main phases. First, we split the tree into constant size subtrees in $O(\log^* n)$ rounds. Then, we operate in a constant number of rounds on these subtrees in parallel.
	We assume that nodes are far enough from the root or the leaves of the tree, as we can imagine our tree to be embedded in a slightly larger tree. Since the leaves are unconstrained, this does not affect the validity of the solution that we compute.

	Consider the following family of problems defined on directed paths. Let $\alpha,\beta$ be two parameters. Labels are from $\{1,a_2,\ldots,a_{\alpha}\} \cup \{1,b_2,\ldots b_{\beta}\}$. Allowed configurations are
	\[
		\bigl\{(a_i,a_{i+1}) \bigm| 2 \le i < \alpha \bigr\} \cup
		\bigl\{(b_i,b_{i+1}) \bigm| 2 \le i < \beta \} \cup
		\bigl\{(1,a_2),(1,b_2),(\alpha,1),(\beta,1)\bigr\}.
	\]
	Essentially, this problem requires to label a directed path such that if a node is labeled $1$, then its successor is either labeled $a_2$ or $b_2$, and then in the first case we continue counting up to $a_\alpha$ and then start again with $1$, while in the second case we continue counting up to $b_\beta$ and then start again with $1$. If $\alpha$ and $\beta$ are coprime, then this problem can be solved in $O(\log^* n)$ rounds, and by \cite[Theorem 16]{lcls_on_paths_and_cycles} we can solve this problem in rooted trees, such that every root-to-leaf path is labeled with a valid labeling w.r.t.\ the definition of the problem on directed paths in $O(\log^* n)$ rounds as well. Let us now modify the solution as follows: let $\ell(v)$ be the labeling obtained on each node $v$. Each node $v$ labels itself with $\ell(u)$, where $u$ is the parent of $v$. We obtain that all siblings have the same labeling.
	Consider now the subtrees obtained by removing edges where endpoints are labeled $(a_\alpha,1)$ or $(b_\beta,1)$. To each subtree, we add as new leaves the nodes on the other side of the edges that have been removed (that is, nodes labeled $1$ are at the same time roots of their tree and leaves of the tree above). By how labels can propagate, we obtain that each obtained subtree is a perfect $\dd$-ary trees, where each tree has height either $\alpha$ or $\beta$. If we are given a coprime certificate, we compute such a splitting with $\alpha,\beta$ equal to the depth pair of the certificate, while if we are given a uniform certificate, we compute such a splitting with $d,d+1$, where $d$ is the depth of the certificate.

We now describe the second phase. In the following, we describe an algorithm that fixes, in constant time, the labeling of each subtree in parallel. 

If we are given a coprime certificate, we proceed as follows. For every subtree of depth $\alpha$ we label the leaves as the trees of $\TT^1$, while for every tree of depth $\beta$ we label the leaves as the trees of $\TT^2$. Note that all the leaves are also the root of a tree below. Hence, for each subtree, we have fixed the labeling of the root and the leaves. Now, for each tree of depth $\alpha$ (resp. $\beta$) we complete the labeling as in $\TT^1_i$ (resp. $\TT^2_i$), where $\sigma_i$ is the label assigned to the root. In this way, we obtain a valid labeling for the whole tree.

If we are given a uniform certificate, to each subtree, we assign to the nodes at depth $d$ the labels of the trees of the certificate. On subtrees of depth $d+1$ we then assign a labeling to the nodes at depth $d+1$ by using only labels of the certificate. This is possible since each label of the certificate is a root of a certificate of the tree, and hence has a continuation below that only uses labels of the certificate. We now need to complete the labeling of trees of depth $d$ where all roots and all leaves have labels of the certificate, and all internal nodes are unlabeled. This is possible by copying the labels assigned to the internal nodes of the trees of the certificate.

The round complexity of the described algorithm is $O(\log^*n)$ for computing the subtrees, and $O(1)$ for everything else, hence we have an algorithm that has a round complexity of $O(\log^*n)$.
\end{proof}

\subsection{Lower bound}

We now prove that, if there is no certificate for $O(\log^*n)$ solvability, then the problem requires $\Omega(\log n)$, even for randomized algorithms.

We start by considering deterministic algorithms. We will prove that if there is a deterministic $o(\log n)$ algorithm for $\Pi$, then we can construct an $O(\log^* n)$ certificate for it. In the following lemma we will prove something stronger, that will be useful later when considering constant-time algorithms. This lemma essentially says that, if there exists a fast enough algorithm that uses some set of labels far enough from the root and the leaves, then we can construct a certificate that uses the same set of labels. Moreover, we can force a leaf of the certificate to contain some specific label (this specific part will be used when considering constant-time algorithms).
\begin{lemma}\label{lem:constructcert}
	Assume that there exists a deterministic algorithm $A$ solving\/ $\Pi$ in $T(n) \in o(\log n)$ rounds on instances of size $n$. Let $n_0$ be any integer satisfying $n_0 > (1+\dd)^{10 T(n_0)}$. Let $S$ be the maximal set of labels satisfying that for each $s \in S$ there exists an instance of\/ $\Pi$ of size $n_0$ in which $A$ outputs $s$ on at least one node at distance strictly larger than $T(n_0)$ from the root and from any leaf. Let $\bar{s}$ be an arbitrary label in $S$. Then there exists a certificate of $O(\log^* n)$ solvability that contains all labels in $S$, and where at least one leaf is labeled $\bar{s}$.
\end{lemma}
\begin{proof}
	For all $s \in S$, let $H_s$ be an instance of size $n_0$ in which there exists a node $v_s$ having distance strictly larger than $T=T(n_0)$ from the root and any leaf, where $A$ outputs $s$. Let $B(v)$ be the radius-$T$ neighborhood of a node $v$.

	We now consider a $\dd$-ary tree $G$ of $n_0$ nodes that is ``as balanced as possible''. Note that the height of $G$ is at least $10T$. Let $r$ be an arbitrary node at distance $T+1$ from the root of $G$. Let $L$ be the set of descendants of $r$ that are at distance exactly $5T$ from $r$. Since the tree is balanced and has height at least $10T$, then $r$ and all nodes of $L$ do not see the root of $G$ or any leaf of $G$. Also, $B(r)$ and $B(\ell)$ are disjoint, for all $\ell \in L$. Moreover, note that for all nodes $x$ that are on paths that connect $r$ with nodes of $L$, it holds that either $B(x)$ is disjoint with $r$ or $B(x)$ is disjoint with $B(\ell)$, for all $\ell \in L$.

	We now pick an arbitrary node $\ell \in L$ and fix the identifiers in its neighborhood to make $B(\ell) = B(v_{\bar{s}})$ (that is the subgraph of $H_s$ in which node $v_{\bar{s}}$ outputs $\bar{s}$).  Then, for each label $s \in S$ we make a copy of $G$ (copying the partial ID assignment as well), and we call it $G_s$. In each copy $G_s$, we additionally fix the neighborhood of $r_s$, that is the copy of $r$, to make it equal to $B(v_s)$. Then, we fix the identifiers of all other nodes by using unique identifiers not in $B(v_s) \cup B(v_{\bar{s}})$. Crucially, we assign the same identifier to all the nodes that are copies of the same node of $G$.
	
	We argue that, by running $A$ on the obtained trees, we must obtain a valid solution, even if some identifiers may not be unique. In fact, assume that there is a node in which the output does not satisfy the constraints of $\Pi$. The radius-$T$ neighborhood of this node must contain unique identifiers by construction, and we can hence construct a different instance of $n_0$ nodes where all identifiers are unique and the same bad neighborhood is contained. This would imply that $A$ fails in a valid instance, that is a contradiction.

	Hence, we obtain that each tree $G_s$ is properly labeled, and since $S$ is by definition maximal, then also nodes that are between $r_s$ and nodes in $L_s$, that are the copy of labels in $L$, must be labeled with only labels in $S$.

	Consider now the $|S|$ trees obtained by taking from each tree $G_s$ the subtree induced by node $r_s$, nodes in $L_s$, and all nodes between them. We obtain $|S|$ trees that have the same labeling for the leaves (that is, leaves that are copies of the same node of $G$ have the same labeling), at least one leaf is labeled $\bar{s}$, each tree has a different label of $S$ assigned to the root, and all nodes are only labeled with labels in $S$. Hence we obtained a certificate for $O(\log^* n)$ solvability that contains all labels in $S$, and where at least one leaf is labeled $\bar{s}$.
\end{proof}

In particular, since for any algorithm running in $o(\log n)$ rounds there exists some $n_0$ satisfying $n_0 > (1+\dd)^{10 T(n_0)}$, and since in $\delta$-ary trees of size $n_0$ there exist nodes at distance strictly larger than $T$ from the root and any leaf, implying that $S$ is non-empty, then  \autoref{lem:constructcert} shows that if there is a deterministic $o(\log n)$ algorithm for $\Pi$, then we can construct an $O(\log^* n)$ certificate for it. Hence, we obtain the following corollary.
\begin{corollary}\label{detgap}
	If\/ $\Pi$ has deterministic complexity $o(\log n)$ then there exists a certificate for\/ $O(\log^* n)$ solvability.
\end{corollary}

We can now prove that uniform and coprime certificates are in some sense equivalent.
\begin{lemma}\label{lem:certequiv}
	A uniform certificate for\/ $O(\log^* n)$ solvability exists if and only if a coprime certificate of\/ $O(\log^* n)$ solvability exists.
\end{lemma}
\begin{proof}
	We first show that, given a uniform certificate $\TT$, we can construct a coprime certificate. Let $d$ be the depth of the uniform certificate, we show how to construct a different certificate of depth $d+1$. Since each leaf is also a root of some tree in $\TT$, then each leaf has a continuation below. We can construct a certificate of depth $d+1$ by starting from the trees in $\TT$ and extending them to depth $d+1$ in a consistent manner by using the continuation below that is guaranteed to exist.

	We now prove that if there exists a coprime certificate then there exists a uniform certificate.
	By \autoref{log*_algorithm} a coprime certificate implies a deterministic $O(\log^* n)$ algorithm, and by \autoref{detgap} this implies the existence of a uniform certificate.
\end{proof}

We are now ready to extend \autoref{detgap} to randomized algorithms.
\begin{lemma}\label{nologstarcert}
	Let\/ $\Pi$ be an $\LCL$ problem for which no certificate for\/ $O(\log^* n)$ solvability exists.
	Then, the randomized and deterministic complexity of\/ $\Pi$ in the\/ $\LOCAL$ model is\/ $\Omega(\log n)$.
\end{lemma}

\begin{proof}
	We start by proving the lemma in the case where randomization is allowed and nodes have no identifiers assigned to them, and we prove the lemma by showing the contrapositive.
	To this end, let $\Pi$ be an $\LCL$ problem with randomized complexity $T(n) \in o(\log n)$.
	We will show that there exists a certificate for $O(\log^* n)$ solvability for $\Pi$.

	Let $\A$ denote an optimal randomized algorithm for $\Pi$; in particular the worst-case runtime of $\A$ on $n$-node trees is $T(n)$.
	Let $\Sigma(\Pi)$ be the output label set of $\Pi$ and set $k = |\Sigma(\Pi)|$.
	Since $T(n) \in o(\log n)$, there exists an integer $n_0 \geq 5k^2$ such that $1 \leq T(n_0) \leq 1/10 \cdot \log_\delta(n_0)$.
	Let $G$ be a rooted tree with $n_0$ nodes that is ``as balanced as possible''; in particular the $\lfloor \log_\delta(n_0) \rfloor$-hop neighborhood of the root $r$ is a perfectly balanced tree. Moreover, for any node $v$, denote the set of descendants of $v$ that are at distance precisely $2 T(n_0) + 1$ from $v$ by $D_1(v)$, and those at distance $2 T(n_0) + 2$ from $v$ by $D_2(v)$.
	Since $T(n_0) \leq 1/10 \cdot \log_\delta(n_0)$, there exists a node $v \in V(G)$ such that the distance between $r$ and $v$ is at least $T(n_0)+1$, and the distance between any node $u \in D_1(v) \cup D_2(v)$ and the leaf closest to $u$ is also at least $T(n_0)+1$.
	In particular, for any node $u \in D_1(v) \cup D_2(v)$ the views of $u$ and $v$ during an $T(n_0)$-round algorithm are disjoint, and both $v$ and any node in $D_1(v) \cup D_2(v)$ do not see any leaf or the root during an $T(n_0)$-round algorithm.

	Let $S \subseteq \Sigma(\Pi)$ be the set of all labels that $\A$ outputs with probability at least $1/(k \sqrt{n_0})$ if $\A$ does not see a leaf or the root (note that the view of all nodes that do not see a leaf or the root is the same, since nodes have no ids).
	By the definition of $D_1(v)$ and $D_2(v)$, we know that $D_1(v) \cup D_2(v) \leq \sqrt{n_0}/2$. Note that $\A$ outputs a specific label not in $S$ with probability strictly less than $1/(k \sqrt{n_0})$, and hence any label not in $S$ with probability less than $1/\sqrt{n_0}$ by a union bound, and hence the probability that at least one node in $D_1(v) \cup D_2(v)$ outputs a label not in $S$ is at most $|D_1(v) \cup D_2(v)| / \sqrt{n_0}$, by a union bound.
	Thus, the probability that all nodes in $D_1(v) \cup D_2(v)$ output a label from $S$ is at least
	\[
	1 - \frac{|D_1(v) \cup D_2(v)|}{\sqrt{n_0}} \geq \frac{1}{2} \enspace.
	\]

	Let us assume for a contradiction that there is no certificate for $O(\log^* n)$ solvability for $\Pi$. By \autoref{lem:certequiv} this implies that there is no coprime certificate as well.
	In particular, then there is no such certificate with labels from $S$ and depth pair $2 T(n_0) + 1$ and $2 T(n_0) + 2$.
	This implies that there exists a $D(v) \in \{D_1(v),D_2(v)\}$ such that for any labeling $\ell \colon D(v) \to S$ of the nodes in $D(v)$ with labels from $S$, there exists some label $s \in S$ such that $\ell$ is \emph{incompatible} with $s$, i.e., such that there is no correct solution for $\Pi$ where $v$ is labeled with $s$ and $D(v)$ is labeled according to $\ell$.

	Since we have $|S| \leq k$, and the view of $v$ during $\A$ and the union of the views of the nodes in $D(v)$ during $\A$ are disjoint, it follows that the probability that the labeling $\ell$ of $D(v)$ that $\A$ outputs is incompatible with the label $s$ that $\A$ outputs at $v$ is at least
	\[
	\frac{1}{2} \cdot \frac{1}{k \sqrt{n_0}} > \frac{1}{n_0}
	\]
	since $n_0 \geq 5k^2$.
	Hence, $\A$ fails with probability larger than $1/n_0$, yielding a contradiction, and proving the lemma for the case in which identifiers are not provided to the nodes, but randomization is allowed.

	We now prove that the same result holds even if we provide identifiers to the nodes. Assume that there exists a $o(\log n)$ randomized algorithm for the case in which identifiers are provided. We can run it in the case in which identifiers are not provided in the same asymptotic running time by first generating unique random identifiers in $\{1,\ldots,n^3\}$, that can be done with a randomized algorithm with high probability of success. But this would imply a $o(\log n)$ randomized algorithm for the case in which identifiers are not provided, contradicting the lemma. Finally, since the existence of a deterministic algorithm implies the existence of a randomized algorithm, then the lemma follows.
\end{proof}

\subsection{Decidability}
We now prove that we can decide if a certificate for $O(\log^* n)$ solvability exists.
\autoref{alg:certificate} describes a procedure that returns a \emph{certificate builder} if and only if a certificate exists. A certificate builder is an object that can be used to easily construct a certificate, and in \autoref{lem:cbtocert} we will show that, given a certificate builder, we can indeed construct a certificate. This procedure uses another subroutine, \autoref{alg:unrestricted_certificate}, to try to find a certificate builder that uses a specific subset of labels, for all possible subsets of labels. On a high level, \autoref{alg:unrestricted_certificate} works as follows. We start with singleton sets, one for each label. Then, we repeatedly try to build new sets. Each new set is obtained as follows. We consider all tuples of size $\dd$ of existing sets, and we see which configurations exist where the label of each leaf $\ell_i$ is contained in the $i$th sets of the tuple. The set of roots of such configurations defines a new set. We repeat this process until we obtain a fixed point. The algorithm that we describe is also able to find a certificate builder that contains a leaf with some specific label, if required and if it exists; we will need this in the next section.

\begin{algorithm}[tbp]
\DontPrintSemicolon
\caption{\label{alg:unrestricted_certificate}$\findUnrestrictedCertificate(\Pi,a)$}
\KwIn{$\LCL$ problem $\Pi$, label $a \in \Sigma(\Pi)$ or $a = \epsilon$}
\KwOut{$\epsilon$ if no certificate exists or a certificate builder}
\BlankLine

$i \gets 0$\;

$R_0 \gets \set{(\set{\sigma},\sigma = a)  | \sigma \in \Sigma(\Pi)}$ \comm*{$R_i$ is a set of pairs. The first element of each pair is a set of possible labels of certificate roots. The second element is a Boolean that indicates whether such set of root labels can have label $a$ as one of its leaves. For $R_0$, the Boolean is true if and only if $\sigma$ is equal to $a$.}

$\CB \gets \emptyset$ \comm*{Certificate builder which will describe how we constructed individual elements in sets $R_i$.}

\Repeat{$R_{i} = R_{i-1}$  \comm*{Until we do not enlarge the set $R_i$.}}{
  $i \gets i + 1$\;
  $R_{i} \gets R_{i-1}$\;
  \For{every $\dd$-tuple of sets of root labels and indicators $(r_1,a_1), (r_2,a_2), \dotsc, (r_\dd,a_\delta)$ from $R_{i-1}$} {

     $r_n \gets \set{\sigma | (\sigma : c_1,c_2,\dotsc,c_\dd) \in \CC(\Pi), c_1 \in r_1, c_2 \in r_2, \dotsc, c_\dd \in r_\dd}$\;
     $a_n \gets$  $a_1$ or $a_2$ or \dots{} or $a_\delta$ \comm*{Indicates whether $r_n$ will have a leaf layer containing label $a$.}
     \If{$r_n \neq \emptyset$  and $(r_n,a_n) \not \in R_i$}{
     $\CB \gets \CB \cup ((r_n,a_n),((r_1,a_1),(r_2,a_2),\dotsc,(r_\dd,a_\delta)))$\;
      $R_i \gets R_i \cup (r_n,a_n) $\;
     }
    }
  }

\uIf{$(\Sigma(\Pi),a \neq \epsilon) \in R_i$ and $\CC(\Pi)$ is non-empty}{
  \Return $\CB$
  }
\Else{
  \Return $\epsilon$
}
\end{algorithm}

\begin{algorithm}[tbp]
\DontPrintSemicolon
\caption{\label{alg:certificate}findCertificate($\Pi$)}
\KwIn{$\LCL$ problem $\Pi$}
\KwOut{$\epsilon$ if no certificate exists or a certificate builder}
\BlankLine

\For{all elements $\Sigma'$ of $2^{\Sigma(\Pi)}$} {
  $\Pi' \gets$  restriction of $\Pi$ to labels from $\Sigma'$\;
  $\certificateBuilder \gets \findUnrestrictedCertificate(\Pi',\epsilon)$ \;
  \If{$\certificateBuilder \neq \epsilon$} {
    \Return $\certificateBuilder$\;
  }
}
\Return $\epsilon$\;
\end{algorithm}

We now prove that \autoref{alg:certificate} outputs a certificate builder if and only if a certificate of $O(\log^* n)$ solvability exists.

\begin{theorem}\label{thm:findlogstarcertifexists}
	Given an $\LCL$ problem $\Pi$, \autoref{alg:certificate} outputs a certificate builder if and only $\Pi$ has an $O(\log^* n)$ certificate satisfying the leaf requirement.
\end{theorem}
\begin{proof}
  We will later show, in \autoref{lem:cbtocert}, that if \autoref{alg:certificate} outputs a certificate builder then a certificate exists, so we now prove the reverse implication, that is, if the problem has a certificate of $O(\log^* n)$ solvability (satisfying the leaf requirement), then \autoref{alg:certificate} will find a certificate builder.

Let $\TT = (\TT_1,\TT_2,\dotsc,\TT_{t})$ be a certificate with labels $\Sigma_\TT$ that satisfies the leaf requirement (that is, if $a \neq \epsilon$ then at least one leaf is labeled $a$). Let $|\TT| = t$. Let $\lambda$ be the labeling function of the certificate (that is, a function mapping each node of each tree of the certificate to its assigned label), and let $t_{i,j,k}$ be the $j$th node on level $i$ of the $k$th tree.

First, we define $S_{i,j}$ as the set of all labels of the $j$th nodes on level $i$, that is, $S_{i,j} = \bigcup_{k=1}^{|\TT|} \lambda(t_{i,j,k})$. $S_{0,0}$ is by definition equal to $\Sigma_\TT$, and for all nodes on level $d$, $S_{d,0},S_{d,1},\dotsc,S_{d,\dd^d}$ are singletons, by definition of certificate (recall that each $\TT_i$ has depth $d$).

We will prove by induction on the depth of $\TT$ that, for all $i$ and $j$, there exists a pair $(S'_{i,j},x)$ in the set $R$ of \autoref{alg:unrestricted_certificate}, where $S'_{i,j} \supseteq S_{i,j}$, and $x$ is true if and only if $t_{i,j,k}$ is an ancestor of leaves labeled $a$. This would imply that $(\Sigma_\TT,b)$ is also in $R$, where $b$ is true if and only if $a \neq \epsilon$, and that \autoref{alg:unrestricted_certificate} outputs certificate builder which is what we want to prove.

In the base case, sets $S_{d,i}$ are just singletons $\{\sigma\} \subseteq \Sigma_\TT$ and we add $(S_{d,i},\sigma = a)$ in the initialization of the set $R$.

For the induction hypothesis, let us assume that all $(S'_{i+1,j},x)$ for level $i+1$ are in $R$, where $S'_{i+1,j} \supseteq S_{i+1,j}$, and $x$ is true if and only if $t_{i+1,j,k}$ are ancestors of a leaf labeled $a$. We prove the statement for $i$. In the algorithm we loop over all $\dd-$tuples of elements from $R$ to enlarge $R$, and hence also over the tuple $((S'_{i+1,m},a_1), (S'_{i+1,m+1},a_2), \dots, (S'_{i+1,m+\dd-1},a_{\delta}))$, where $m=j\delta$, that is, a tuple containing supersets of the sets assigned to nodes that are children of nodes in position $(i,j)$. Since certificate trees are labeled correctly, this implies that, starting from this tuple, we compute $S_{i,j}$ or a superset of it. Also, the Boolean that we put in the pair that we add to $R$ is also correct, since we compute it as the or of the ones of the children.
\end{proof}

\begin{lemma}
\label{lem:cbtocert}
Let $\CB$ be a non-empty certificate builder obtained from \autoref{alg:unrestricted_certificate} for\/ $\LCL$ problem $\Pi$ and a label $a$. Then there exists a certificate of\/ $O(\log^* n)$ solvability\/ $\TT$ with at least one leaf labeled with $a$ (for $a \neq \epsilon$) and without such restriction for $a = \epsilon$.
\end{lemma}
\begin{proof}
Before dealing with the general case, if $\Sigma(\Pi)$ consists of only one label $\sigma$ and we have a non-empty certificate builder $\CB$, then we also know that $\CC(\Pi)$ is non-empty (see last part of \autoref{alg:unrestricted_certificate}). Hence a certificate will be just a single tree of depth one labeled with $\sigma$. Hence for the rest of the proof, we assume that $\Sigma(\Pi)$ has size at least two.

To convert a certificate builder to a certificate for $O(\log^* n)$ solvability, we proceed in four phases. The first phase consists of creating a \emph{verbose temporary} tree $\TT_{\mathrm{vt}}$ which will be labeled with pairs with first element being sets of labels and second element being an indicator for where to find label $a$ (that is, each node of $\TT_{\mathrm{vt}}$ is labeled with an element of~$R_i$). Tree $\TT_{\mathrm{vt}}$ is created recursively as follows:
\begin{enumerate}
  \item root of $\TT_{\mathrm{vt}}$ is labeled with $(\Sigma(\Pi),a \neq \epsilon)$
  \item each node labeled with $(r,a)$, where $r$ contains at least two labels, will have $\delta$ children labeled with $(r_i,a_i)$ according to the (unique) pair $((r,a),((r_1,a_1),\dots,(r_\delta,a_\delta))$ in $\CB$ which contains pair $(r,a)$ as its first element.
  \item each node labeled with $(r,a)$, where $r$ is a singleton set, is a leaf.
\end{enumerate}

The recursive definition is legal as the labels for children will always be placed in the certificate builder earlier, so we cannot have any loop (see \autoref{alg:unrestricted_certificate}). In the case when $a \neq \epsilon$, by following the indicators (second pair of label) from root of $\TT_{\mathrm{vt}}$ which will be labeled with $(\Sigma(\Pi),true)$, we must be able to reach a leaf that is labeled $(\set{a},true)$. This implies that we have a singleton $a$ as one of the leaves.

As we don't need the second element of each label anymore, let us simplify the further analysis by creating a \emph{simplified temporary} tree $\TT_{\mathrm{st}}$ as a simplification of $\TT_{\mathrm{vt}}$ where each node is labeled only by the first element from the pair. Examples of such trees are depicted in  \autoref{fig:logstar-cert-3col}b and \autoref{fig:constant-cert-mis}b.

The second phase considers the case when label $a$ is not $\epsilon$. In this phase we want to ``push down'' a leaf node labeled with the singleton label $a$ so it is a deepest node of the tree $\TT_{\mathrm{st}}$. We will do it as follows.
Let $n_a$ be a node in $\TT_{\mathrm{st}}$ that is labeled with the singleton label $a$. Let $d_a$ be its depth. Since the root of $\TT_{\mathrm{st}}$ is labeled with $\Sigma(\Pi)$, we know that there exists a hairy path of length $d_a$ labeled with $\Sigma(\Pi)$ that has both of its endpoints labeled $a$. For convenience, let $P_{\mathrm{aa}}$ denote such a hairy path and replace all labels by their singleton labels (label $\sigma$ will become $\set{\sigma}$). We will use $P_{\mathrm{aa}}$ and replace node $n_a$ with the path $P_{\mathrm{aa}}$. We have now essentially ``pushed down'' a leaf node labeled with the singleton label $a$ by $d_a$ steps. We will repeat such ``pushing down'' until we have that $n_a$ is the deepest node of the tree $\TT_{\mathrm{st}}$. To summarize, now we have a tree $\TT_{\mathrm{st}}$ that has its deepest leaf labeled with the singleton $a$.

In the third phase, we want to make all leaves to be on the same level. We do it as follows.
Again, observe that since $\Sigma(\Pi)$ is the root of $\TT_{\mathrm{st}}$, we have a continuation below for every label from $\Sigma(\Pi)$. We can use such continuation to ``push down'' every leaf node of $\TT_{\mathrm{st}}$ that is not the deepest by one step in the same manner as in the previous phase. We replace a leaf node $n_l$ labeled with a singleton $\sigma$ with a $\delta$-ary tree of depth one, labeled with singletons corresponding to a continuation below for $\sigma$. To summarize, now we have a tree $\TT_{\mathrm{st}}$ that has all leaves at the same level. Examples of such trees are again depicted in \autoref{fig:logstar-cert-3col}b and \autoref{fig:constant-cert-mis}b.

Finally, in the last phase, we use $\TT_{\mathrm{st}}$ to build $|\Sigma(\Pi)|$ individual labeled trees $\TT_{\mathrm{st}}$ that would form a certificate. At the beginning, let each $\TT_i$ be labeled exactly as $\TT_{\mathrm{st}}$. Then, for each $\TT_{\mathrm{st}}$, we fix its root label to a distinct label $\sigma_i$. Then, we recursively fix the labels of the children such that the resulting configuration is in $\CC(\Pi)$. Such a configuration will always exist as it is how we have constructed the certificate builder (see \autoref{alg:unrestricted_certificate}). Examples of the obtained trees are depicted in \autoref{fig:logstar-cert-3col}c and \autoref{fig:constant-cert-mis}c.
\end{proof}

We now prove an upper bound on the running time of \autoref{alg:certificate}.
\begin{theorem}
  \label{certificate_search_runtime}
  The running time of \autoref{alg:certificate} is at most exponential in the size of the $\LCL$ problem.
\end{theorem}

\begin{proof}
  Observe that every iteration of the for loop in \autoref{alg:unrestricted_certificate} either adds an element to $R_i$ or finishes the algorithm. Hence, we can upper bound the number of iterations by the maximum size of each set $R_i$, that is $2^{|\Sigma|+1}$. Also, each iteration requires at most exponential time in $\Sigma$ and $\dd$. Hence, the total running time of \autoref{alg:unrestricted_certificate} is exponential in $\Sigma$ and $\dd$. Since \autoref{alg:certificate} just calls \autoref{alg:unrestricted_certificate} for every choice over $\Sigma$, then we get one more exponential slowdown, hence the claim follows.
\end{proof}

Hence we conclude the following theorem.

\begin{theorem}
  \label{certificate_search_runtime-main}
  Whether an $\LCL$ problem $\Pi$ has round complexity $O(\log^* n)$ or\/ $\Omega(\log n)$ can be decided in time at most exponential in the size of the $\LCL$ problem.
\end{theorem}

\section{Sub-log-star region}\label{sec:sublogstar}

In this section we prove that there is no $\LCL$ problem $\Pi$ with distributed time complexity between $\omega(1)$ and $o(\log^* n)$. Also, we prove that, given a problem $\Pi$, we can \emph{decide} if its complexity is $O(1)$ or $\Omega(\log^* n)$. Moreover, we prove that randomness cannot help: if a problem has \emph{randomized} complexity $O(1)$, then it has the same \emph{deterministic} complexity.

\subsection{High-level idea}

We prove that deciding if a problem $\Pi$ can be solved in constant time is surprisingly simple: a problem is $O(1)$ rounds solvable if and only if it can be solved in $O(\log^* n)$ rounds and $\Pi$ contains an allowed configuration of a specific form (this configuration will be called \emph{special}). This configuration must allow a node to use the same label $\ell$ that one of its children uses, the labels used by this configuration should be contained in the ones used by some certificate for $O(\log^*n)$ solvability, and $\ell$ should be used by at least one leaf of the certificate. If we consider the definition of the MIS problem given in Section~\ref{ssec:mis-example}, we can see that it allows the configuration $(b : b1)$, and informally this configuration is what makes the problem constant-time solvable. Note that, however, the algorithm that we can obtain by using this certificate, while still being constant time, may have a worse complexity compared to the one described in Section~\ref{ssec:mis-example}. On the other hand, we can see that in the definition of the $3$-coloring problem given in Section~\ref{ssec:3col-example} there is no configuration of this form, and this is what makes the problem $\Omega(\log^* n)$.

Informally, the reason is the following. The $n$ appearing in $O(\log^*n)$ complexities does not usually refer to the size of the graph, but to the range of the identifiers assigned to the nodes. In fact, in the proof of \autoref{log*_algorithm},  $O(\log^*n)$ is spent only to compute some ruling set, while the rest only requires constant time, and in order to compute such a ruling set, a distance-$k$ coloring, for some large enough $k$, is sufficient. Unfortunately, it is not possible to compute a distance-$k$ coloring in constant time, but as we will show, some defective coloring (that is, a coloring that allows some neighbors of a node to use the same color of the node) will be sufficient for our purposes. We show that in constant time we can produce some defective distance-$k$ coloring, for some large enough constant $k$, such that:
\begin{itemize}
  \item we can label defective nodes with the special configuration,
  \item unlabeled nodes are properly colored, and
  \item labeled nodes that are in different connected components are far enough from each other.
\end{itemize}
We can then complete the partial labeling in constant time with the help of the certificate, similarly to how we use the certificate of $O(\log^* n)$ solvability to solve problems in $O(\log^* n)$ rounds, but this time we can speed the computation of the ruling set up, and make it run in constant time by exploiting the defective distance-$k$ coloring. In the other direction, we show that if the special configuration does not exist, or if it does not satisfy the required properties, then any algorithm solving the problem can also be used to solve the coloring problem with a constant size palette, that is known to require $\Omega(\log^* n)$ rounds.

\begin{figure}
\figlayout

(a) Problem: maximal independent set in binary trees:
\begin{center}
\includegraphics[page=4,scale=\figscale]{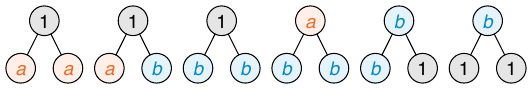}
\end{center}
\bigskip
\bigskip

(b) Finding a certificate where one of the leaf nodes is labeled with $b$:
\vspace{1mm}
\begin{center}
\includegraphics[page=5,scale=\figscale]{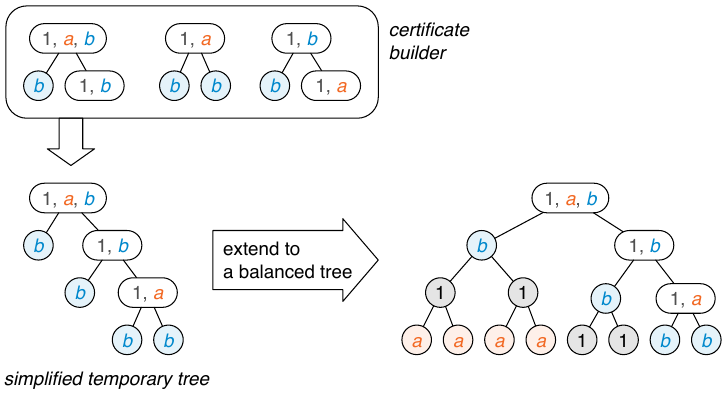}
\end{center}
\bigskip
\bigskip

(c) Certificate for $O(1)$-round solvability:
\begin{center}
\includegraphics[page=6,scale=\figscale]{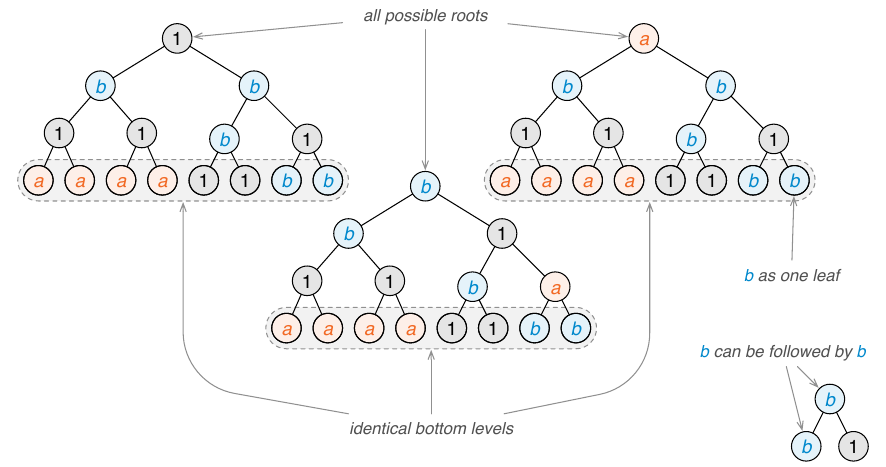}
\end{center}

\caption{Finding a certificate for $O(1)$ solvability (Definition~\ref{def:const-cert}) for the maximal independent set problem (Section~\ref{ssec:mis-example}).}\label{fig:constant-cert-mis}
\Description{}
\end{figure}

\subsection{Certificate}

We start by defining what is a certificate for $O(1)$ solvability, that is nothing else but a certificate for $O(\log^*n)$ solvability and a configuration of some specific form. An example of such a certificate for the MIS problem is depicted in \autoref{fig:constant-cert-mis}.
\begin{definition}\label{def:const-cert}
Let $\Pi$ be an $\LCL$ problem. A certificate for $O(1)$ solvability for problem $\Pi$ is a pair $\SX$ consisting of a certificate for $O(\log^*n)$ solvability $\TT$ and a configuration $(a : b_1,\dots,a,\dots,b_\dd) \in \CC(\Pi)$ where  $a,b_i \in \Sigma_\TT$ and at least one leaf of the trees in $\TT$ is labeled~$a$.
\end{definition}

\subsection{Upper bound}

We now prove that we can use a certificate for $O(1)$ solvability to construct an algorithm that solves the problem $\Pi$ in constant time. Informally, we first spend a constant number of rounds to try to construct some distance-$k$ coloring. This coloring cannot always be correct, since the coloring problem requires $\Omega(\log^* n)$ rounds. We will use the special configuration to label nodes in which the coloring procedure failed. The coloring will also satisfy some desirable property, such as having improperly colored regions that are far enough from each other. This will give us a proper distance-$k$ coloring in the unlabeled regions, and we will use this coloring to complete the labeling in constant time.
The proof of this theorem will use some useful lemmas that we will prove later.

\begin{theorem}
\label{certificate_and_aab_implies_O(1)}
Any $\LCL$ problem $\Pi$ that has a certificate of $O(1)$ solvability is constant-time solvable with a deterministic $\CONGEST$ algorithm.
\end{theorem}

\begin{proof}
We prove the statement by constructing a constant-time algorithm with the help of a special configuration $(a : b_1,\dots,a,\dots,b_\dd)$ and $O(\log^*n)$ certificate $\TT$ of depth $d$ where at least one leaf of the trees in $\TT$ is labeled $a$.
Let $k = 20 d + 1$.

For each node $v$, let $p(v) \in \{1,\ldots,\dd\}$ be the index of $v$ in the sorted sequence containing the identifier of $v$ and all its siblings (that is, $p(\cdot)$ emulates port numbers). We start by assigning a (possibly non-proper) coloring  $c(v)$ to each node $v$, as follows. Consider the sequence $(v_i, i \ge 0)$ of nodes obtained by starting from $v_0 = v$ and following edges going up. The color $c(v)$ is defined as $(p(v_0), p(v_1), \dots, p(v_{10k-1}))$. If $p(v_i)$ is undefined because $v_i$ does not exist, that is, while going up we found the root, then we complete the sequence with $1$s (this is equivalent to imagine the rooted tree to be embedded into a larger rooted tree, where all nodes of the original are far enough from the root of the larger tree). This can be done in constant time in $\CONGEST$, by first computing and sending the $p(\cdot)$ value of all children, and then repeatedly propagating down the $p(\cdot)$ value received from the parent. We say that a prefix of length $x$ of a color $c=c(v)$ has period $r$ if and only if $c[i] = c[i+r]$ for all $i < x - r$.

We define a \emph{vertical path} to be a subpath of a root-to-leaf path. We mark all nodes $v$ satisfying that the prefix of length $9k$ of $c(v)$ has period at most $k$. By \autoref{marked_regions_form_a_path}, the connected components induced by marked nodes form vertical paths. We use the configuration $(a : b_1,\dots,a,\dots,b_\dd)$ to label all marked nodes with $a$, and all their children with the other labels of the configuration in some arbitrary consistent manner.

By \autoref{marked_regions_are_far_apart}, all connected components of the marked nodes are at distance at least $k$, and by \autoref{unmarked_regions_form_distance-k_coloring}, all unmarked regions are distance-$k$ colored. We now show how to complete the partial labeling assigned to the nodes. We will start by splitting the unlabeled parts of the tree into constant size subtrees, by exploiting the distance coloring to perform this step in constant time. Then, we will operate in a constant number of rounds on these subtrees in parallel.

Let us now focus on the first phase, namely splitting the trees. We will compute a splitting of the unlabeled regions satisfying the following:

\begin{itemize}
  \item Subtrees have overlapping \emph{boundaries}, meaning that each leaf of a subtree is also the root of another subtree.
  \item The distance of each leaf from the root is in $\{d,\ldots,10d\}$.
  \item All inner nodes of each subtree are unlabeled.
\end{itemize}

Consider the following problem, which we call \emph{ruling set extension problem} (see Figure~\ref{fig:one-sided-ruling-set}a): We are given a directed path and some set $S$ satisfying that the nodes in $S$ are at distance at least $2d$ from each other.  The goal is to compute a set $S'$ such that each node not in $S \cup S'$ has a node in $S \cup S'$ at distance \emph{at most} $4d$, and for each node in $S'$ it holds that the closest successor in $S \cup S'$ is at distance \emph{at least} $2d$.
That is, $S \cup S'$ is almost a $(2d,4d)$-ruling set: nodes in the given set $S$ can violate the ruling set requirements, since they could have a successor in $S'$ at distance less than $2d$.

\begin{figure}
\figlayout

(a) Ruling set extension problem (for $d = 1$):
\begin{center}
\includegraphics[page=19,scale=\figscale]{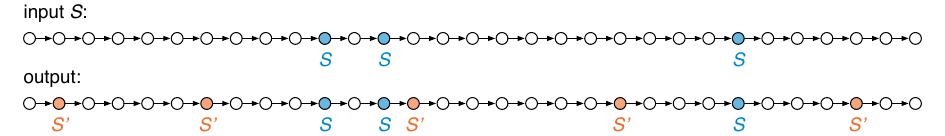}
\end{center}
\bigskip
\bigskip
\bigskip

(b) One-sided algorithm that solves the ruling set extension problem (for $d = 1$):
\begin{center}
\includegraphics[page=20,scale=\figscale]{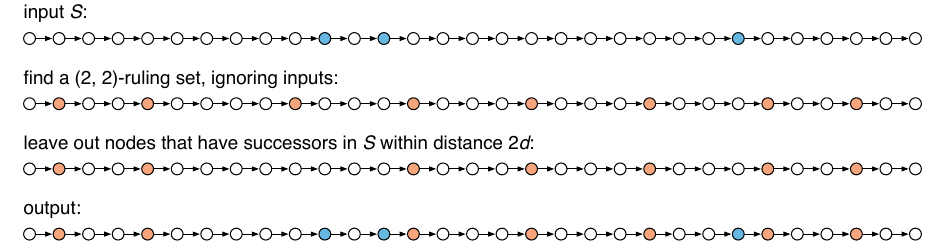}
\end{center}

\caption{Ruling set extension problem.}\label{fig:one-sided-ruling-set}
\Description{}
\end{figure}

We show that the ruling set extension problem can be solved in $O(1)$ rounds on directed paths by using a \emph{one-sided} algorithm, provided that we are given a distance-$k$ coloring with $O(1)$ colors for a sufficiently large $k$. A one-sided algorithm is an algorithm in which nodes only send messages to their predecessors (or equivalently, nodes only receive information from their successors). One-sided algorithms are convenient, as they are directly applicable in rooted trees \cite{lcls_on_paths_and_cycles}: if we have a one-sided algorithm that finds a ruling set extension in directed paths, we can apply the same algorithm in rooted trees and it will produce an output in which all root-to-leaf paths satisfy the constraints of the ruling set extension problem. Our one-sided algorithm works as follows (see Figure~\ref{fig:one-sided-ruling-set}b):
\begin{enumerate}
  \item[(i)] Compute $S'$ by using a $(2d,2d)$-ruling set algorithm (ignoring $S$ entirely).
  \item[(ii)] Nodes in $S'$ that have a successor in $S$ within distance $2d$ are removed from $S'$. 
\end{enumerate}
We obtain that nodes in $S$ have no predecessors in $S \cup S'$ within distance $2d$, nodes in $S'$ are at distance at least $2d$ from each other, and nodes not in $S \cup S'$ have at least a node in $S \cup S'$ within distance $4d$, and hence a solution for the ruling set extension problem. Here step~(ii) is easy to implement with a one-sided algorithm. There is also a simple two-sided algorithm $A$ for solving step~(i): process the nodes by color classes; whenever we consider a particular node, check if there is already another node within distance $2d$ that we have selected, and if not, select the node. Finally, we use the standard trick of ``shifting the output'' \cite{lcls_on_paths_and_cycles} to turn $A$ into a one-sided algorithm $A'$ that solves the same problem, as follows: Let $T = O(1)$ be the running time of algorithm $A$, and let us label the nodes by $v_1, v_2, \dotsc$ along the path. In algorithm $A'$ node $v_i$ will output whatever node $v_{i+T}$ outputs in algorithm $A$. The output of node $v_{i+T}$ in algorithm $A$ only depends on the input colors of nodes $v_i, v_{i+1}, \dotsc, v_{i+2T}$, and hence one-sided information is sufficient for $A'$ to simulate $A$. We have simply shifted the ruling set by $T$ steps.

Let us now get back to the task of splitting trees. We proceed as follows. We add all nodes whose label is fixed to a set $S$. Then, we run the one-sided algorithm for the ruling set extension problem on the subgraph induced by unlabeled nodes and their neighbors. Let $S'$ be the output of the algorithm. As discussed, in any root-to-leaf path we obtain a solution for the problem described above, and observe that this implies that, in any root-to-leaf path, nodes in $S \cup S'$ have at least one successor and one predecessor in $S \cup S'$ at distance at most $8d+1 < 10d$.
What we obtained almost satisfies the requirements of the splitting, except that some nodes of $S'$ may be too near to the nodes that were already in the set, since the minimum distance is guaranteed only while following successors (that is, by going up), and we now fix this issue.

Let $P \subseteq S'$ be the set of nodes that have an already labeled (that is, nodes in $S$) node as one of its descendants at distance less than $d$.
For each such node $n_\mathrm{p}$ from $P$ we do the following.
Let $d_c$ denote the distance from node $n_\mathrm{p}$ to its closest labeled node $n_\mathrm{c}$ below.
We remove $n_\mathrm{p}$  from $S'$ and add all nodes that are descendants of $n_\mathrm{p}$ at distance exactly $d_\mathrm{c}$ to the set $S'$. Note that, as the connected components of labeled nodes are strictly more than $20d$ steps apart (by definition of $k$), all nodes $n_\mathrm{u}$ just added to $S'$ (except node $n_\mathrm{c}$) will \emph{not} have any node of $S$ below it that is closer than $20d - \mathrm{distance}(n_\mathrm{u},n_\mathrm{c}) \ge 18d$ steps. The nodes added to $S'$ will be closer to the nodes of $S'$ below, but still at least $2d - d = d$ far away, as distance between two nodes of $S'$ on any root-to-leaf path was originally always at least $2d$.
Similarly, they will be further away from a node in $S'$ that is above them, but again at most distance $8d+1+d = 9d+1 \le 10d$ far away.
Now, observe that the nodes in $S \cup S'$ partition the input tree in subtrees with the required properties.

We now describe the second phase. For that, we use the certificate $\TT$, and for each subtree in parallel, we do the following.
First, we check whether it has one of its leaves already fixed.
If not, we directly fix labels at depth $d$ exactly as a leaf layer of any tree from the certificate (they are by definition the same). Otherwise, more care is needed. Observe that each subtree can have at most one fixed leaf (labeled with $a$). In fact, in any tree whose depth is upper bounded by $10d$, any two nodes are at most distance $20d$ apart (the distance is upper bounded by the length of a walk that starts from a node, goes to the root, and then goes to other node); but as fixed nodes are strictly more than $20d$ apart (by the definition of $k$), having one fixed leaf means that all other leaves are unlabeled. Let such fixed leaf be denoted by $n_\mathrm{f}$ and also denote by $n_\mathrm{d}$ the node that is on a the path connecting $n_\mathrm{f}$ to the root at distance $d$ from the root. We fix labels at depth $d$ exactly as a leaf layer of any tree from the certificate such that node $n_\mathrm{d}$ will be labeled with label $a$. This is possible as we can freely choose the ordering of the children. Then, we use the configuration $\aab$ to label the nodes of the hairy path connecting $n_\mathrm{d}$ to $n_\mathrm{f}$.

Now, we have fixed layer $d$ of all subtrees, and also some of the nodes below layer $d$. We proceed from layer $d$ and towards lower layers and label all of its children arbitrarily as every label has a continuation below. This procedure will stop after constant time as our trees have constant depth.

The only remaining part is to fix labels for nodes that are between the roots and the nodes at layer $d$. For that, we use the certificate trees and for every subtree with root fixed to label $\sigma_i$, we use tree $\TT_i$ to label the upper layers.
\end{proof}

We now prove that if marked nodes are at distance at most $k$, then they must lie in the same vertical path. Intuitively, if two nodes $x_1$ and $x_2$ are siblings, then $p(x_1) \neq p(x_2)$, implying that they cannot both have a periodic color, and that the same must hold for the descendants of $x_1$ and $x_2$, up to some distance. Hence, we can find nodes with a periodic color only by following vertical paths.

\begin{lemma}
\label{knowledge_of_marked_regions}
If two marked nodes are at distance at most $k$, then they are in the same connected component of marked nodes, and each connected component forms a vertical path.
\end{lemma}
\begin{proof}
We will prove the statement by contradiction. Suppose that $v_1$ and $v_2$ are two marked nodes from different connected components that are at distance $d\le k$ and the prefix of length $9k$ of their colors has period at most $k$. Let $k_1$ and $k_2$ be, respectively, the period of the colors of $v_1$ and $v_2$. Since they are at distance $d \le k$ from each other, then their lowest common ancestor $v_3$ is at distance at most $k$ from both. Let $d_1$ and $d_2$ be, respectively, the distance of $v_1$ and $v_2$ from $v_3$. Since a prefix of length at least $8k$ of $c(v_3)$ is equal to suffixes of length at least $8k$ of $c(v_1)$ and $c(v_2)$, then the prefix of length $8k$ of $c(v_3)$ has also period $k_3$ that is at most $k$, and that satisfies $k_3 = k_1 = k_2 = k'$. We now prove that either  $v_1 = v_3$ or $v_2 = v_3$. Assume it is not the case, then there must be two children $x_1$ and $x_2$ of $v_3$ that lie in the two paths connecting $v_3$ to $v_1$ and $v_2$. Since $v_3$ is the lowest common ancestor, $x_1 \neq x_2$, and since $x_1$ and $x_2$ are siblings, then $p(x_1) \neq p(x_2)$. Since $p(x_1) = c(v_1)[d_1-1]$, and $p(x_2) = c(v_2)[d_2-1]$, then by the period assumption $c(v_1)[d_1-1+k'] \neq  c(v_2)[d_2-1+k']$, which is a contradiction since by going up from $v_1$ for $d_1-1+k'$ steps we reach the same node that we reach by going up from $v_2$ for $d_2-1+k'$ steps. Hence, $v_1$ and $v_2$ lie in the same vertical path. Also, note that all nodes in the vertical path between $v_1$ and $v_2$ will be marked as well, since the prefix of length $9k$ of their color is contained in the union of the prefixes of $v_1$ and $v_2$, which have period at most $k$ in their $9k$-length prefixes.
\end{proof}

This lemma implies the following corollaries.
\begin{corollary}[Marked regions are far apart]
\label{marked_regions_are_far_apart}
For all marked nodes $v_1$ and $v_2$ from different regions, their distance is strictly larger than $k$.
\end{corollary}

\begin{corollary}[Marked region is a vertical path]
\label{marked_regions_form_a_path}
Every marked region forms a vertical path.
\end{corollary}

We now prove that, if we consider the subgraph induced by unmarked nodes, the computed coloring $c$ forms a proper distance-$k$ coloring.
\begin{lemma}[Unmarked regions form a distance-k coloring]
\label{unmarked_regions_form_distance-k_coloring}
Let $C$ be a connected component of unmarked nodes. Then labels of nodes in $C$ form a distance-$k$ coloring.
\end{lemma}

\begin{proof}
By contradiction.
Suppose that $v_1$ and $v_2$ are two different unmarked nodes having the same color $c$ and distance $d<k$.
We will show that there is a marked node on the path between $v_1$ and $v_2$, contradicting that $v_1$ and $v_2$ are from the same connected component.
The distance to the lowest common ancestor (denoted by $v_3$) is at most $d$ for both nodes (the distance will be denoted by $d_1$ and $d_2$ for nodes $v_1$ and $v_2$ respectively).
W.l.o.g.\ we assume that $d_1 < d_2$ (because of symmetry and the fact that $d_1 = d_2$ would contradict that nodes $v_1$ and $v_2$ have the same color and are different).
As the path upwards from node $v_3$ is the same for both nodes, we obtain that $c[d_1 + i] = c[d_2 + i]$ for all $0\leq i<10k-d_2$. If we look at color $c'$ of the lowest common ancestor, we obtain from the previous equalities that $c'[i] = c'[d_2-d_1 + i]$ for all $0\leq i<10k-d_2$, hence the prefix of length $9k$ of $c'$ has period at most $d_2-d_1 \le k$. Hence node $v_3$ would be marked and on a path from $v_1$ to $v_2$ contradicting that they are from the same connected component.
\end{proof}

\subsection{Lower bound}

We now prove that, if a certificate for $O(1)$ solvability does not exist, then the problem requires $\Omega(\log^* n)$, even for randomized algorithms. On a high level, we can prove that if there is no $O(\log^* n)$ algorithm that can use the special configuration, then it means that we can convert any solution for $\Pi$ into a proper coloring, implying that $\Pi$ requires $\Omega(\log^* n)$.

\begin{theorem}\label{no1cert}
  Let $\Pi$ be an $\LCL$ problem for which no certificate for\/ $O(1)$ solvability exists.
  Then, the randomized and deterministic complexity of\/ $\Pi$ in the $\LOCAL$ model is $\Omega(\log^* n)$.
\end{theorem}

\begin{proof}
  We consider two possible cases: either there is a configuration of the form  $(a : b_1,\dots,a,\dots,b_\dd)$ or not. In the latter case, each solution for $\Pi$ is such that all nodes have a label that is different from the labels of the neighbors, meaning that we can interpret such a labeling as a coloring from a constant size palette. Since an algorithm for $O(1)$-coloring $\dd$-ary rooted trees could be simulated in directed paths, by imagining $\dd-1$ additional nodes connected to each node of the path, and since $O(1)$-coloring in paths is known to require $O(\log^*n)$ rounds, even for randomized algorithms~\cite{Linial92,Naor1991}, then the claim follows.

  Hence, assume that there is a configuration of the form $(a : b_1,\dots,a,\dots,b_\dd)$, but there does not exist a certificate for $O(\log^* n)$ solvability that contains all the labels of the special configuration, and that $a$ is the label of at least one leaf. We prove that for any algorithm $A$ that solves $\Pi$, there exists some $n_0$, such that for any $n>n_0$, algorithm $A$ running on any instance of size $n$ must label all nodes that are at $\omega(\log^* n)$ distance from the root and from any leaf such that they have a different label from all their neighbors. In other words, algorithm $A$ computes a proper coloring in
  the intermediate layers. 
  
  Assume it is not the case, then there must exist an algorithm $A$ such that, for any $n_0$, there exists some $n > n_0$ such that on some instances of size $n$ it labels at least one node that is at $\omega(\log^* n)$ distance from the root and from any leaf by some configuration of the form $(a : b_1,\dots,a,\dots,b_\dd)$. This implies that also each label in $\{b_1,\dots,a,\dots,b_\dd\}$ is used by at least one node that is at $\omega(\log^* n)$ distance from the root and from any leaf, and hence, the requirements of Lemma~\ref{lem:constructcert} apply. By applying Lemma~\ref{lem:constructcert} with some $n_0$ large enough guaranteed to exist by the running time of the algorithm, we get that there exists a certificate for $O(\log^*n)$ solvability that contains all labels in $\{b_1,\dots,a,\dots,b_\dd\}$ and that uses $a$ in at least one leaf, contradicting the fact that there is no certificate for $O(1)$ solvability.
  
  Hence, in any solution for $\Pi$  constructed by an algorithm running in $O(\log^* n)$ rounds all nodes that are at $\omega(\log^* n)$ distance from the root and from any leaf are labeled such that they have a different label from all their neighbors, and hence that all these nodes are properly colored, in any instance that is large enough. Hence we can use any $O(\log^* n)$ algorithm for $\Pi$ to solve $O(1)$-coloring in paths, by creating a virtual graph in which we connect large enough trees to each node of the path and extend the path on the endpoints, such that no node sees any root or leaf, and then running the algorithm. This implies that $\Pi$ requires $\Omega(\log^* n)$ rounds, and hence the claim follows.
\end{proof}

\subsection{Decidability}

The only additional requirement for a problem that is $O(\log^*n)$ solvable to be constant-time solvable is the existence of configuration $\aab$ where $a,b_i$ are from $\Sigma_\TT$, that are the labels used by the certificate $\TT$, and at least one leaf in $\TT$ is labeled $a$.

\autoref{alg:unrestricted_certificate} allows us to search for a certificate builder containing a specific leaf, and by \autoref{lem:cbtocert} a certificate builder of this form implies a certificate of the same form. Hence, we can just augment \autoref{alg:certificate} to additionally search only for a certificate builder that would satisfy having a configuration of the form $\aab$ consisting of certificate labels, such that $a$ appears in at least one leaf. This is done in \autoref{alg:constant_certificate}.

\begin{algorithm}[tbp]
\DontPrintSemicolon
\caption{\label{alg:constant_certificate}$\constantCertificate(\Pi)$}
\KwIn{$\LCL$ problem $\Pi$}
\KwOut{$\epsilon$ if no certificate exists, or a certificate builder}
\BlankLine

\For{all subsets $\Sigma'$ of ${\Sigma(\Pi)}$} {
  $\Pi' \gets$ restriction of $\Pi$ to labels from $\Sigma'$\;
  \For{all configurations of the form $\aab$ in $\Pi'$}{
    $\certificateBuilder \gets \findUnrestrictedCertificate(\Pi',a)$ \;
    \If{$\certificateBuilder \neq \epsilon$} {
      \Return $\certificateBuilder$\;
    }
    }
}
\Return $\epsilon$\;
\end{algorithm}

\begin{theorem}
\label{constant_certificate_search_runtime}
The running time of \autoref{alg:constant_certificate} is exponential in the size of the $\LCL$ problem.
\end{theorem}

\begin{proof}
Follows by using the same arguments as for \autoref{alg:certificate} (\autoref{certificate_search_runtime}).
\end{proof}

\begin{theorem}\label{thm:constantcertdecidability}
\autoref{alg:constant_certificate} outputs a certificate builder if and only if an $\LCL$ problem has a certificate for\/ $O(1)$ solvability.
\end{theorem}
\begin{proof}
  Since \autoref{alg:constant_certificate} tries to find a certificate builder for all subsets of labels for which a configuration of the form $\aab$ exists, then the statement follows by using the same arguments as in \autoref{thm:findlogstarcertifexists}.
\end{proof}

Hence we conclude the following theorem.

\begin{theorem}
\label{constant_certificate_search_runtime-main}
  Whether an $\LCL$ problem $\Pi$ has round complexity $O(1)$ or\/ $\Omega(\log^* n)$ can be decided in time at most exponential in the size of the $\LCL$ problem.
\end{theorem}

\section{Polynomial region}\label{sec:poly-region}

In this section, we describe an infinite sequence of $\LCL$ problems $\Pi_1, \Pi_2, \ldots$ with $\delta = 2$ such that the complexity of $\Pi_k=(2, \Sigma_k, \CC_k)$  is  $\Theta(n^{1/k})$  in both of the $\LOCAL$ and $\CONGEST$ models, as the lower bound applies to $\LOCAL$ and the upper bound applies to $\CONGEST$.

The alphabet $\Sigma_k$ for $\Pi_k$ is
\[\Sigma_k = \{a_1, b_1, x_1, a_2, b_2, x_2, \ldots, a_k, b_k\},\]
and the set of permitted configurations $\CC_k$ for $\Pi_k$ is defined as follows.
\begin{itemize}
    \item For $1 \leq i \leq k$, add  $(a_i : \sigma, \sigma')$ to $\CC_k$ for all $\sigma, \sigma' \in \{a_1, b_1, x_1, a_2, b_2, x_2, \ldots, a_{i-1}, b_{i-1}, x_{i-1}\} \cup \{b_i\}$.
    \item For $1 \leq i \leq k$, add  $(b_i : \sigma, \sigma')$ to $\CC_k$ for all $\sigma, \sigma' \in \{a_1, b_1, x_1, a_2, b_2, x_2, \ldots, a_{i-1}, b_{i-1}, x_{i-1}\} \cup \{a_i\}$.
    \item For $1 \leq i \leq k-1$, add  $(x_i : \sigma, \sigma')$ to $\CC_k$ for all $\sigma \in \Sigma_k$ and $\sigma' \in \{a_1, b_1, x_1, a_2, b_2, x_2, \ldots, a_{i}, b_{i}\}$.
\end{itemize}

When $k = 1$, $\Pi_1$ is exactly the proper $2$-coloring problem with the two colors $\Sigma_1 = \{a_1, b_1\}$. When $k=2$, $\Pi_2$ is a combination of two proper $2$-coloring problems with the color sets $\{a_1, b_1\}$ and $\{a_2, b_2\}$ via the special label $x_1$. Whenever a node $v$ is labeled $x_1$, it must have at least one child $u$ such that the entire subtree rooted at $u$ is properly 2-colored by $\{a_1, b_1\}$. For general $k$, $\Pi_k$ can be seen as a combination of $k$ proper 2-coloring problems. See \autoref{fig:poly} for an illustration.

\begin{figure}
  \figlayout
  \centering
  $\Pi_1$:\quad
  \includegraphics[align=t,vshift=1mm,page=21,scale=\figscale]{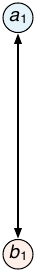}
  \hspace{2cm}
  $\Pi_2$:\quad
  \includegraphics[align=t,vshift=1mm,page=22,scale=\figscale]{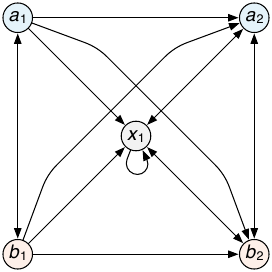}
	\caption{The automata associated with the path-forms of $\Pi_1$ and $\Pi_2$.}
	\label{fig:poly}
  \Description{}
\end{figure}

\begin{lemma}\label{lem:poly-upper-bound}
For each positive integer $k$, the round complexity of\/ $\Pi_k$ is $O(n^{1/k})$ in the $\CONGEST$ model.
\end{lemma}
\begin{proof}
Let $V$ be the node set for a given $n$-node rooted tree $T$.
We show that in $O(n^{1/k})$ rounds, we can partition the set of nodes $V$ into $2k-1$ parts \[V = B_1 \cup X_1 \cup B_2 \cup X_2 \cup \cdots \cup X_{k-1} \cup B_k\] satisfying the following properties.
\begin{description}
    \item[(P1)] For each $1 \leq i \leq k$, each connected component of $B_i$ has at most $O(n^{1/k})$ nodes.
    \item[(P2)] For each $1 \leq i \leq k-1$,  at least one child of each $v \in X_i$  is in $B_1 \cup X_1  \cup B_2 \cup X_2 \cup \cdots \cup B_{i-1} \cup X_{i-1} \cup B_i$.
    \item[(P3)] For each $1 \leq i \leq k$,  the children of each $v \in B_i$ are in $B_1 \cup X_1  \cup B_2 \cup X_2 \cup \cdots \cup X_{i-1} \cup B_i$.
\end{description}

Once we have this partition, $\Pi_k$ can be solved in $O(n^{1/k})$ rounds by assigning $x_i$ to each $v \in X_i$ and labeling each connected component of $B_i$ by an arbitrary proper 2-coloring with $\{a_i, b_i\}$.

The algorithm for computing the partition $V = B_1 \cup X_1 \cup B_2 \cup X_2 \cup \cdots \cup X_{k-1} \cup B_k$ has $k$ iterations. We assume that at the beginning of iteration $i$, we have already computed the parts $B_1, X_1, B_2, X_2, \ldots, B_{i-1}, X_{i-1}$ in such a way that the set of remaining nodes $U_i = V \setminus ( B_1 \cup X_1 \cup B_2 \cup X_2 \cup \cdots \cup B_{i-1} \cup X_{i-1})$ satisfies the following induction hypothesis.
\begin{description}
    \item[(IH)] For each $1 \leq i \leq k$, we have $|U_i| \leq n^{1 - (i-1)/k}$.
\end{description}
Note that (IH) holds initially for $i = 1$, as $U_1 = V$ and $|V| = n$.

Suppose we are at the beginning of iteration $i$.
Consider the subtree $T_i$ induced by the remaining nodes $U_i$. For each $v \in U_i$, we write $N_v$ to denote the number of nodes in the subtree of $T_i$ rooted at $v$.
We compute the two parts $B_i$ and $X_i$ as follows. It is clear that the computation takes $O(n^{1/k})$ rounds.
\begin{itemize}
    \item If $i = k$, then $B_i = U_i$.
    \item If $1 \leq i < k$, then $B_i$ is the set of nodes $v \in U_i$ with $N_v \leq n^{1/k}$.
    \item If $1 \leq i < k$, then $X_i$  is the set of nodes $v \in U_i$ with $N_v > n^{1/k}$ that satisfy at least one of the following.
    \begin{itemize}
        \item $N_u \leq n^{1/k}$ for at least one child $u$ of $v$.
        \item $v$ has exactly one child in $T_i$.
    \end{itemize}
\end{itemize}

It is straightforward to verify that  the properties (P1), (P2), and (P3) are satisfied for $B_i$ and $X_i$.
\begin{itemize}
    \item Consider the first property (P1). If $i=k$, then clearly $|B_k| = |U_k| \leq n^{1/k}$ by (IH). If $1 \leq i < k$, then each $v \in B_i$ can have at most $n^{1/k}$ descendants in $T_i$, including $v$ itself, so each connected component of $B_i$ has at most $n^{1/k}$ nodes.
    \item The second property (P2) follows from the definition of $X_i$. There are two cases for each $v \in X_i$.
    The first case is that there is a child $u$ of $v$ with $N_u \leq n^{1/k}$, so $u \in B_i$.
    The second case is that $v$ has exactly one child in $T_i$, so the other child $u$ of $v$ is in $V \setminus U_i = B_1 \cup X_1 \cup B_2 \cup X_2 \cup \cdots \cup B_{i-1} \cup X_{i-1}$.  In both cases, (P2) is satisfied.
    \item For the third property (P3), consider any child $u$ of a node $v \in B_i$. If $u \in U_i$, then the definition of $B_i$ ensures that $u \in B_i$. Otherwise $u \in B_1 \cup X_1 \cup B_2 \cup X_2 \cup \cdots \cup B_{i-1} \cup X_{i-1}$. In both cases, (P3) is satisfied.
\end{itemize}

For the rest of the proof, we consider the case $1 \leq i < k$ and we will show that $|U_{i} \setminus (B_i \cup X_i)| = |U_{i+1}| \leq n^{1 - i/k}$, so the induction hypothesis (IH) holds for $U_{i+1}$. Observe that each $v \in U_{i+1}$ must have exactly two children in $U_i$, so the set $W$ of nodes $v \in U_i \setminus U_{i+1}$ whose parent belongs to $U_{i+1}$ has size $|W| \ge |U_{i+1}|$. It is clear that $W \subseteq X_i$. Since each $v \in X_i$ has $N_v > n^{1/k}$, we can lower bound the size of $U_i$ by $|U_i| \geq \sum_{v \in W} N_v > |W|n^{1/k} \ge |U_{i+1}| n^{1/k}$, so  $|U_{i+1}| <  |U_i| n^{-1/k} \leq  n^{1 - i/k}$, as $|U_i| \leq n^{1 - (i-1)/k}$ by (IH) for $U_{i}$.
\end{proof}

\begin{lemma}\label{lem:poly-lower-bound}
For each positive integer $k$, the round complexity of\/ $\Pi_k$ is\/ $\Omega(n^{1/k})$ in the $\LOCAL$ model.
\end{lemma}
\begin{proof}
Observe that, given the $\LCL$ problem $\Pi_k$,
\autoref{alg:path_flexible_form_finder} takes exactly $k$ iterations to output $\epsilon$. For the first iteration, the labels $\{a_1, b_1\}$ are path-inflexible in $\Pi_k$. For iteration $1 < i \leq k$, $\{x_{i-1}, a_i, b_i\}$ are  path-inflexible in $\Pi_k$ restricted to the labels $\Sigma_k \setminus \{a_1, b_1, x_1, a_2, b_2, x_2, \ldots, a_{i-1}, b_{i-1}\}$. Therefore, the round complexity of\/ $\Pi_k$ is\/ $\Omega(n^{1/k})$ by \autoref{lem-no-path-flexible-lb}.
\end{proof}

Combining \autoref{lem:poly-upper-bound} and \autoref{lem:poly-lower-bound}, we conclude the  following theorem.

\begin{theorem}\label{lem:poly-main}
For each positive integer $k$, the round complexity of\/ $\Pi_k$ is\/ $\Theta(n^{1/k})$ in both $\LOCAL$ and $\CONGEST$.
\end{theorem}

\section{Future work}
While we completely characterize all complexities for $\LCL$s in rooted trees in both $\LOCAL$ and $\CONGEST$, for both deterministic and randomized algorithms, and we show that we can decide what is the complexity of a given problem, there are many questions that are left open.

The first question regards the running time of the algorithm that tries to find a certificate for $O(\log^* n)$ solvability. The current running time is exponential, and an open question is whether we can find such a certificate in polynomial time, or if we can prove that e.g.\ deciding the existence of a certificate is an NP-hard problem.

The second question regards the complexity class of $n^{\Theta(1)}$. While we present a practical algorithm that determines if the complexity is $n^{\Theta(1)}$, our algorithm does not determine the precise complexity. Whether there is an efficient algorithm for finding the precise value of $k$ such that the complexity is $\Theta(n^{1/k})$  remains open.

Another natural question regards extending our results to unrooted trees. While decidability is known in the $\Omega(\log n)$ region, it is known to require exponential time \cite{Chang20}. In our setting, we can decide if a problem is $O(\log n)$ or $n^{\Omega(1)}$ in polynomial time via a characterization based on the existence of a minimal absorbing subgraph. Coincidentally, Brandt~et~al.~\cite[Section 6]{brandt2021local} also recently proved that essentially the same characterization characterizes whether a problem is $O(\log n)$ or $n^{\Omega(1)}$ for regular unrooted trees. With some minor modification, our polynomial-time algorithm that decides if a problem is $O(\log n)$ or $n^{\Omega(1)}$ can also be adapted to the setting of regular unrooted trees. However, deciding if a problem on regular trees requires $O(1)$, $\Theta(\log^* n)$, or $\Omega(\log n)$ rounds remains a major open question.

\begin{acks}
We would like to thank Juho Hirvonen, Henrik Lievonen, Yannic Maus, and Mika\"el Rabie for discussions and comments, and anonymous reviewers for their helpful feedback on previous versions of this work. We also wish to acknowledge CSC -- IT Center for Science, Finland, for computational resources.
\end{acks}

\bibliographystyle{ACM-Reference-Format}
\bibliography{references}


\begin{thebibliography}{27}


\ifx \showCODEN    \undefined \def \showCODEN     #1{\unskip}     \fi
\ifx \showDOI      \undefined \def \showDOI       #1{#1}\fi
\ifx \showISBNx    \undefined \def \showISBNx     #1{\unskip}     \fi
\ifx \showISBNxiii \undefined \def \showISBNxiii  #1{\unskip}     \fi
\ifx \showISSN     \undefined \def \showISSN      #1{\unskip}     \fi
\ifx \showLCCN     \undefined \def \showLCCN      #1{\unskip}     \fi
\ifx \shownote     \undefined \def \shownote      #1{#1}          \fi
\ifx \showarticletitle \undefined \def \showarticletitle #1{#1}   \fi
\ifx \showURL      \undefined \def \showURL       {\relax}        \fi
\providecommand\bibfield[2]{#2}
\providecommand\bibinfo[2]{#2}
\providecommand\natexlab[1]{#1}
\providecommand\showeprint[2][]{arXiv:#2}

\bibitem[\protect\citeauthoryear{Balliu, Brandt, Chang, Olivetti, Rabie, and
  Suomela}{Balliu et~al\mbox{.}}{2019a}]%
        {balliu19lcl-decidability}
\bibfield{author}{\bibinfo{person}{Alkida Balliu}, \bibinfo{person}{Sebastian
  Brandt}, \bibinfo{person}{Yi-Jun Chang}, \bibinfo{person}{Dennis Olivetti},
  \bibinfo{person}{Mika{\"e}l Rabie}, {and} \bibinfo{person}{Jukka Suomela}.}
  \bibinfo{year}{2019}\natexlab{a}.
\newblock \showarticletitle{The distributed complexity of locally checkable
  problems on paths is decidable}. In \bibinfo{booktitle}{\emph{Proc.\ 38th ACM
  Symposium on Principles of Distributed Computing (PODC 2019)}}.
  \bibinfo{publisher}{ACM Press}, \bibinfo{pages}{262--271}.
\newblock
\urldef\tempurl%
\url{https://doi.org/10.1145/3293611.3331606}
\showDOI{\tempurl}
\showeprint{1811.01672}


\bibitem[\protect\citeauthoryear{Balliu, Brandt, Efron, Hirvonen, Maus,
  Olivetti, and Suomela}{Balliu et~al\mbox{.}}{2020a}]%
        {binary_lcls}
\bibfield{author}{\bibinfo{person}{Alkida Balliu}, \bibinfo{person}{Sebastian
  Brandt}, \bibinfo{person}{Yuval Efron}, \bibinfo{person}{Juho Hirvonen},
  \bibinfo{person}{Yannic Maus}, \bibinfo{person}{Dennis Olivetti}, {and}
  \bibinfo{person}{Jukka Suomela}.} \bibinfo{year}{2020}\natexlab{a}.
\newblock \showarticletitle{Classification of distributed binary labeling
  problems}. In \bibinfo{booktitle}{\emph{Proc.\ 34th International Symposium
  on Distributed Computing (DISC 2020)}} \emph{(\bibinfo{series}{LIPIcs},
  Vol.~\bibinfo{volume}{179})}. \bibinfo{publisher}{Schloss
  Dagstuhl--Leibniz-Zentrum f{\"u}r Informatik}, \bibinfo{pages}{17:1--17:17}.
\newblock
\urldef\tempurl%
\url{https://doi.org/10.4230/LIPIcs.DISC.2020.17}
\showDOI{\tempurl}
\showeprint{1911.13294}


\bibitem[\protect\citeauthoryear{Balliu, Brandt, Hirvonen, Olivetti, Rabie, and
  Suomela}{Balliu et~al\mbox{.}}{2021a}]%
        {BBHORS19MMlowerBound}
\bibfield{author}{\bibinfo{person}{Alkida Balliu}, \bibinfo{person}{Sebastian
  Brandt}, \bibinfo{person}{Juho Hirvonen}, \bibinfo{person}{Dennis Olivetti},
  \bibinfo{person}{Mika{\"e}l Rabie}, {and} \bibinfo{person}{Jukka Suomela}.}
  \bibinfo{year}{2021}\natexlab{a}.
\newblock \showarticletitle{Lower bounds for maximal matchings and maximal
  independent sets}.
\newblock \bibinfo{journal}{\emph{J. {ACM}}} \bibinfo{volume}{68},
  \bibinfo{number}{5} (\bibinfo{year}{2021}), \bibinfo{pages}{39:1--39:30}.
\newblock
\urldef\tempurl%
\url{https://doi.org/10.1145/3461458}
\showDOI{\tempurl}


\bibitem[\protect\citeauthoryear{Balliu, Brandt, Olivetti, and Suomela}{Balliu
  et~al\mbox{.}}{2020b}]%
        {BBOS18almostGlobal}
\bibfield{author}{\bibinfo{person}{Alkida Balliu}, \bibinfo{person}{Sebastian
  Brandt}, \bibinfo{person}{Dennis Olivetti}, {and} \bibinfo{person}{Jukka
  Suomela}.} \bibinfo{year}{2020}\natexlab{b}.
\newblock \showarticletitle{Almost global problems in the {LOCAL} model}.
\newblock \bibinfo{journal}{\emph{Distributed Computing}}.
\newblock
\urldef\tempurl%
\url{https://doi.org/10.1007/s00446-020-00375-2}
\showDOI{\tempurl}
\showeprint{1805.04776}


\bibitem[\protect\citeauthoryear{Balliu, Brandt, Olivetti, and Suomela}{Balliu
  et~al\mbox{.}}{2020c}]%
        {BBOS20paddedLCL}
\bibfield{author}{\bibinfo{person}{Alkida Balliu}, \bibinfo{person}{Sebastian
  Brandt}, \bibinfo{person}{Dennis Olivetti}, {and} \bibinfo{person}{Jukka
  Suomela}.} \bibinfo{year}{2020}\natexlab{c}.
\newblock \showarticletitle{How much does randomness help with locally
  checkable problems?}. In \bibinfo{booktitle}{\emph{Proc.\ 39th ACM Symposium
  on Principles of Distributed Computing (PODC 2020)}}. \bibinfo{publisher}{ACM
  Press}, \bibinfo{pages}{299--308}.
\newblock
\urldef\tempurl%
\url{https://doi.org/10.1145/3382734.3405715}
\showDOI{\tempurl}
\showeprint{1902.06803}


\bibitem[\protect\citeauthoryear{Balliu, Censor-Hillel, Maus, Olivetti, and
  Suomela}{Balliu et~al\mbox{.}}{2021b}]%
        {smallmessages}
\bibfield{author}{\bibinfo{person}{Alkida Balliu}, \bibinfo{person}{Keren
  Censor-Hillel}, \bibinfo{person}{Yannic Maus}, \bibinfo{person}{Dennis
  Olivetti}, {and} \bibinfo{person}{Jukka Suomela}.}
  \bibinfo{year}{2021}\natexlab{b}.
\newblock \showarticletitle{Locally Checkable Labelings with Small Messages}.
  In \bibinfo{booktitle}{\emph{Proc.\ 35th International Symposium on
  Distributed Computing (DISC 2021)}} \emph{(\bibinfo{series}{LIPIcs},
  Vol.~\bibinfo{volume}{209})}. \bibinfo{publisher}{Schloss Dagstuhl --
  Leibniz-Zentrum f{\"u}r Informatik}, \bibinfo{pages}{8:1--8:18}.
\newblock
\urldef\tempurl%
\url{https://doi.org/10.4230/LIPIcs.DISC.2021.8}
\showDOI{\tempurl}


\bibitem[\protect\citeauthoryear{Balliu, Hirvonen, Korhonen, Lempi{\"a}inen,
  Olivetti, and Suomela}{Balliu et~al\mbox{.}}{2018}]%
        {BHKLOS18lclComplexity}
\bibfield{author}{\bibinfo{person}{Alkida Balliu}, \bibinfo{person}{Juho
  Hirvonen}, \bibinfo{person}{Janne~H. Korhonen}, \bibinfo{person}{Tuomo
  Lempi{\"a}inen}, \bibinfo{person}{Dennis Olivetti}, {and}
  \bibinfo{person}{Jukka Suomela}.} \bibinfo{year}{2018}\natexlab{}.
\newblock \showarticletitle{New classes of distributed time complexity}. In
  \bibinfo{booktitle}{\emph{Proc.\ 50th ACM Symposium on Theory of Computing
  (STOC 2018)}}. \bibinfo{publisher}{ACM Press}, \bibinfo{pages}{1307--1318}.
\newblock
\urldef\tempurl%
\url{https://doi.org/10.1145/3188745.3188860}
\showDOI{\tempurl}
\showeprint{1711.01871}


\bibitem[\protect\citeauthoryear{Balliu, Hirvonen, Olivetti, and
  Suomela}{Balliu et~al\mbox{.}}{2019b}]%
        {BHOS19HomogeneousLCL}
\bibfield{author}{\bibinfo{person}{Alkida Balliu}, \bibinfo{person}{Juho
  Hirvonen}, \bibinfo{person}{Dennis Olivetti}, {and} \bibinfo{person}{Jukka
  Suomela}.} \bibinfo{year}{2019}\natexlab{b}.
\newblock \showarticletitle{Hardness of minimal symmetry breaking in
  distributed computing}. In \bibinfo{booktitle}{\emph{Proc.\ 38th ACM
  Symposium on Principles of Distributed Computing (PODC 2019)}}.
  \bibinfo{publisher}{ACM Press}, \bibinfo{pages}{369--378}.
\newblock
\urldef\tempurl%
\url{https://doi.org/10.1145/3293611.3331605}
\showDOI{\tempurl}
\showeprint{1811.01643}


\bibitem[\protect\citeauthoryear{Barenboim and Elkin}{Barenboim and
  Elkin}{2013}]%
        {barenboim13distributed}
\bibfield{author}{\bibinfo{person}{Leonid Barenboim} {and}
  \bibinfo{person}{Michael Elkin}.} \bibinfo{year}{2013}\natexlab{}.
\newblock \bibinfo{booktitle}{\emph{Distributed Graph Coloring: Fundamentals
  and Recent Developments}}.
\newblock \bibinfo{publisher}{Morgan \& Claypool}.
\newblock
\urldef\tempurl%
\url{https://doi.org/10.2200/S00520ED1V01Y201307DCT011}
\showDOI{\tempurl}


\bibitem[\protect\citeauthoryear{Brandt}{Brandt}{2019}]%
        {Brandt19RE}
\bibfield{author}{\bibinfo{person}{Sebastian Brandt}.}
  \bibinfo{year}{2019}\natexlab{}.
\newblock \showarticletitle{An Automatic Speedup Theorem for Distributed
  Problems}. In \bibinfo{booktitle}{\emph{Proc.\ 38th ACM Symposium on
  Principles of Distributed Computing (PODC 2019)}}.
  \bibinfo{publisher}{{ACM}}, \bibinfo{pages}{379--388}.
\newblock
\urldef\tempurl%
\url{https://doi.org/10.1145/3293611.3331611}
\showDOI{\tempurl}


\bibitem[\protect\citeauthoryear{Brandt, Chang, Greb{\'\i}k, Grunau,
  Rozho{\v{n}}, and Vidny{\'a}nszky}{Brandt et~al\mbox{.}}{2022}]%
        {brandt2021local}
\bibfield{author}{\bibinfo{person}{Sebastian Brandt}, \bibinfo{person}{Yi-Jun
  Chang}, \bibinfo{person}{Jan Greb{\'\i}k}, \bibinfo{person}{Christoph
  Grunau}, \bibinfo{person}{V{\'a}clav Rozho{\v{n}}}, {and}
  \bibinfo{person}{Zolt{\'a}n Vidny{\'a}nszky}.}
  \bibinfo{year}{2022}\natexlab{}.
\newblock \showarticletitle{Local Problems on Trees from the Perspectives of
  Distributed Algorithms, Finitary Factors, and Descriptive Combinatorics}. In
  \bibinfo{booktitle}{\emph{Proc.\ 13th Innovations in Theoretical Computer
  Science Conference (ITCS 2022)}} \emph{(\bibinfo{series}{(LIPIcs)},
  Vol.~\bibinfo{volume}{215})}. \bibinfo{publisher}{Schloss Dagstuhl --
  Leibniz-Zentrum f{\"u}r Informatik}, \bibinfo{pages}{29:1--29:26}.
\newblock
\urldef\tempurl%
\url{https://doi.org/10.4230/LIPIcs.ITCS.2022.29}
\showDOI{\tempurl}


\bibitem[\protect\citeauthoryear{Brandt, Fischer, Hirvonen, Keller,
  Lempi{\"a}inen, Rybicki, Suomela, and Uitto}{Brandt et~al\mbox{.}}{2016}]%
        {BFHKLRSU16}
\bibfield{author}{\bibinfo{person}{Sebastian Brandt}, \bibinfo{person}{Orr
  Fischer}, \bibinfo{person}{Juho Hirvonen}, \bibinfo{person}{Barbara Keller},
  \bibinfo{person}{Tuomo Lempi{\"a}inen}, \bibinfo{person}{Joel Rybicki},
  \bibinfo{person}{Jukka Suomela}, {and} \bibinfo{person}{Jara Uitto}.}
  \bibinfo{year}{2016}\natexlab{}.
\newblock \showarticletitle{A lower bound for the distributed {L}ov{\'a}sz
  local lemma}. In \bibinfo{booktitle}{\emph{Proc.\ 48th ACM Symposium on
  Theory of Computing (STOC 2016)}}. \bibinfo{publisher}{ACM Press},
  \bibinfo{pages}{479--488}.
\newblock
\urldef\tempurl%
\url{https://doi.org/10.1145/2897518.2897570}
\showDOI{\tempurl}
\showeprint{1511.00900}


\bibitem[\protect\citeauthoryear{Brandt, Hirvonen, Korhonen, Lempi{\"a}inen,
  {\"O}sterg{\aa}rd, Purcell, Rybicki, Suomela, and Uzna{\'n}ski}{Brandt
  et~al\mbox{.}}{2017}]%
        {Brandt2017}
\bibfield{author}{\bibinfo{person}{Sebastian Brandt}, \bibinfo{person}{Juho
  Hirvonen}, \bibinfo{person}{Janne~H. Korhonen}, \bibinfo{person}{Tuomo
  Lempi{\"a}inen}, \bibinfo{person}{Patric R.~J. {\"O}sterg{\aa}rd},
  \bibinfo{person}{Christopher Purcell}, \bibinfo{person}{Joel Rybicki},
  \bibinfo{person}{Jukka Suomela}, {and} \bibinfo{person}{Przemys{\l}aw
  Uzna{\'n}ski}.} \bibinfo{year}{2017}\natexlab{}.
\newblock \showarticletitle{{LCL} problems on grids}. In
  \bibinfo{booktitle}{\emph{Proc.\ 36th ACM Symposium on Principles of
  Distributed Computing (PODC 2017)}}. \bibinfo{publisher}{ACM Press},
  \bibinfo{pages}{101--110}.
\newblock
\urldef\tempurl%
\url{https://doi.org/10.1145/3087801.3087833}
\showDOI{\tempurl}
\showeprint{1702.05456}


\bibitem[\protect\citeauthoryear{Chang}{Chang}{2020}]%
        {Chang20}
\bibfield{author}{\bibinfo{person}{Yi{-}Jun Chang}.}
  \bibinfo{year}{2020}\natexlab{}.
\newblock \showarticletitle{The Complexity Landscape of Distributed Locally
  Checkable Problems on Trees}. In \bibinfo{booktitle}{\emph{Proc.\ 34th
  International Symposium on Distributed Computing (DISC 2020)}}
  \emph{(\bibinfo{series}{LIPIcs}, Vol.~\bibinfo{volume}{179})}.
  \bibinfo{publisher}{Schloss Dagstuhl--Leibniz-Zentrum f{\"u}r Informatik},
  \bibinfo{pages}{18:1--18:17}.
\newblock
\urldef\tempurl%
\url{https://doi.org/10.4230/LIPIcs.DISC.2020.18}
\showDOI{\tempurl}


\bibitem[\protect\citeauthoryear{Chang, Kopelowitz, and Pettie}{Chang
  et~al\mbox{.}}{2019}]%
        {CKP19exponential}
\bibfield{author}{\bibinfo{person}{Yi{-}Jun Chang}, \bibinfo{person}{Tsvi
  Kopelowitz}, {and} \bibinfo{person}{Seth Pettie}.}
  \bibinfo{year}{2019}\natexlab{}.
\newblock \showarticletitle{An Exponential Separation between Randomized and
  Deterministic Complexity in the {LOCAL} Model}.
\newblock \bibinfo{journal}{\emph{{SIAM} J. Comput.}} \bibinfo{volume}{48},
  \bibinfo{number}{1} (\bibinfo{year}{2019}), \bibinfo{pages}{122--143}.
\newblock
\urldef\tempurl%
\url{https://doi.org/10.1137/17M1117537}
\showDOI{\tempurl}


\bibitem[\protect\citeauthoryear{Chang and Pettie}{Chang and Pettie}{2019}]%
        {CP19timeHierarchy}
\bibfield{author}{\bibinfo{person}{Yi{-}Jun Chang} {and} \bibinfo{person}{Seth
  Pettie}.} \bibinfo{year}{2019}\natexlab{}.
\newblock \showarticletitle{A Time Hierarchy Theorem for the {LOCAL} Model}.
\newblock \bibinfo{journal}{\emph{{SIAM} J. Comput.}} \bibinfo{volume}{48},
  \bibinfo{number}{1} (\bibinfo{year}{2019}), \bibinfo{pages}{33--69}.
\newblock
\urldef\tempurl%
\url{https://doi.org/10.1137/17M1157957}
\showDOI{\tempurl}


\bibitem[\protect\citeauthoryear{Chang, Studen{\'y}, and Suomela}{Chang
  et~al\mbox{.}}{2021}]%
        {lcls_on_paths_and_cycles}
\bibfield{author}{\bibinfo{person}{Yi-Jun Chang}, \bibinfo{person}{Jan
  Studen{\'y}}, {and} \bibinfo{person}{Jukka Suomela}.}
  \bibinfo{year}{2021}\natexlab{}.
\newblock \showarticletitle{Distributed graph problems through an
  automata-theoretic lens}. In \bibinfo{booktitle}{\emph{Proc.\ 28th
  International Colloquium on Structural Information and Communication
  Complexity (SIROCCO 2021)}} \emph{(\bibinfo{series}{LNCS},
  Vol.~\bibinfo{volume}{12810})}. \bibinfo{publisher}{Springer},
  \bibinfo{pages}{31--49}.
\newblock
\urldef\tempurl%
\url{https://doi.org/10.1007/978-3-030-79527-6\_3}
\showDOI{\tempurl}
\showeprint{2002.07659}


\bibitem[\protect\citeauthoryear{Chung, Pettie, and Su}{Chung
  et~al\mbox{.}}{2017}]%
        {CPS17DistrLLL}
\bibfield{author}{\bibinfo{person}{Kai{-}Min Chung}, \bibinfo{person}{Seth
  Pettie}, {and} \bibinfo{person}{Hsin{-}Hao Su}.}
  \bibinfo{year}{2017}\natexlab{}.
\newblock \showarticletitle{Distributed algorithms for the Lov{\'{a}}sz local
  lemma and graph coloring}.
\newblock \bibinfo{journal}{\emph{Distributed Comput.}} \bibinfo{volume}{30},
  \bibinfo{number}{4} (\bibinfo{year}{2017}), \bibinfo{pages}{261--280}.
\newblock
\urldef\tempurl%
\url{https://doi.org/10.1007/s00446-016-0287-6}
\showDOI{\tempurl}


\bibitem[\protect\citeauthoryear{Cole and Vishkin}{Cole and Vishkin}{1986}]%
        {ColeVishkin86}
\bibfield{author}{\bibinfo{person}{Richard Cole} {and} \bibinfo{person}{Uzi
  Vishkin}.} \bibinfo{year}{1986}\natexlab{}.
\newblock \showarticletitle{Deterministic Coin Tossing with Applications to
  Optimal Parallel List Ranking}.
\newblock \bibinfo{journal}{\emph{Inf. Control.}} \bibinfo{volume}{70},
  \bibinfo{number}{1} (\bibinfo{year}{1986}), \bibinfo{pages}{32--53}.
\newblock
\urldef\tempurl%
\url{https://doi.org/10.1016/S0019-9958(86)80023-7}
\showDOI{\tempurl}


\bibitem[\protect\citeauthoryear{Fischer and Ghaffari}{Fischer and
  Ghaffari}{2017}]%
        {FischerGhaffari17LLL}
\bibfield{author}{\bibinfo{person}{Manuela Fischer} {and}
  \bibinfo{person}{Mohsen Ghaffari}.} \bibinfo{year}{2017}\natexlab{}.
\newblock \showarticletitle{Sublogarithmic Distributed Algorithms for
  Lov{\'{a}}sz Local Lemma, and the Complexity Hierarchy}. In
  \bibinfo{booktitle}{\emph{Proc.\ 31st International Symposium on Distributed
  Computing (DISC 2017)}} \emph{(\bibinfo{series}{LIPIcs},
  Vol.~\bibinfo{volume}{91})}. \bibinfo{publisher}{Schloss
  Dagstuhl--Leibniz-Zentrum f{\"u}r Informatik}, \bibinfo{pages}{18:1--18:16}.
\newblock
\urldef\tempurl%
\url{https://doi.org/10.4230/LIPIcs.DISC.2017.18}
\showDOI{\tempurl}


\bibitem[\protect\citeauthoryear{Linial}{Linial}{1992}]%
        {Linial92}
\bibfield{author}{\bibinfo{person}{Nathan Linial}.}
  \bibinfo{year}{1992}\natexlab{}.
\newblock \showarticletitle{Locality in Distributed Graph Algorithms}.
\newblock \bibinfo{journal}{\emph{{SIAM} J. Comput.}} \bibinfo{volume}{21},
  \bibinfo{number}{1} (\bibinfo{year}{1992}), \bibinfo{pages}{193--201}.
\newblock
\urldef\tempurl%
\url{https://doi.org/10.1137/0221015}
\showDOI{\tempurl}


\bibitem[\protect\citeauthoryear{Miller and Reif}{Miller and Reif}{1985}]%
        {Miller1985}
\bibfield{author}{\bibinfo{person}{Gary~L. Miller} {and}
  \bibinfo{person}{John~H. Reif}.} \bibinfo{year}{1985}\natexlab{}.
\newblock \showarticletitle{{Parallel tree contraction and its application}}.
  In \bibinfo{booktitle}{\emph{Proc.\ 26th Annual Symposium on Foundations of
  Computer Science (FOCS 1985)}}. \bibinfo{publisher}{IEEE},
  \bibinfo{pages}{478--489}.
\newblock
\urldef\tempurl%
\url{https://doi.org/10.1109/SFCS.1985.43}
\showDOI{\tempurl}


\bibitem[\protect\citeauthoryear{Naor}{Naor}{1991}]%
        {Naor1991}
\bibfield{author}{\bibinfo{person}{Moni Naor}.}
  \bibinfo{year}{1991}\natexlab{}.
\newblock \showarticletitle{A Lower Bound on Probabilistic Algorithms for
  Distributive Ring Coloring}.
\newblock \bibinfo{journal}{\emph{{SIAM} J. Discret. Math.}}
  \bibinfo{volume}{4}, \bibinfo{number}{3} (\bibinfo{year}{1991}),
  \bibinfo{pages}{409--412}.
\newblock
\urldef\tempurl%
\url{https://doi.org/10.1137/0404036}
\showDOI{\tempurl}


\bibitem[\protect\citeauthoryear{Naor and Stockmeyer}{Naor and
  Stockmeyer}{1995}]%
        {NaorStockmeyer95}
\bibfield{author}{\bibinfo{person}{Moni Naor} {and} \bibinfo{person}{Larry~J.
  Stockmeyer}.} \bibinfo{year}{1995}\natexlab{}.
\newblock \showarticletitle{What Can be Computed Locally?}
\newblock \bibinfo{journal}{\emph{{SIAM} J. Comput.}} \bibinfo{volume}{24},
  \bibinfo{number}{6} (\bibinfo{year}{1995}), \bibinfo{pages}{1259--1277}.
\newblock
\urldef\tempurl%
\url{https://doi.org/10.1137/S0097539793254571}
\showDOI{\tempurl}


\bibitem[\protect\citeauthoryear{Olivetti}{Olivetti}{2020}]%
        {Olivetti2019REtor}
\bibfield{author}{\bibinfo{person}{Dennis Olivetti}.}
  \bibinfo{year}{2020}\natexlab{}.
\newblock \bibinfo{title}{{Round Eliminator: a tool for automatic speedup
  simulation}}.
\newblock
\newblock
\urldef\tempurl%
\url{https://github.com/olidennis/round-eliminator}
\showURL{%
\tempurl}


\bibitem[\protect\citeauthoryear{Rozho{\v n} and Ghaffari}{Rozho{\v n} and
  Ghaffari}{2020}]%
        {RG20NetDecomposition}
\bibfield{author}{\bibinfo{person}{V{\'a}clav Rozho{\v n}} {and}
  \bibinfo{person}{Mohsen Ghaffari}.} \bibinfo{year}{2020}\natexlab{}.
\newblock \showarticletitle{Polylogarithmic-time deterministic network
  decomposition and distributed derandomization}. In
  \bibinfo{booktitle}{\emph{Proc.\ 52nd Annual ACM SIGACT Symposium on Theory
  of Computing (STOC 2020)}}. \bibinfo{publisher}{ACM},
  \bibinfo{pages}{350--363}.
\newblock
\urldef\tempurl%
\url{https://doi.org/10.1145/3357713.3384298}
\showDOI{\tempurl}


\bibitem[\protect\citeauthoryear{Studen{\'y} and Tereshchenko}{Studen{\'y} and
  Tereshchenko}{2021}]%
        {AnonymousRepo}
\bibfield{author}{\bibinfo{person}{Jan Studen{\'y}} {and}
  \bibinfo{person}{Aleksandr Tereshchenko}.} \bibinfo{year}{2021}\natexlab{}.
\newblock \bibinfo{title}{Rooted Tree Classifier}.
\newblock
\newblock
\urldef\tempurl%
\url{https://github.com/jendas1/rooted-tree-classifier}
\showURL{%
\tempurl}


\end{thebibliography}

\end{document}